\documentclass[11pt]{article}%
\usepackage{amsmath}
\usepackage{amsfonts}
\usepackage{amssymb}
\usepackage{graphicx}
\usepackage{boxedminipage}
\usepackage{hyperref}
\usepackage{fullpage}%
\providecommand{\U}[1]{\protect\rule{.1in}{.1in}}
\newtheorem{theorem}{Theorem}

\newtheorem{claim}[theorem]{Claim}

\newtheorem{corollary}[theorem]{Corollary}

\newtheorem{definition}[theorem]{Definition}

\newtheorem{lemma}[theorem]{Lemma}

\newtheorem{proposition}[theorem]{Proposition}

\newenvironment{proof}[1][Proof]{\noindent\textbf{#1.} }{\ \rule{0.5em}{0.5em}}
\begin{document}

\title{Gentle Measurement of Quantum States and Differential Privacy}
\author{Scott Aaronson\thanks{University of Texas at Austin. \ Email:
aaronson@utexas.edu. \ Supported by a Vannevar Bush Fellowship from the US
Department of Defense, a Simons Investigator Award, and the Simons
\textquotedblleft It from Qubit\textquotedblright\ collaboration.}
\and Guy N. Rothblum\thanks{Weizmann Institute of Science. \ Email:
rothblum@alum.mit.edu. \ Supported by ISF grant no. 5219/17.}}
\date{}
\maketitle

\begin{abstract}
In \textit{differential privacy (DP)}, we want to query a database about $n$
users, in a way that \textquotedblleft leaks at most $\varepsilon$ about any
individual user,\textquotedblright\ even conditioned on any outcome of the
query. \ Meanwhile, in \textit{gentle measurement}, we want to measure $n$
quantum states, in a way that \textquotedblleft damages the states by at most
$\alpha$,\textquotedblright\ even conditioned on any outcome of the
measurement. \ In both cases, we can achieve the goal by techniques like
deliberately adding noise to the outcome before returning it. \ This paper
proves a new and general connection between the two subjects. \ Specifically,
we show that on products of $n$ quantum states, any measurement that is
$\alpha$-gentle for small $\alpha$\ is also $O\left(  \alpha\right)  $-DP, and
any product measurement that is $\varepsilon$-DP is also $O\left(
\varepsilon\sqrt{n}\right)  $-gentle.

Illustrating the power of this connection, we apply it to the recently studied
problem of \textit{shadow tomography}. \ Given an unknown $d$-dimensional
quantum state $\rho$, as well as known two-outcome measurements $E_{1}%
,\ldots,E_{m}$, shadow tomography asks us to estimate $\Pr\left[  E_{i}\text{
accepts }\rho\right]  $, for \textit{every} $i\in\left[  m\right]  $, by
measuring few copies of $\rho$. \ Using our connection theorem, together with
a quantum analog of the so-called \textit{private multiplicative weights}
algorithm of Hardt and Rothblum, we give a protocol to solve this problem
using $O\left(  \left(  \log m\right)  ^{2}\left(  \log d\right)  ^{2}\right)
$\ copies of $\rho$, compared to Aaronson's previous bound of $\widetilde{O}%
\left(  \left(  \log m\right)  ^{4}\left(  \log d\right)  \right)  $. \ Our
protocol has the advantages of being \textit{online} (that is, the $E_{i}$'s
are processed one at a time), gentle, and conceptually simple.

Other applications of our connection include new \textit{lower} bounds for
shadow tomography from lower bounds on DP, and a result on the safe use of
estimation algorithms as subroutines inside larger quantum algorithms.

\end{abstract}
\tableofcontents

\newpage

\section{Introduction\label{INTRO}}

This paper is about a new mathematical connection between two
concepts---\textit{gentle measurement} in quantum mechanics, and
\textit{differential privacy} in classical computer science---and the
applications of this connection to the design of new quantum measurement
procedures and algorithms. \ Since the paper is meant to be accessible to
researchers in both fields (and beyond), we begin by saying a few words about
each of the concepts separately.

\subsection{Gentle Measurement\label{GENTLE}}

In quantum mechanics, \textit{measurement} is, famously, an inherently
destructive process. \ For example, if we measure a qubit $\alpha\left\vert
0\right\rangle +\beta\left\vert 1\right\rangle $ in the $\left\{  \left\vert
0\right\rangle ,\left\vert 1\right\rangle \right\}  $\ basis, we
\textquotedblleft force the qubit to decide\textquotedblright\ whether to be
$\left\vert 0\right\rangle $\ (with probability $\left\vert \alpha\right\vert
^{2}$) or $\left\vert 1\right\rangle $\ (with probability $\left\vert
\beta\right\vert ^{2}$). \ The qubit's state then \textquotedblleft
collapses\textquotedblright\ to whichever choice it made. \ There's no way to
measure again, unless of course we happen to have (or know how to prepare) a
second qubit in the same state.\footnote{This destructiveness is not unique to
quantum mechanics: it has a close analogue in classical Bayesian conditioning,
where a probability distribution can \textquotedblleft
collapse\textquotedblright\ to a single point when we make an observation.
\ But in classical probability, the \textquotedblleft
collapse\textquotedblright\ is purely internal and mental, in the sense that
we could undo it by simply forgetting the observation. \ In quantum mechanics,
by contrast, collapse causes an \textit{objective} change to the measured
system, one that could also be detected by someone else who later measured the
system.}

Even in quantum mechanics, though, not all measurements on all states are
destructive. \ For example, if a qubit happens to be in the $\left\vert
0\right\rangle $\ state\ already, then measuring in the $\left\{  \left\vert
0\right\rangle ,\left\vert 1\right\rangle \right\}  $\ basis\ causes no damage
at all. \ And if the qubit is the state $\left\vert \psi\right\rangle
=\sqrt{1-\varepsilon^{2}}\left\vert 0\right\rangle +\varepsilon\left\vert
1\right\rangle $ for small $\varepsilon$, then measuring in the $\left\{
\left\vert 0\right\rangle ,\left\vert 1\right\rangle \right\}  $\ basis causes
only minimal damage, since the result is almost always that the qubit
\textquotedblleft snaps\textquotedblright\ to $\left\vert 0\right\rangle
\approx\left\vert \psi\right\rangle $. \ More generally, the principle is this:

\begin{quotation}
\noindent A measurement $M$ applied to a state $\left\vert \psi\right\rangle
$\ necessarily severely damages $\left\vert \psi\right\rangle $ if, and only
if, the outcome of $M$ is highly unpredictable even to someone who already
knows $\left\vert \psi\right\rangle $.
\end{quotation}

This principle, which can be quantified in various ways, is called
\textit{information/disturbance tradeoff}: if $M$ creates lots of new (random)
information, then it must also cause lots of disturbance to $\left\vert
\psi\right\rangle $, and vice versa.

A corollary is that, if someone who knew $\left\vert \psi\right\rangle $ could
usually predict the measurement outcome in advance, then applying $M$\ need
not damage $\left\vert \psi\right\rangle $ by much. \ Note that this corollary
does not describe only trivial or uninteresting measurements, since in general
the measurer does \textit{not} know $\left\vert \psi\right\rangle $ in
advance---that's why she's measuring it!\footnote{And also, even if she
\textit{did} know a description of $\left\vert \psi\right\rangle $, she might
still find predicting the outcome of a measurement on $\left\vert
\psi\right\rangle $\ to be computationally intractable.}

Indeed, so-called \textit{gentle measurements}, which can be limited in how
much damage they cause, have found numerous applications in experimental
physics, the foundations of quantum mechanics, and quantum computing
theory.\footnote{Physicists more often refer to \textquotedblleft weak
measurement,\textquotedblright\ a related but not identical concept, which
typically means that the measurement returns very little information about the
state (in this paper, we'll call such measurements \textquotedblleft%
$\varepsilon$-trivial\textquotedblright). \ All weak measurements can be
implemented gently, and we'll show in Lemma \ref{mixtrick}\ that the only
measurements that are gentle on \textit{all} states are weak. \ But
measurements that are gentle on large sets of interesting states (such as
product states) can be far from weak, a point that will be crucial for us.}
\ Experimentalists, for example, know how to perform a measurement on a large
number of identically prepared particles, in a way that reveals the particles'
quantum state to high accuracy while causing very little damage.\footnote{With
a single particle, this is of course impossible.} \ More theoretically, gentle
measurement has also played a central role in proposals for
\textit{publicly-verifiable quantum money} that can be verified many times,
\textit{quantum software} that can be evaluated on many inputs, and so forth
(see \cite{aar:qcopy,aar:shadow}). \ Gentle measurement is also needed in work
on the nonabelian hidden subgroup problem \cite{ehk}, and on quantum advice
complexity classes like $\mathsf{BQP/qpoly}$ (see \cite{aar:adv,aar:qmaqpoly}).

Let's now define a bit more formally what we'll mean, in this paper, by a
quantum measurement being \textquotedblleft gentle.\textquotedblright

\begin{definition}
[Gentle Measurements]Given a set $S$ of quantum mixed states in some Hilbert
space, an implementation of a measurement $M$,\footnote{In this paper, by an
\textquotedblleft implementation\textquotedblright\ of a measurement $M$, we
mean a specification from which, given a state $\rho$, one can calculate not
only the probabilities of the various outcomes $y$, but \textit{also} the
post-measurement states $\rho_{M\rightarrow y}$.} and a parameter $\alpha
\in\left[  0,1\right]  $, we define $M$ to be $\mathbf{\alpha}$\textbf{-gentle
on} $S$\ if for all states $\rho\in S$, and all possible outcomes $y$\ of
applying $M$ to $\rho$, we have%
\begin{equation}
\left\Vert \rho_{M\rightarrow y}-\rho\right\Vert _{\operatorname*{tr}}%
\leq\alpha. \label{eq:alpha}%
\end{equation}
Here $\left\Vert \cdot\right\Vert _{\operatorname*{tr}}$\ represents trace
distance, the standard distance metric on quantum states, while $\rho
_{M\rightarrow y}$\ represents the new, \textquotedblleft
collapsed\textquotedblright\ state assuming that the measurement outcome was
$y$. \ (For a review of these and other quantum information concepts, see
Sections \ref{QI} and \ref{DIST}.)

More generally, we say $M$ is $\left(  \alpha,\delta\right)  $-gentle on $S$
if for all states $\rho\in S$, inequality (\ref{eq:alpha}) holds with
probability at least $1-\delta$\ over the possible outcomes $y$\ of applying
$M$\ to $\rho$. \ We recover $\alpha$-gentleness by setting $\delta=0$.

The most common choices for $S$ will be the set of product states $\rho
=\rho_{1}\otimes\cdots\otimes\rho_{n}$, and the set of all states.

If a measurement $M$ is specified by its output probabilities only
(technically, as a \textquotedblleft POVM\textquotedblright), then we say that
$M$ is $\alpha$-gentle\ if and only if there \textit{exists} an $\alpha
$-gentle implementation of it.
\end{definition}

As an example, suppose we have $n$ qubits in a pure product state:%
\[
\left\vert \psi\right\rangle =\left(  \alpha_{1}\left\vert 0\right\rangle
+\beta_{1}\left\vert 1\right\rangle \right)  \otimes\cdots\otimes\left(
\alpha_{n}\left\vert 0\right\rangle +\beta_{n}\left\vert 1\right\rangle
\right)  .
\]
Then consider the measurement $M$ that simply returns the total Hamming
weight. \ This measurement is \textit{not} $\alpha$-gentle for any nontrivial
$\alpha$. \ So for example, if we apply $M$ to the equal superposition
$\left(  \frac{\left\vert 0\right\rangle +\left\vert 1\right\rangle }{\sqrt
{2}}\right)  ^{\otimes n}$, we'll collapse the superposition over possible
Hamming weights---from a Gaussian wavepacket (as the physicists might call
it)\ of width $\Theta\left(  \sqrt{n}\right)  $ centered at $n/2$, all the way
down to a single random Hamming weight.

By contrast, now consider a measurement $L_{\sigma}$\ that returns the Hamming
weight, \textit{plus} a random noise term $\eta$\ of average magnitude
$\sigma\gg\sqrt{n}$. \ As an example, we could take this noise to follow a
Laplace distribution:%
\begin{equation}
\Pr\left[  \eta=k\right]  =\frac{1}{Z}e^{-\left\vert k\right\vert /\sigma},
\label{laplaceeq}%
\end{equation}
where%
\[
Z=\frac{2}{1-e^{-1/\sigma}}-1\approx2\sigma-1
\]
for large $\sigma$. \ We can implement the measurement $L_{\sigma}$\ as
follows. \ Given $\left\vert \psi\right\rangle $, which we now think of as a
superposition $\sum_{X\in\left\{  0,1\right\}  ^{n}}\alpha_{X}\left\vert
X\right\rangle $ over $n$-bit strings, first prepare alongside $\left\vert
\psi\right\rangle $ the state%
\[
\left\vert \eta\right\rangle :=\sum_{k=-\infty}^{\infty}\sqrt{\Pr\left[
\eta=k\right]  }\left\vert k\right\rangle .
\]
(In practice, we would of course impose a cutoff on $\left\vert k\right\vert
$.) \ Next, perform the unitary transformation%
\[
\sum_{X\in\left\{  0,1\right\}  ^{n}}\sum_{k=-\infty}^{\infty}\alpha_{X}%
\sqrt{\Pr\left[  \eta=k\right]  }\left\vert X\right\rangle \left\vert
k\right\rangle \rightarrow\sum_{X\in\left\{  0,1\right\}  ^{n}}\sum
_{k=-\infty}^{\infty}\alpha_{X}\sqrt{\Pr\left[  \eta=k\right]  }\left\vert
X\right\rangle \left\vert \left\vert X\right\vert +k\right\rangle .
\]
Finally, measure the $\left\vert \left\vert X\right\vert +k\right\rangle
$\ register in the standard basis and output the result.

It turns out that this noisy measurement $L_{\sigma}$ is $O\left(  \sqrt
{n}/\sigma\right)  $-gentle.\footnote{While there are other ways to prove that
$L_{\sigma}$ is $O\left(  \sqrt{n}/\sigma\right)  $-gentle, the nicest proof
we know will deduce it as an immediate corollary of this paper's main
results.} \ Intuitively, this is because the various Hamming weights that are
well-represented in the \textquotedblleft Gaussian
wavepacket\textquotedblright\ $\left\vert \psi\right\rangle $---e.g., in the
example $\left\vert \psi\right\rangle =\left(  \frac{\left\vert 0\right\rangle
+\left\vert 1\right\rangle }{\sqrt{2}}\right)  ^{\otimes n}$,\ those Hamming
weights $w$ such that $\left\vert w-\frac{n}{2}\right\vert =O\left(  \sqrt
{n}\right)  $---lead to probability distributions over measurement outcomes
that mostly overlap. \ In other words, when we observe an outcome of the form
$\left\vert X\right\vert +k$, the intrinsic variation in $\left\vert
X\right\vert $\ within the superposition is dominated by the variation in $k$.

\subsection{Differential Privacy\label{DP}}

Differential privacy (DP) is a young subfield of computer science---younger
than quantum computing, actually---that's seen tremendous growth since its
beginnings around 2006 \cite{dmns,dworkroth,vadhan:dp}. \ Though as we'll see,
DP's concepts turn out to have much broader applicability, the original
motivating problem is as follows. \ Suppose that a hospital (or bank, or
social media site) has a database of sensitive personal records. \ The
hospital wants to let medical researchers query its database in such a way that

\begin{enumerate}
\item[(1)] the researchers can learn as much accurate statistical information
as possible about the patient population (e.g., how many of them have colon
cancer), but

\item[(2)] each patient has a mathematical guarantee that, by opting to
participate in the database, she's exposing to the researchers
\textquotedblleft only a negligible amount\textquotedblright\ of data about
herself that would otherwise be private.
\end{enumerate}

\noindent The question is, how should we design the queries to balance these
two apparently conflicting demands?

More formally, call two databases $X,X^{\prime}$\ \textit{neighbors} if they
differ only in a single entry $x_{i}$. \ Then here is the key definition:

\begin{definition}
[Differential Privacy \cite{dmns}]\label{dpdef}Given a randomized algorithm
$A$ that queries a database $X$, as well as a parameter $\varepsilon\geq0$, we
define $A$ to be $\mathbf{\varepsilon}$\textbf{-DP} if for all databases
$X,X^{\prime}$\ that are neighbors, and all possible outputs $y$ of $A$, we
have%
\[
\Pr\left[  A\left(  X\right)  =y\right]  \leq e^{\varepsilon}\Pr\left[
A\left(  X^{\prime}\right)  =y\right]  .
\]
Here the probabilities are over the internal randomness used by $A$.
\end{definition}

In place of $e^{\varepsilon}$, one could also use the more intuitive
$1+\varepsilon$. \ However, the choice of $e^{\varepsilon}$\ has the
advantages that it composes nicely and is symmetric under inversions.

As an example---which should look familiar!---suppose the databases $X$ are
$n$-bit strings, and consider the algorithm that simply returns the Hamming
weight $\left\vert X\right\vert $. \ This algorithm is \textit{not}
$\varepsilon$-DP for any $\varepsilon$, since flipping just a single bit of
$X$ can change the probability of an output (namely, the new Hamming weight)
from $0$ to $1$. \ By contrast, now consider the algorithm $L_{\sigma}$\ that
returns the Hamming weight $\left\vert X\right\vert $, \textit{plus} a Laplace
noise term $\eta$\ that's distributed according to equation (\ref{laplaceeq}).
\ For any two neighboring databases $X,X^{\prime}$, and any possible output
$y$, we have%
\begin{equation}
\frac{\Pr\left[  L_{\sigma}\left(  X\right)  =y\right]  }{\Pr\left[
L_{\sigma}\left(  X^{\prime}\right)  =y\right]  }=\frac{e^{-\left\vert
y-\left\vert X\right\vert \right\vert /\sigma}}{e^{-\left\vert y-\left\vert
X^{\prime}\right\vert \right\vert /\sigma}}\leq e^{1/\sigma}. \label{lsigmaeq}%
\end{equation}
So we see that $L_{\sigma}$\ is $\frac{1}{\sigma}$-DP. \ Yet, as long as
$\sigma$\ is not too enormous, the output $\left\vert X\right\vert +\eta
$\ still provides a useful estimate of $\left\vert X\right\vert $.

Requiring multiplicative closeness in the probabilities of every output $y$
might seem overly strong. \ But if we relaxed the definition to an additive
one, we'd need to admit the algorithm that simply chooses a user $i\in\left[
n\right]  $\ uniformly at random and publishes all of her data. \ This
algorithm is manifestly not \textquotedblleft private,\textquotedblright\ and
yet it satisfies a strong additive guarantee: if user $i$ changes her data,
that will affect the probability distribution over outputs by at most
$\frac{1}{n}$\ in variation distance. \ On the other hand, one can check that
this algorithm is not $\varepsilon$-DP\ for any finite $\varepsilon$.

DP has been applied in deployed systems, for example at Apple and Google; see
for example \cite{tkbww} for discussion. \ The concept has also found
application to other problems, not obviously related to privacy---for example,
adaptive data analysis (for more see Section \ref{RELATED}). \ But what does
DP have to do with quantum information in general, or gentle measurement in particular?

\subsection{The Connection\label{CONNECTION}}

Given two quantum mixed states $\rho,\sigma$\ on $n$ registers each, call
them\ \textit{neighbors} if it's possible to reach either $\sigma$\ from
$\rho$, or $\rho$\ from $\sigma$,\ by performing a general quantum operation
(a so-called superoperator) on a single register only. \ In the special case
where $\rho=\rho_{1}\otimes\cdots\otimes\rho_{n}$\ and $\sigma=\sigma
_{1}\otimes\cdots\otimes\sigma_{n}$ are both product states, this reduces to
saying: $\rho$\ and $\sigma$\ are neighbors if and only if $\rho_{i}\neq
\sigma_{i}$ for at most one $i$.

Using this notion, we can easily generalize the definition of differential
privacy from Section \ref{DP}\ to the quantum setting:

\begin{definition}
[Quantum Differential Privacy]Given a set $S$\ of quantum mixed states each on
$n$ registers, a measurement $M$, and a parameter $\varepsilon\geq0$, we
define $M$ to be $\mathbf{\varepsilon}$\textbf{-DP on} $S$\ if for all states
$\rho,\sigma\in S$\ that are neighbors, and all possible outputs $y$ of $M$,
we have%
\begin{equation}
\Pr\left[  M\left(  \rho\right)  =y\right]  \leq e^{\varepsilon}\Pr\left[
M\left(  \sigma\right)  =y\right]  . \label{eq:dp}%
\end{equation}
Here the probabilities are over the intrinsic randomness of the measurement outcome.

More generally, we say $M$ is $\left(  \varepsilon,\delta\right)  $-DP on $S$
if for all neighboring states $\rho,\sigma\in S$, inequality (\ref{eq:dp})
holds with probability at least $1-\delta$\ over the possible outcomes $y$\ of
applying $M$\ to $\rho$. \ We recover $\varepsilon$-DP by setting $\delta
=0$.\footnote{This is a slightly nonstandard definition of $\left(
\varepsilon,\delta\right)  $-DP, but can be related to the standard definition
by a nontrivial result. \ See e.g.\ Vadhan \cite[Lemma 1.5]{vadhan:dp}.}

The most common choices for $S$ will be the set of product states $\rho
=\rho_{1}\otimes\cdots\otimes\rho_{n}$, and the set of all states.
\end{definition}

Note that unlike with gentleness, the property of being $\varepsilon$-DP
depends only on the output probabilities, and not at all on the
post-measurement states (i.e., on the \textquotedblleft
implementation\textquotedblright\ of the measurement).

Perhaps the first question we should ask is: \textit{are} there any nontrivial
quantum measurements that satisfy the above definition? \ Indeed there are.

Recall the DP algorithm $L_{\sigma}$ from Section \ref{DP}, which returns the
Hamming weight $\left\vert X\right\vert $\ of an $n$-bit input database $X$,
plus Laplace noise $\eta$\ of average magnitude $\sigma$. \ We can promote
$L_{\sigma}$\ to a quantum measurement on $n$-qubit states, by implementing it
using the procedure described in Section \ref{GENTLE}. \ We then have the following:

\begin{proposition}
\label{lsigma}$L_{\sigma}$ is $\frac{1}{\sigma}$-DP on the set of all
$n$-qubit states.
\end{proposition}

\begin{proof}
Since $L_{\sigma}$\ only involves measuring the Hamming weight in the
computational basis, for any $n$-qubit state $\rho$\ we can write%
\[
\Pr\left[  L_{\sigma}\left(  \rho\right)  =y\right]  =\sum_{X\in\left\{
0,1\right\}  ^{n}}\rho_{X,X}\Pr\left[  L_{\sigma}\left(  X\right)  =y\right]
.
\]
Also, if we act on a single register of $\rho$, and then measure in the
computational basis (which by the above, we can do without loss of
generality), we map each database $X$ to a distribution over neighbors
$X^{\prime}$\ of $X$. \ The proposition now follows from convexity together
with equation (\ref{lsigmaeq}).
\end{proof}

Stepping back, we've seen that simply measuring the Hamming weight\ of an
$n$-qubit state is neither gentle nor private. \ And yet the same
fix---namely, adding random noise to the measurement outcome before returning
it---makes the measurement both gentle \textit{and} private. \ Is this
convergence, between gentle quantum measurement and differential privacy, just
a coincidence?

Our main result asserts that it's not a coincidence: there's a strong two-way
connection between the two notions.

\begin{theorem}
[Main Result]\label{main}For all quantum measurements $M$ on $n$ registers:

\begin{enumerate}
\item[(1)] If $M$ is $\alpha$-gentle on product states for $\alpha\leq\frac
{1}{4.01}$, then $M$ is $O\left(  \alpha\right)  $-DP on product
states.\footnote{Indeed, it suffices for $\alpha$\ to be bounded below
$\frac{1}{4}$\ by any fixed constant (which then affects the multiplier in the
$O\left(  \alpha\right)  $). \ Similar remarks apply wherever constants like
$\frac{1}{4.01}$ appear in this paper.}

\item[(2)] If $M$ is $\varepsilon$-DP on product states, and is a product
measurement,\footnote{That is, if $M$ can be implemented by first applying a
classical algorithm to the outcomes of separate POVM measurements on the $n$
registers, and then uncomputing the outcomes of those $n$ measurements.} then
$M$ is $O\left(  \varepsilon\sqrt{n}\right)  $-gentle on product
states.\footnote{On non-product states, $M$ will still produce the correct
output probabilities, but it need not be gentle.}
\end{enumerate}
\end{theorem}

Again, here a \textquotedblleft measurement\textquotedblright\ $M$ corresponds
to a specification of output probabilities; for $M$ to be $\alpha$-gentle
means that there \textit{exists} an $\alpha$-gentle\ implementation of $M$.

Intuitively, it's far from obvious that gentleness and differential privacy
should be connected in this way. \ After all, the definition of $\alpha
$-gentleness makes no reference to the notion of \textquotedblleft
neighboring\textquotedblright\ states.\ \ Conversely, the definition of
$\varepsilon$-DP is exclusively concerned with output probabilities, and makes
no reference to post-measurement states. \ Our goal is to explain why
gentleness and DP \textit{are} connected in this way, and to explore the
consequences of the connection.

We'll see some applications of Theorem \ref{main}\ shortly, in Sections
\ref{APP} and \ref{APPST}. \ Before we do so, however, let's make a few
comments about the theorem statement.

At first glance, part (2) of the theorem seems weaker than part
(1)---especially because of the $\sqrt{n}$\ blowup in parameters---but it's
the part that carries many of the interesting implications. \ In Section
\ref{EXAMPLES}, we'll show that the $\sqrt{n}$\ blowup is unavoidable.
\ Indeed the measurement $L_{\sigma}$, with $\sigma=\Theta\left(  \sqrt
{n}\right)  $, already demonstrates this.

By contrast, the condition that $M$ is a product measurement\ is \textit{not}
clearly necessary; one of the central open problems we leave is whether that
condition can be removed. \ In Appendix \ref{NOTPRODMEAS}, we'll give examples
of quantum DP measurements that can't be approximated by any product or (we
conjecture) even LOCC measurements. \ However, all the examples we currently
know of such measurements are extremely artificial.

The restriction to product states might seem strange, but it's provably
unavoidable if we want Theorem \ref{main} to say anything about nontrivial
measurements. \ As we'll show in Section \ref{BASIC}, there \textit{is} a
counterpart of Theorem \ref{main} for states that could have arbitrary
correlation or entanglement among the registers. \ It turns out, however, that
if a measurement $M$ is $\alpha$-gentle on \textit{all} states for $\alpha
\ll\frac{1}{4}$, then $M$ is close to trivial (i.e., it barely depends on the
input state at all). \ And conversely, if $M$ is $\varepsilon$-DP, then the
best we can deduce is that $M$ is $O\left(  \varepsilon n\right)  $-gentle on
all states, rather than $O\left(  \varepsilon\sqrt{n}\right)  $-gentle.
\ While that might sound like a merely quantitative gap, the trouble is,
again, that the only measurements that are $\varepsilon$-DP\ for
$\varepsilon\ll\frac{1}{n}$\ are close to trivial. \ By contrast, plenty of
interesting measurements are $\varepsilon$-DP\ for $\varepsilon\ll\frac
{1}{\sqrt{n}}$.

One might wonder whether our reductions between privacy and gentleness
preserve computational efficiency. \ In one direction---turning gentleness
into privacy---the answer is clearly yes, since an $\alpha$-gentle measurement
\textit{is} $O\left(  \alpha\right)  $-DP; nothing further needs to be done.
\ However, in the other direction---turning privacy into gentleness---we can
implement a gentle measurement $M$ efficiently only if we have an efficient
algorithm to \textquotedblleft QSample\textquotedblright\ $M$'s output
distribution on a given input. \ QSampling is a term coined in 2003 by
Aharonov and Ta-Shma \cite{at}, which just means that we can efficiently
prepare a superposition over outputs of the form%
\[
\sum_{y}\sqrt{\Pr\left[  y\right]  }\left\vert y\right\rangle ,
\]
which is not entangled with any \textquotedblleft garbage\textquotedblright%
\ dependent on $y$. \ In practice, most DP algorithms that we know about
\textit{do} give rise to efficient QSampling procedures, but this property
doesn't follow automatically from a DP algorithm's being efficient. \ In
Section \ref{COMPUT}, we'll explore the issue of computational efficiency
further, and give nontrivial conditions under which efficiency is preserved.

\subsection{Applications\label{APP}}

Can we \textit{exploit} the connection between gentle measurement and
differential privacy to port results from one field to the other, as was done
with the connections between communication complexity and circuit lower
bounds, cryptography and learning, etc.? \ The second main contribution of
this paper is to use Theorem \ref{main}, together with previous work in DP, to
derive new results in quantum measurement theory and quantum
algorithms.\footnote{Some of these applications could also have been obtained
by \textquotedblleft brute force\textquotedblright\ (e.g., directly designing
and analyzing the desired gentle measurements), but the connection to DP will
both guide us to the correct statements, and enable the simplest proofs of
them that we know. Meanwhile, our applications to so-called \textit{shadow
tomography} of quantum states, described in Section \ref{APPST}, will make
essential use of sophisticated algorithms and lower bounds from the DP
literature.}

As a tiny warmup application, notice that $L_{\sigma}$, the Laplace
noise\ measurement from Section \ref{DP}, is a \textquotedblleft product
measurement,\textquotedblright\ in the sense that it can be implemented via an
algorithm that measures each register separately. \ And thus, by combining
part (2) of Theorem \ref{main} with Proposition \ref{lsigma}, we immediately
obtain the following:

\begin{corollary}
[Gentleness of Laplace Noise Measurement]\label{lsigmagentle}$L_{\sigma}$\ is
$O\left(  \sqrt{n}/\sigma\right)  $\ -gentle on product states.
\end{corollary}

As far as we know, proving Corollary \ref{lsigmagentle} directly would require
a laborious calculation.

Here is another application. \ In the early days of quantum computing, Bennett
et al.\ \cite{bbbv} observed that a quantum algorithm can safely invoke other
quantum algorithms for decision problems as subroutines inside of a
superposition---or in terms of complexity classes, that $\mathsf{BQP}%
^{\mathsf{BQP}}=\mathsf{BQP}$. \ The proof uses amplification, to push down
the subroutine's error probability, combined with \textit{uncomputing}, to
eliminate any \textquotedblleft garbage\textquotedblright\ that the subroutine
leaves entangled with its input. \ However, this straightforward uncomputing
strategy no longer works for subroutines whose purpose is to \textit{estimate
an expectation value to within} $\pm\varepsilon$ (say, the acceptance
probability of a quantum circuit).

In Section \ref{COMPUT}, we'll point out one simple solution to this problem:
namely, \textit{run the subroutine }$n^{O\left(  1\right)  }$\textit{\ times
in parallel, then estimate the desired expectation values by simulating gentle
measurements on the resulting states.} \ If we implement this idea using the
Laplace noise measurement $L_{\sigma}$, then Corollary \ref{lsigmagentle}
yields the following:

\begin{theorem}
\label{promisebqpthm}Without loss of generality, a $\mathsf{BQP}$\ algorithm
can at any point estimate $\Pr\left[  C\text{ accepts}\right]  $\ to within
$\pm\frac{1}{n^{O\left(  1\right)  }}$, on any superposition containing
descriptions of quantum circuits $C$, while maintaining the superposition's coherence.
\end{theorem}

While it's possible to prove Theorem \ref{promisebqpthm} \textquotedblleft
bare-handedly,\textquotedblright\ without knowing about the connection between
gentleness and DP,\ the point is that the floodgates are now open.\ \ Given a
quantum algorithm $P$, which is run as a subroutine inside a larger quantum
algorithm $Q$, there are many things that $Q$ might want to know about $P$'s
output behavior, beyond just additive estimates for specific probabilities.
\ Whatever the details, Theorem \ref{main} reduces the task to designing a
suitable efficient DP algorithm, or finding such an algorithm in the
literature. \ Gentleness then follows automatically.

\subsection{Shadow Tomography\label{APPST}}

In Section \ref{SHADOW}, we present our \textquotedblleft
flagship\textquotedblright\ application for the connection between gentleness
and DP: a new quantum measurement procedure, called \textit{Quantum Private
Multiplicative Weights (QPMW)}, which achieves parameters and properties that
weren't previously known.

QPMW addresses a task that Aaronson \cite{aarbados}, in 2016, called
\textit{shadow tomography}. \ Here we're given $n$\ copies of an unknown
$d$-dimensional mixed state $\rho$. \ We're also given known two-outcome
measurements $E_{1},\ldots,E_{m}$. \ Our goal is to learn $\Pr\left[
E_{i}\left(  \rho\right)  \text{ accepts}\right]  $\ to within an additive
error of $\pm\varepsilon$, for \textit{every} $i\in\left[  m\right]  $, with
high success probability (say, at least $2/3$), by carefully measuring the
$\rho$'s. \ Setting aside computational difficulty, how many copies of $\rho
$\ are information-theoretically necessary for this?

At one extreme of parameters, and suppressing the dependence on $\varepsilon$,
it's clear that $n=\widetilde{O}\left(  m\right)  $\ copies of $\rho
$\ suffice, since we could just apply each $E_{i}$\ to different copies. \ At
a different extreme, it's also clear that $n=\widetilde{O}\left(
d^{2}\right)  $\ copies suffice---or not \textquotedblleft
clear,\textquotedblright\ but it follows from celebrated recent work by
O'Donnell and Wright \cite{owright} and (independently) Haah et
al.\ \cite{hhjwy}, who showed that $\widetilde{O}\left(  d^{2}\right)
$\ copies of $\rho$\ are necessary and sufficient for full \textit{quantum
state tomography}: that is, reconstructing the entire state $\rho$\ to
suitable precision.

But what if we only want to learn the \textquotedblleft
shadow\textquotedblright\ that $\rho$\ casts on the measurements $E_{1}%
,\ldots,E_{m}$? \ Aaronson \cite{aarbados} raised the question of whether
shadow tomography might be possible using a number of copies $n$ that scales
only polylogarithmically in both $m$ and $d$---so in particular, that's
polynomial even if $m$ and $d$\ are both exponential. \ While this seemed
overly ambitious, Aaronson was unable to rule it out;\ and indeed, last year
he showed:

\begin{theorem}
[Aaronson \cite{aar:shadow}]\label{oldshadow}There exists an explicit
procedure to perform shadow tomography using%
\[
n=\widetilde{O}\left(  \frac{\left(  \log m\right)  ^{4}\left(  \log d\right)
}{\varepsilon^{4}}\right)
\]
copies of $\rho$. \ Here the $\widetilde{O}$\ hides factors of $\log\log m$,
$\log\log d$, and $\log\frac{1}{\varepsilon}$.
\end{theorem}

Shortly afterward, Brand\~{a}o et al.\ \cite{bkllsw} gave a different shadow
tomography procedure, based on semidefinite programming, which achieved the
same sample complexity as Aaronson's but was more efficient computationally.

However, these developments left several questions open:

\begin{enumerate}
\item[(1)] What is the true sample complexity of shadow tomography? \ The best
lower bound in \cite{aar:shadow}\ is that $\Omega\left(  \frac{\min\left\{
d^{2},\log m\right\}  }{\varepsilon^{2}}\right)  $\ copies are needed.

\item[(2)] The procedures of \cite{aar:shadow,bkllsw}\ destroy the copies of
$\rho$\ in the process of measuring them. \ Is there a shadow tomography
procedure that's also \textit{gentle}?

\item[(3)] The procedures of \cite{aar:shadow,bkllsw}\ require the full list
$E_{1},\ldots,E_{m}$\ to be known in advance. \ Is there a shadow tomography
procedure that's \textit{online}---i.e., that receives the measurements one by
one, and estimates each $\Pr\left[  E_{i}\text{ accepts }\rho\right]
$\ immediately\ after receiving $E_{i}$?
\end{enumerate}

In Section \ref{SHADOW}, by exploiting our connection between gentleness and
DP, and by quantizing a known classical DP algorithm called Private
Multiplicative Weights \cite{hardtrothblum}, we prove a new shadow tomography
theorem that addresses all of the above questions.

\begin{theorem}
[Quantum PMW]\label{newshadow}There exists an explicit procedure, Quantum
Private Multiplicative Weights (QPMW), that performs shadow tomography with
success probability $1-\beta$\ using%
\[
n=O\left(  \frac{\left(  \log^{2}m + \log\frac{1}{\delta}\right)  \cdot
\log^{2}d \cdot\log\frac{1}{\beta}}{\varepsilon^{8}}\right)
\]
copies of $\rho$,\ and which is also online and $\left(  \varepsilon
,\delta\right)  $-gentle.
\end{theorem}

Most notably, QPMW is both \textit{online} and \textit{gentle}; the previous
procedures \cite{aar:shadow,bkllsw} were neither. \ Because of its simplicity
and its online nature, QPMW seems better suited than its predecessors to
potential experimental realization.

Meanwhile, compared to Theorem \ref{oldshadow}, Theorem \ref{newshadow}
improves the dependence on $m$ from $\left(  \log m\right)  ^{4}$\ to $\left(
\log m\right)  ^{2}$. \ The dependence on $d$\ and$\ 1/\varepsilon$\ is worse,
but we conjecture that this is an artifact of our analysis, and that porting
so-called \textquotedblleft advanced composition\textquotedblright%
\ \cite{drv}\ to the quantum setting would ameliorate the situation. \ The
running time of QPMW is roughly $O\left(  mL\right)  +d^{O\left(  1\right)  }%
$, where $L$\ is the time needed to implement a single $E_{i}$; this improves
on the $O\left(  mL\right)  +d^{O\left(  \log\log d\right)  }$\ running time
of Aaronson's procedure, and matches an improvement from $d^{O\left(  \log\log
d\right)  }$ to $d^{O\left(  1\right)  }$\ in Brand\~{a}o et
al.\ \cite{bkllsw}.

It's hard to give a simple intuition for the improvement in $m$-dependence
from $O\left(  \log^{4}m\right)  $\ to $O\left(  \log^{2}m\right)  $.
\ Loosely, though, gentleness (derived from DP) lets QPMW be online, and being
online lets QPMW avoid the \textquotedblleft gentle search
procedure,\textquotedblright\ a key subroutine in Aaronson's earlier procedure
\cite{aar:shadow} that was responsible for the $\log^{4}m$\ factor. \ In any
case, we wish to stress that quantitative improvements in sample complexity
are not the main point here.\ \ The point, rather, is that the connection
between DP and gentleness leads to an entirely new approach to shadow tomography.

The DP/gentleness connection turns out to be useful not just for
\textit{upper} bounds on the sample complexity of shadow tomography, but also
lower bounds. \ In Section \ref{LBGENTLE}, we'll combine a recent lower bound
on DP algorithms \cite{buv}\ with \textit{part (1)} of Theorem \ref{main}%
\ (i.e., the gentleness implies DP direction), to deduce a new lower bound on
the sample complexity of gentle shadow tomography, where here
\textquotedblleft gentle\textquotedblright\ means \textquotedblleft gentle on
all product state inputs.\textquotedblright\ \ We'll also use recent work from
adaptive data analysis \cite{nsssu}\ to observe a lower bound on the sample
complexity of \textit{online} shadow tomography---showing that, for the latter
task, QPMW's sample complexity is optimal up to polynomial factors.

Finally, in Section \ref{COMPUTST}, we prove lower bounds on the
\textit{computational} complexity of gentle and online shadow tomography, by
deducing them as corollaries of recent lower cryptographic bounds for
differential privacy and adaptive data analysis
\cite{Ullman16,HardtU14,steinkeullman}. \ Assuming the existence of a one-way
function that takes $2^{\Omega\left(  n\right)  }$\ time to invert, these
lower bounds say that any algorithm for online or gentle shadow tomography
needs $d^{\Omega\left(  1\right)  }$\ time, so in that respect the QPMW
procedure is optimal for those tasks.

We stress that all our lower bounds for shadow tomography---both
information-theoretic and computational---are obtained by using this paper's
machinery to port known classical results to our setting. \ Thus, all of the
lower bounds apply equally well to the \textquotedblleft classical special
case\textquotedblright\ of shadow tomography, where we are trying to learn
properties of a probability distribution $\mathcal{D}$\ given independent
samples from $\mathcal{D}$, and none of them yet say anything specific to
quantum mechanics.

\subsection{Techniques\label{TECHNIQUES}}

\textbf{Relating Gentleness to DP.} \ In the proof of our main result---i.e.,
the connection between gentleness and differential privacy---the easy
direction is that gentleness implies DP. \ This direction produces only
constant loss in parameters, and does not even have much to do with quantum
mechanics. \ We consider the contrapositive: if a measurement $M$\ is
\textit{not} DP, then there are two neighboring states, call them $\rho$\ and
$\sigma$, as well as a measurement outcome $y$, such that $\Pr\left[  M\left(
\rho\right)  \text{ outputs }y\right]  $\ and $\Pr\left[  M\left(
\sigma\right)  \text{ outputs }y\right]  $ differ by a large multiplicative
factor. \ But in that case, we can study what happens if we apply $M$ to the
equal mixture $\frac{\rho+\sigma}{2}$, and then condition on outcome $y$.
\ Here we can use Bayes' theorem to show that the post-measurement state will
\textit{not} be close to $\frac{\rho+\sigma}{2}$---intuitively, because it
will have \textquotedblleft more $\rho$\ than $\sigma$\textquotedblright\ or
vice versa. \ Therefore $M$ is not gentle on product states (for if $\rho
$\ and $\sigma$\ are neighbors and are themselves product states, then
$\frac{\rho+\sigma}{2}$\ is a product state).

The harder direction is to show that $\varepsilon$-DP implies $O\left(
\varepsilon\sqrt{n}\right)  $-gentleness (for product states, and at least for
a restricted class of measurements). \ We work up to this result in a sequence
of steps. \ The first is to prove a purely classical analogue: namely, any
$\varepsilon$-DP classical algorithm is $2\varepsilon\sqrt{n}$-gentle\ on
product \textit{distributions}, $\mathcal{D}=\mathcal{D}_{1}\times\cdots
\times\mathcal{D}_{n}$---and indeed, the posterior distribution $\mathcal{D}%
_{y}$, conditioned on some output $y$, has KL-divergence at most
$2\varepsilon^{2}n$\ from $\mathcal{D}$. \ While this step has echoes in
earlier work on adaptive data analysis
\cite{dwork:adaptive,dwork:adaptive2,rrst}\ (see Section \ref{RELATED}), we
provide our own proof for completeness. \ Our proof uses the $\varepsilon$-DP
property of $A$, together with the fact that $\mathcal{D}$\ is a product
distribution, to show that, if we reveal a sample from $\mathcal{D}_{y}$\ a
single register at a time, from the $1^{st}$ to the $n^{th}$, then the
expected KL-divergence from $\mathcal{D}$\ increases by at most $2\varepsilon
^{2}$\ per register, and is therefore at most $2\varepsilon^{2}n$\ overall.

The second step is to prove an analogous result if the classical algorithm $A$
is applied, not to a sample from the distribution $\mathcal{D}$, but in
superposition to each component of the \textquotedblleft
QSampling\textquotedblright\ state%
\[
\left\vert \psi\right\rangle :=\sum_{x}\sqrt{\Pr_{\mathcal{D}}\left[
x\right]  }\left\vert x\right\rangle .
\]
To prove this, we let $\left\vert \psi_{y}\right\rangle $\ be the
post-measurement state conditioned on outcome $y$, and then upper-bound the
trace distance,%
\[
\left\Vert |\psi\rangle\langle\psi|-|\psi_{y}\rangle\langle\psi_{y}%
|\right\Vert _{\operatorname*{tr}}=\sqrt{1-\left\vert \left\langle \psi
|\psi_{y}\right\rangle \right\vert ^{2}},
\]
in terms of the square root of the KL-divergence between $\mathcal{D}_{y}$ and
$\mathcal{D}$, which we previously showed was $O\left(  \varepsilon\sqrt
{n}\right)  $. \ (To do that, in turn, we use the \textit{Hellinger distance}
between $\mathcal{D}_{y}$ and $\mathcal{D}$\ as an intermediate measure.)

The last step is to generalize from algorithms $A$ that act separately on each
computational basis state to measurements $M$ that can apply a separate POVM
to each register, and also from pure product states to mixed product states.
\ We achieve these generalizations using standard manipulations in quantum
information. \ We expect that further generalizations are possible with more
work.\bigskip

\textbf{Shadow Tomography.} \ The analysis of Quantum Private Multiplicative
Weights\textit{ }(QPMW), our new online, gentle procedure for shadow
tomography, is our technically most demanding result. \ The QPMW procedure
itself is relatively simple,\footnote{Indeed, QPMW is arguably simpler than
previous shadow tomography procedures, especially because it completely avoids
the use of the so-called Quantum OR Bound of Harrow, Lin, and Montanaro
\cite{hlm}. \ QPMW could, in fact, be used to give an independent proof of the
OR Bound, one where the procedure would moreover be gentle (albeit, possibly
with worse sample complexity).} and is directly inspired by an analogous
procedure from classical differential privacy, the so-called Private
Multiplicative Weights (PMW) algorithm of Hardt and Rothblum
\cite{hardtrothblum} from 2010.

Given a database $X\in\left[  d\right]  ^{n}$, of $n$\ records $x_{1}%
,\ldots,x_{n}$\ drawn independently from some underlying probability
distribution $\mathcal{D}$, the goal of PMW is to answer an enormous number of
statistical queries about $\mathcal{D}$, possibly as many as $\exp\left(
n\right)  $\ of them, in a way that preserves the overall differential privacy
of $X$. \ Here the queries need to be answered one by one, as they arrive, and
could be chosen by an adaptive adversary.

PMW achieves this by maintaining, at all times, a current hypothesis
$\mathcal{H}$\ about $\mathcal{D}$. \ Whenever a new query arrives, the first
thing PMW does is to check whether $\mathcal{H}$\ and $X$ lead to
approximately the same answer for that query. \ If the answers are equal to
within some threshold, then PMW simply answers the query using $\mathcal{H}$,
without looking further at $X$. \ Only if $\mathcal{H}$\ and $X$\ disagree
substantially does PMW query $X$\ a second time---both to learn the correct
answer to the current query, and to use that answer to update the hypothesis
$\mathcal{H}$. \ For \textit{both} types of queries, PMW uses the standard DP
trick of adding a small amount of Laplace noise to any statistics gathered
from $X$, before using those statistics for anything else.

It's clear, by construction, that this strange two-pronged approach will
always return approximately correct answers, with high probability. \ But why
does it help in preserving privacy? \ The privacy analysis depends on proving
three facts:

\begin{enumerate}
\item[(1)] Each query leads to only a negligible loss in privacy (say,
$\sim1/\exp\left(  n\right)  $), \textit{unless} it has an appreciably large
probability of triggering an update.

\item[(2)] Even when an update \textit{is} triggered, the loss in privacy is
still modest, say $\sim1/\sqrt{n}$.

\item[(3)] The number of updates is always extremely small, say $O\left(  \log
d\right)  $. \ This is true for \textquotedblleft the usual multiplicative
weights reasons.\textquotedblright
\end{enumerate}

Once one understands the connection between privacy and gentleness, it's
natural to wonder whether a quantum analogue of PMW might let one apply a huge
sequence of measurements $E_{1}\ldots,E_{m}$, one at a time, to a small
collection of identical quantum states $\rho^{\otimes n}$ (where, say,
$n\leq\left(  \log m\right)  ^{O\left(  1\right)  }$), in a way that yields
accurate estimates of $\Pr\left[  E_{i}\left(  \rho\right)  \text{
accepts}\right]  $\ for every $i$, without destroying the states in the
process or even damaging them too much. \ This, of course, is precisely the
problem of (gentle, online) shadow tomography.

In Section \ref{SHADOW}, we prove that indeed this is possible. \ Our QPMW
algorithm is just the \textquotedblleft obvious\textquotedblright\ quantum
generalization of PMW. \ That is, QPMW at all times maintains a current
hypothesis, $\sigma$, about the unknown quantum state $\rho$. \ Initially
$\sigma$\ is the maximally mixed state $\mathbb{I}/d$.\ \ Whenever a new
measurement $E_{t}$\ arrives, QPMW first checks whether%
\[
\Pr\left[  E_{t}\left(  \sigma\right)  \text{ accepts}\right]  \approx
\Pr\left[  E_{t}\left(  \rho\right)  \text{ accepts}\right]  ,
\]
with the check being done using a thresholded version of the Laplace noise
measurement\ from Section \ref{GENTLE}. \ If the answer is yes, then QPMW
simply returns $\Pr\left[  E_{t}\left(  \sigma\right)  \text{ accepts}\right]
$ as its estimate for $\Pr\left[  E_{t}\left(  \rho\right)  \text{
accepts}\right]  $, without measuring the actual quantum states any further.
\ Only if the answer is no does QPMW measure a second time---both to learn an
accurate estimate for $\Pr\left[  E_{t}\left(  \rho\right)  \text{
accepts}\right]  $, and to use that estimate to update its hypothesis $\sigma
$. \ This second measurement also involves the deliberate addition of Laplace noise.

Intuitively, the reason why we might expect this to work is that each round of
PMW leaks very little privacy---and by our central connection between DP and
gentleness, that suggests that we can implement each round of QPMW in a way
that damages the states very little. \ However, formalizing this requires
substantial new ideas, which are not contained in the classical analyses of PMW.

Of course, if we had a sufficiently general theorem about privacy implying
gentleness, then perhaps everything we needed would follow immediately from
that theorem, combined with the privacy of PMW. \ However, our existing
implication---applying, as it does, only to product measurements on product
states, and saying nothing about adaptively chosen \textit{sequences} of
measurements---will force us to work harder.

The core difficulty concerns what, before, we called step (1) in the analysis
of PMW: namely, the connection between loss in privacy and the probability of
triggering an update. \ We note that while, by construction, the answer in
each round is close to the answer on the current state in the registers, we
need the answer to be accurate with respect to the \emph{original} state
$\rho$. \ The algorithm's gentleness plays an essential role in proving
accuracy: it's only because of gentleness that we know that the state in the
registers hasn't been corrupted, and that the algorithm's answers are accurate
with respect to the original state. \ We further note that, since we want to
handle many measurements, and the damage from these measurements will
accumulate, we truly need to show that the overwhelming majority of
measurements result in only negligible damage.

The original analysis \cite{hardtrothblum} conditioned on so-called
\textquotedblleft borderline rounds,\textquotedblright\ which are rounds that
have a reasonable probability of triggering an update, and argued that the
privacy loss in other rounds was zero. \ In the quantum setting, however, this
is a non-starter: so long as there is some probability of an update, the
damage is never zero. \ Instead, we show how to bound the damage each
no-update round would cause to the original state as a function of the
probability that it could have triggered an update. \ Thus, rounds that are
likely to trigger an update (of which there are few) can cause damage, but
rounds that are unlikely to trigger an update (of which there are many) each
cause very little damage once we condition on \textquotedblleft no
update.\textquotedblright\ \ Since the number of updates is bounded, this is a
promising start. \ Bounding the damage as a function of the probability of an
update requires a delicate analysis, leveraging the differential privacy of
the Laplace measurement and the fact that we have a product state in the
registers, which induces a Gaussian distribution on the answers before noise
is added to each measurement (see Claim \ref{claim:noupdate_gentle}).

In the classical setting, once we bound the privacy loss per round, we can
apply composition theorems to bound the loss across rounds. \ Crucially, this
composition maintains \textit{multiplicative} guarantees on the closeness of
probabilities. \ But damage to quantum states (in the trace distance metric,
for example) is additive, not multiplicative. \ Indeed, even if the amplitudes
in a quantum state $\left\vert \psi\right\rangle $\ were to change by only
small multiplicative amounts, that could easily turn into an additive change
when we rotate $\left\vert \psi\right\rangle $\ to a different basis---a
phenomenon with no classical analog. \ So once $\left\vert \psi\right\rangle
$\ becomes even slightly corrupted, why doesn't this sever the multiplicative
connection between damage and the probability of an update---thereby
preventing the necessary updates from happening, and allowing $\left\vert
\psi\right\rangle $\ to become corrupted even further, and so forth, until
inaccurate answers are returned?

We address these worries using several tools. \ The first is a
\textquotedblleft Damage Lemma,\textquotedblright\ Lemma \ref{facepalmlemma},
which tightly connects the probability of an update being triggered in the
\textquotedblleft real\textquotedblright\ world, where the state
$\rho^{\otimes n}$\ is slightly damaged by each measurement round, to its
probability of being triggered in the \textquotedblleft
ideal\textquotedblright\ world, where the algorithm gets a fresh copy of
$\rho^{\otimes n}$\ at each round. \ This lemma is quite general and might
find other uses. \ With this lemma in place, we divide the execution of the
QPMW algorithm into epochs, where each epoch has a constant probability of
triggering an update. \ By the connection between damage and update
probabilities, this means that the sum of the damage incurred by an
\textquotedblleft ideal\textquotedblright\ execution would be bounded, and by
the Damage Lemma the total damage in the \textquotedblleft
real\textquotedblright\ execution remains bounded as well. \ Since, moreover,
each epoch triggers an update with constant probability, and the number of
updates is bounded, the number of epochs will be bounded too. \ This gives us
a bound on the total damage to the state, and is crucial both for proving
gentleness \textit{and} for proving accuracy.\bigskip

\textbf{Other Results.} \ The paper's other results are proved using a variety
of techniques. \ In Appendix \ref{SEPVSENT}, for example, we show that any
measurement that's $0$-DP on product states (i.e., accepts all product states
with the same probability) is actually $0$-DP on \textit{all} states, and
hence trivial. \ Though simple, this result makes essential use of the fact
that the separable mixed states have positive density within the set of all
mixed states, and would be \textit{false} if amplitudes were reals rather than
complex numbers. \ Since most results in quantum information are insensitive
to the distinction between real and complex quantum mechanics, it's noteworthy
to find an exception.

To prove, in Appendix \ref{CANTCOMPOSE}, a weak form of composition for
quantum DP algorithms, we use the same \textquotedblleft Damage
Lemma\textquotedblright\ (Lemma \ref{facepalmlemma}) that we used for the
analysis of QPMW. \ In that appendix, however, we also construct an example,
involving DP measurements in two incompatible bases, that shows why any
composition theorem for quantum DP will come with caveats that weren't needed classically.

To prove, in Section \ref{EXAMPLES}, that our \textquotedblleft DP implies
gentleness\textquotedblright\ implications are asymptotically optimal, we use
the Laplace noise measurement $L_{\sigma}$\ as a separating example. \ When
$\sigma=\Theta\left(  n\right)  $, we get a measurement that's $O\left(
\frac{1}{n}\right)  $-DP,\ but not $\alpha$-gentle on arbitrary states for any
$\alpha=o\left(  1\right)  $. \ When $\sigma=\Theta\left(  \sqrt{n}\right)  $,
we get a measurement that's $O(\frac{1}{\sqrt{n}})$-DP, but not $\alpha
$-gentle on product states for any $\alpha=o\left(  1\right)  $.

\subsection{Related Work\label{RELATED}}

To our knowledge, this paper is the first to make the connection between
gentle measurement of quantum states and differential privacy. \ Nevertheless,
there were two previous papers that tried to combine quantum information and
differential privacy in other ways; there was a previous study of gentle
tomography;\ and there was a celebrated (purely classical) connection between
differential privacy and so-called \textit{adaptive data analysis}, which in
some ways foreshadowed our connection between DP and gentle
measurement.\bigskip

\textbf{Quantum information and DP.} \ Senekane et al.\ \cite{senekane}
discuss first applying a classical DP algorithm to classical data, and then
encoding the output into a quantum state for use in a quantum machine learning
algorithm. \ Naturally this composition preserves DP, but the DP and quantum
aspects don't seem to interact much.

Zhou and Ying \cite{zhouying}\ define and study an interesting notion of
\textquotedblleft quantum DP,\textquotedblright\ which however is very
different from ours. \ Given an algorithm $A$ that takes a quantum state as
input and produces another quantum state as output, they define $A$ to be
$\left(  \Delta,\varepsilon,\delta\right)  $-DP if for all states $\rho
,\sigma$ with trace distance at most $\Delta$, and all $2$-outcome
measurements $M$,%
\[
\Pr\left[  M\left(  A\left(  \rho\right)  \right)  \text{ accepts}\right]
\leq e^{\varepsilon}\Pr\left[  M\left(  A\left(  \sigma\right)  \right)
\text{ accepts}\right]  +\delta.
\]
In other words: unlike us, Zhou and Ying don't consider $A$'s behavior on two
databases that differ in a single entry (but which could have arbitrarily
large trace distance)---only on two states that are actually close \textit{as
quantum states}. \ For them, essentially, a DP algorithm is a quantum channel
that converts \textquotedblleft mere\textquotedblright\ closeness in trace
distance into a stronger, multiplicative kind of closeness between quantum
states. \ Zhou and Ying's main results are that

\begin{enumerate}
\item[(1)] the standard depolarizing and amplitude-damping channels (i.e.,
just adding noise to a quantum state, like in the simplest models of
decoherence) are DP in their sense, and

\item[(2)] their notion of quantum DP satisfies many composition theorems,
including advanced composition.
\end{enumerate}

These results are interesting and non-obvious, but only tangentially related
to what we do.\bigskip

\textbf{Gentle tomography.} \ Bennett, Harrow, and Lloyd \cite{bhl} studied
the task of \textquotedblleft gentle quantum state
tomography\textquotedblright---that is, recovering a full description of a
quantum state $\rho$ from identical copies $\rho^{\otimes n}$,\ without
appreciably damaging the $\rho^{\otimes n}$'s. \ Their notion of
\textquotedblleft gentleness\textquotedblright\ was very similar to ours. \ To
achieve the task, they gave a protocol that, like many of our protocols,
deliberately adds noise to the measurement outcomes before returning them
(although they used a randomized binning strategy rather than Laplace noise).
\ They did not make a connection to differential privacy, and also did not
consider shadow tomography, or any other tasks besides full
tomography.\bigskip

\textbf{DP and adaptive data analysis.} \ Perhaps the work that most clearly
anticipated ours, at a technical level, had nothing to do with quantum
information at all. \ Dwork et al.\ \cite{dwork:adaptive} studied the problem
of \textit{adaptive data analysis}: given a dataset, drawn i.i.d.\ from an
underlying distribution, the goal is accurately to answer a long sequence of
adaptively chosen statistical queries or analyses. \ Each query can be chosen
as a function of the answers to all previous queries. \ Accuracy is measured
with respect to the underlying distribution, rather than the specific dataset
drawn, and the goal is to avoid overfitting. \ A sequence of works
\cite{dwork:adaptive,dwork:adaptive2,bassily} showed that differentially
private mechanisms are particularly well-suited to this application, and can
be used to guarantee adaptive accuracy automatically.

Let $X$ be the dataset, with $n$ entries drawn i.i.d.\ from a distribution
$\mathcal{D}$. \ A priori, before any queries are answered, an observer's view
of the dataset $X$ is that it is a draw from the distribution $\mathcal{D}%
^{n}$. \ As queries are answered, this view might change. \ One way to prevent
overfitting is to guarantee that the query answers do not change the
observer's view much:\ i.e., that the a-posteriori view of $X$'s distribution,
conditioned on the observed answers, is almost unchanged. \ This can be
interpreted as \textquotedblleft classical gentleness.\textquotedblright\ \ At
a technical level, our results use the fact that in the above scenario, if we
run a classical DP algorithm $A$ on the database $X$, then conditioning on
$A(X)$ outputting any particular value $y$ results in a bounded change to the
prior (see Lemma \ref{dpgentleclassical}). \ We note that a similar result
follows from the work of Dwork et al.\ \cite{dwork:adaptive2} and Rogers et
al.\ \cite{rrst} (their results are phrased in terms of the so-called
\textquotedblleft max information\textquotedblright).

While there are technical and conceptual connections, the setting of quantum
measurement or shadow tomography (even without gentleness) presents altogether
different challenges from the adaptive data analysis setting. \ Most notably,
as we discussed in Section \ref{GENTLE}, running an algorithm on a quantum
state can collapse the state. \ This is a physical phenomenon, not just a
change in a particular observer's prior and posterior views as was the case
classically. \ In particular, quantum measurements that collapse the state
cannot be forgotten or undone. \ Restricting our attention to computing the
average of two-outcome measurements over $n$ registers, this difference is
best illustrated by the fact that, in the quantum setting, computing accurate
answers to a large collection of \textit{non-adaptive} measurements is already
a challenging task (even without requiring gentleness). \ In the classical
setting, on the other hand, if the measurements are specified non-adaptively
then the na\"{\i}ve algorithm that simply outputs the empirical mean for each
measurement performs quite well; the only challenge is answering an adaptively
specified sequence of measurements.

\section{Preliminaries\label{PRELIM}}

\subsection{Classical Probability Theory\label{PROB}}

Given two probability distributions $\mathcal{D}=\left(  p_{x}\right)  _{x}%
$\ and $\mathcal{D}^{\prime}=\left(  q_{x}\right)  _{x}$, we'll use all three
of the following measures of distance between them:%
\begin{align*}
\left\Vert \mathcal{D}-\mathcal{D}^{\prime}\right\Vert  &  :=\frac{1}{2}%
\sum_{x}\left\vert p_{x}-q_{x}\right\vert \text{ \ \ (trace distance)}\\
\operatorname{KL}\left(  \mathcal{D},\mathcal{D}^{\prime}\right)   &
:=\sum_{x}p_{x}\ln\frac{p_{x}}{q_{x}}\text{ \ \ (Kullback-Leibler
divergence)}\\
H^{2}\left(  \mathcal{D},\mathcal{D}^{\prime}\right)   &  :=1-\sum_{x}%
\sqrt{p_{x}q_{x}}\text{ \ \ (squared Hellinger distance)}%
\end{align*}
Interestingly, Hellinger distance was invented in 1909, prior to the discovery
of quantum mechanics, and is used for purely classical purposes in probability
theory. \ But, as it involves the square roots of probabilities, it might be
said to have a \textquotedblleft secret affinity\textquotedblright\ for
quantum mechanics that occasionally reveals itself, as it will in this paper.

\begin{proposition}
[Pinsker's Inequality]\label{pinsker}$\left\Vert \mathcal{D}-\mathcal{D}%
^{\prime}\right\Vert \leq\sqrt{2\operatorname{KL}\left(  \mathcal{D}%
,\mathcal{D}^{\prime}\right)  }$.
\end{proposition}

The following is less well-known, but we'll need it as well:

\begin{proposition}
[{e.g.\ \cite[p.\ 99]{reiss}}]\label{hkl}$H\left(  \mathcal{D},\mathcal{D}%
^{\prime}\right)  \leq\sqrt{\operatorname{KL}\left(  \mathcal{D}%
,\mathcal{D}^{\prime}\right)  }$.
\end{proposition}

\subsection{Quantum Information Basics\label{QI}}

In the following sections, we'll briefly review some standard notation and
definitions from quantum information. \ More details can be found, for
example, in Nielsen and Chuang \cite{nc}.

A $d$-dimensional \textit{pure state} is a unit vector in $\mathbb{C}^{d}$,
which we write in ket notation as%
\[
\left\vert \psi\right\rangle =\sum_{i=1}^{d}\alpha_{i}\left\vert
i\right\rangle .
\]
Here $\left\vert 1\right\rangle ,\ldots,\left\vert d\right\rangle $\ is an
orthonormal basis for $\mathbb{C}^{d}$,\ and the $\alpha_{i}$'s are complex
numbers called \textit{amplitudes} satisfying $\left\vert \alpha
_{1}\right\vert ^{2}+\cdots+\left\vert \alpha_{d}\right\vert ^{2}=1$. \ The
state $\left\vert \psi\right\rangle $\ is also called a \textit{superposition}
over the basis states $\left\vert 1\right\rangle ,\ldots,\left\vert
d\right\rangle $, which we can think of as the possible classical states of
the system.\footnote{Note that any linear combination of the basis states
$\left\vert 1\right\rangle ,\ldots,\left\vert d\right\rangle $,\ and not just
$\left\vert 1\right\rangle ,\ldots,\left\vert d\right\rangle $ themselves, is
called a \textquotedblleft pure state.\textquotedblright} \ We also denote by
$\left\langle \psi\right\vert $\ the conjugate transpose of $\left\vert
\psi\right\rangle $\ (thus, $\left\vert \psi\right\rangle $\ is a column
vector while $\left\langle \psi\right\vert $\ is a row vector). \ The
unit-norm condition can then be written succinctly as $\left\langle \psi
|\psi\right\rangle =1$; and more generally, the complex inner product between
$\left\vert \psi\right\rangle $\ and $\left\vert \varphi\right\rangle $\ can
be written $\left\langle \psi|\varphi\right\rangle $.

In the special case $d=2$, we call $\left\vert \psi\right\rangle $\ a
\textit{qubit}, and typically label the orthonormal basis vectors by
$\left\vert 0\right\rangle $\ and $\left\vert 1\right\rangle $. \ It's also
convenient to give standard names to the following two superpositions of
$\left\vert 0\right\rangle $\ and $\left\vert 1\right\rangle $:%
\[
\left\vert +\right\rangle :=\frac{\left\vert 0\right\rangle +\left\vert
1\right\rangle }{\sqrt{2}},~~~~\left\vert -\right\rangle :=\frac{\left\vert
0\right\rangle -\left\vert 1\right\rangle }{\sqrt{2}}.
\]

The reader might be familiar with two types of operations that we can apply to
pure states. \ First, given any unitary matrix $U$, we can map $\left\vert
\psi\right\rangle $\ to $U\left\vert \psi\right\rangle $. \ Second, we can
\textit{measure} $\left\vert \psi\right\rangle $\ in the $\left\vert
1\right\rangle ,\ldots,\left\vert d\right\rangle $\ basis. \ Doing so returns
the outcome $\left\vert i\right\rangle $\ with probability $\left\vert
\alpha_{i}\right\vert ^{2}=\left\vert \left\langle \psi|i\right\rangle
\right\vert ^{2}$. \ Furthermore, the state $\left\vert \psi\right\rangle
$\ then \textquotedblleft collapses\textquotedblright\ to $\left\vert
i\right\rangle $.

More generally, we could measure $\left\vert \psi\right\rangle $\ with respect
to \textit{any} orthonormal basis $\left\vert v_{1}\right\rangle
,\ldots,\left\vert v_{d}\right\rangle $, which is equivalent to first applying
a unitary $U$ that maps each $\left\vert v_{i}\right\rangle $\ to $\left\vert
i\right\rangle $, then measuring in the $\left\vert 1\right\rangle
,\ldots,\left\vert d\right\rangle $\ basis, and finally applying $U^{\dagger}%
$, where $\dagger$\ denotes conjugate transpose. \ This returns the outcome
$\left\vert v_{i}\right\rangle $ with probability $\left\vert \left\langle
\psi|v_{i}\right\rangle \right\vert ^{2}$, whereupon the state collapses to
$\left\vert v_{i}\right\rangle $. \ A measurement of this type is called a
\textit{projective} measurement.

\subsection{Mixed States, Superoperators, Quantum Operations, and
POVMs\label{MIXED}}

In general, we may have ordinary probabilistic uncertainty about which quantum
superposition we have. \ This leads us to \textit{mixed states}, the most
general kind of state in quantum mechanics. \ Formally, a $d$-dimensional
mixed state $\rho$\ is a$\ d\times d$\ positive semidefinite matrix\ that
satisfies $\operatorname{Tr}\left(  \rho\right)  =1$. \ Equivalently, $\rho$
is a convex combination of outer products of pure states with themselves
(without loss of generality, at most $d$ pure states):%
\[
\rho=\sum_{i=1}^{d}p_{i}|\psi_{i}\rangle\langle\psi_{i}|,
\]
where $p_{1},\ldots,p_{d}\geq0$\ and $p_{1}+\cdots+p_{d}=1$. \ This can be
interpreted as a probability distribution wherein each $\left\vert \psi
_{i}\right\rangle $\ occurs with probability $p_{i}$, though note that
different distributions can give rise to the same $\rho$. \ In the special
case where $\rho=|\psi\rangle\langle\psi|$\ has rank $1$, it represents a pure
state (i.e., a superposition). \ In the special case where $\rho$\ is
diagonal, it represents a classical probability distribution over $\left\vert
1\right\rangle ,\ldots,\left\vert d\right\rangle $. \ The \textit{maximally
mixed state}, $\frac{\mathbb{I}}{d}$ where $\mathbb{I}$\ is the identity
matrix, corresponds to the uniform distribution over $\left\vert
1\right\rangle ,\ldots,\left\vert d\right\rangle $, and has the unique
property of being unaffected by unitary transformations.

We can restate the basic rules of quantum mechanics in terms of mixed states,
as follows. \ First, a unitary transformation $U$ maps $\rho$\ to $U\rho
U^{\dagger}$. \ Second, a measurement of $\rho$\ in the $\left\vert
1\right\rangle ,\ldots,\left\vert d\right\rangle $\ basis returns the outcome
$\left\vert i\right\rangle $ with probability $\rho_{ii}=\left\langle
i|\rho|i\right\rangle $, whereupon $\rho$\ collapses to $\left\vert
i\right\rangle \left\langle i\right\vert $. \ Likewise, a measurement in the
$\left\vert v_{1}\right\rangle ,\ldots,\left\vert v_{d}\right\rangle $\ basis
returns $\left\vert v_{i}\right\rangle $ with probability $\left\langle
v_{i}|\rho|v_{i}\right\rangle $, whereupon $\rho$\ collapses to $|v_{i}%
\rangle\langle v_{i}|$.

More generally, a \textit{superoperator} $S$, the most general
(norm-preserving) mapping from mixed states to mixed states allowed by quantum
mechanics, maps $\rho$\ to the mixed state%
\[
S\left(  \rho\right)  =\sum_{i=1}^{k}B_{i}\rho B_{i}^{\dagger},
\]
where $B_{1},\ldots,B_{k}$\ can be any matrices satisfying%
\[
\sum_{i=1}^{k}B_{i}^{\dagger}B_{i}=\mathbb{I}.
\]
Here $\rho$\ and $S\left(  \rho\right)  $\ do not even need to have the same
dimension. \ Superoperators encompass unitary transformations, measurements,
and other interactions with an external environment in a single formalism.

Even more generally still, if we have $S$\ as above where $B_{1},\ldots,B_{k}%
$\ only satisfy%
\[
\sum_{i=1}^{k}B_{i}^{\dagger}B_{i}\preceq\mathbb{I},
\]
then we call $S$ a \textit{quantum operation}.\footnote{In the literature,
these are also called \textquotedblleft non-trace-increasing completely
positive maps.\textquotedblright} \ If $S$ is a quantum operation, then
$S\left(  \rho\right)  $\ is Hermitian and positive semidefinite, but it might
not be a normalized mixed state, because its trace might be less than $1$.
\ Quantum operations are useful for capturing the effects of superoperators
when we additionally \textit{condition} on some event happening (e.g., a
measurement outcome being \textquotedblleft accept\textquotedblright). \ The
event's probability is then $\operatorname{Tr}\left(  S\left(  \rho\right)
\right)  $, and the final mixed state conditioned on the event is
$\frac{S\left(  \rho\right)  }{\operatorname{Tr}\left(  S\left(  \rho\right)
\right)  }$.

Quantum operations act linearly on mixed states, in the sense that%
\[
S\left(  p\rho+q\sigma\right)  =pS\left(  \rho\right)  +qS\left(
\sigma\right)  .
\]

Although any measurement can be represented by a superoperator, when
discussing measurements it's convenient to use a related formalism called
\textquotedblleft POVMs\textquotedblright\ (Positive Operator Valued
Measures). \ POVMs capture all measurements allowed by quantum mechanics,
including those whose implementations might involve ancillary systems besides
the ones being measured. \ In this formalism, a \textit{measurement} $M$ is
given by a list of $d\times d$\ positive semidefinite matrices $E_{1}%
,\ldots,E_{k}$, which satisfy $E_{1}+\cdots+E_{k}=\mathbb{I}$. \ The rule is:%
\[
\Pr\left[  M\left(  \rho\right)  \text{ returns outcome }i\right]
=\operatorname{Tr}\left(  E_{i}\rho\right)  .
\]
Importantly, specifying the $E_{i}$'s doesn't uniquely determine the
post-measurement states (i.e., what happens to $\rho$\ if the outcome is $i$).
\ Thus, by an \textit{implementation}\ of the measurement $M$, in this paper
we'll mean a list of $d\times d$ matrices $B_{1},\ldots,B_{k}$, which satisfy
$B_{i}^{\dagger}B_{i}=E_{i}$. \ For a given implementation, if the measurement
outcome is $i$, then the post-measurement state is%
\[
\frac{B_{i}\rho B_{i}^{\dagger}}{\operatorname{Tr}\left(  B_{i}\rho
B_{i}^{\dagger}\right)  }.
\]
Note that the mapping%
\[
\rho\rightarrow\sum_{i=1}^{k}B_{i}\rho B_{i}^{\dagger}%
\]
is a superoperator, that each individual mapping $\rho\rightarrow B_{i}\rho
B_{i}^{\dagger}$\ is a quantum operation,\ and that $\operatorname{Tr}\left(
B_{i}\rho B_{i}^{\dagger}\right)  =\operatorname{Tr}\left(  E_{i}\rho\right)
$ is the probability of outcome $i$.

In the special case of two-outcome POVMs $\left(  E_{1},E_{2}\right)  $, we'll
sometimes identify the POVM itself with the \textquotedblleft
accept\textquotedblright\ outcome $E_{1}$, treating the \textquotedblleft
reject\textquotedblright\ outcome $E_{2}=\mathbb{I}-E_{1}$\ as implied.

\subsection{Separable and Entangled\label{SEPENT}}

A pure state $\left\vert \psi\right\rangle $\ on $n$ registers is called a
\textit{product state} if it can be written as a tensor product,%
\[
\left\vert \psi\right\rangle =\left\vert \psi_{1}\right\rangle \otimes
\cdots\otimes\left\vert \psi_{n}\right\rangle .
\]
Any pure state that cannot be so written is called \textit{entangled}. \ A
famous example of an entangled pure state is the \textit{Bell pair},
$\frac{\left\vert 00\right\rangle +\left\vert 11\right\rangle }{\sqrt{2}}$.

A mixed state $\rho$\ is likewise called a product state if it can be written
as a tensor product%
\[
\rho=\rho_{1}\otimes\cdots\otimes\rho_{n}.
\]
Also, $\rho$\ is called \textit{separable} if it can be written as a convex
combination of product states, and \textit{entangled} if it can't be. \ Unlike
a pure state, a mixed state can be separable but non-product, meaning that it
has classical correlations but no entanglement, as with the example
$\frac{|00\rangle\langle00|+|11\rangle\langle11|}{2}$ (i.e., $00$\ and
$11$\ with equal probabilities).

A measurement $M$ on an $n$-register state is called \textit{product} if there
exist POVMs $M_{1},\ldots,M_{n}$\ such that $M$ can be implemented as follows:

\begin{itemize}
\item For each $i\in\left[  n\right]  $, apply $M_{i}$\ to the $i^{th}$ register.

\item Return some function of the $n$ classical measurement outcomes, possibly
together with auxiliary randomness.
\end{itemize}

In the special case where $M_{1},\ldots,M_{n}$\ are all projective
measurements, we call $M$ a \textit{product-of-projectives}.

More generally, we call $M$ \textit{mixture-of-products} if the POVMs
$M_{1},\ldots,M_{n}$\ can be chosen randomly, from some correlated probability
distribution, in advance of applying them.

More generally still, we call $M$ \textit{LOCC}---the acronym stands for Local
Operations and Classical Communication---if $M$ can be implemented by applying
a POVM to some register $i_{1}\in\left[  n\right]  $, then (depending on the
outcome) applying another POVM to some register $i_{2}\in\left[  n\right]  $,
and so on, then finally returning some function of the classical measurement
outcomes, possibly together with auxiliary randomness. \ Here we allow any
finite, adaptively chosen sequence of POVMs, which could include repeated
POVMs applied to the same register.

Let us stress that, even if a measurement happens to be product, or
mixture-of-products, or LOCC, if we want to implement the measurement
\textit{gently}, we might need to apply a quantum circuit that acts on all $n$
registers coherently. \ This is because, if we measure the registers
separately, we'll generate \textit{garbage}---i.e., information about the
state besides the final measurement outcome---that might destroy gentleness.
\ Only if we've taken care to do everything in coherent superposition,
simulating the \textquotedblleft measurements\textquotedblright\ on the
individual registers (and the computations on the outcomes of those
measurements) using ancilla qubits, can we later uncompute the garbage. \ This
is likely to be a significant challenge for experimental implementation of
gentle measurements like the ones discussed in this paper, since coherent
measurements across $n$ registers are much harder than incoherent ones to
realize in practice. \ On the other hand, this issue makes no difference for
DP, since even if the garbage isn't uncomputed, it need not be revealed to the
end user.\footnote{Or to say it another way, the definition of quantum DP
talks only about the probabilities of outcomes, not about the post-measurement
states.}

\subsection{Distance Between Quantum States\label{DIST}}

Given a Hermitian matrix $A$, its \textit{trace norm} is defined as%
\[
\left\Vert A\right\Vert _{\operatorname*{tr}}:=\frac{1}{2}\sum_{i=1}%
^{d}\left\vert \lambda_{i}\right\vert ,
\]
where $\lambda_{1},\ldots,\lambda_{d}$\ are the eigenvalues of $A$. \ In
particular, given two mixed states $\rho$\ and $\sigma$, their \textit{trace
distance} is defined as $\left\Vert \rho-\sigma\right\Vert
_{\operatorname*{tr}}$.

Trace distance is a metric on mixed states---i.e., it's reflexive, symmetric,
and satisfies the triangle inequality. \ It's equal to%
\[
\max_{M}\left(  \Pr\left[  M\left(  \rho\right)  \text{ accepts}\right]
-\Pr\left[  M\left(  \sigma\right)  \text{ accepts}\right]  \right)  ,
\]
where the maximum is taken over all possible two-outcome measurements $M$.
\ As such, trace distance generalizes the total variation distance between
classical probability distributions, reducing to the latter when $\rho$\ and
$\sigma$\ are both diagonal matrices.

We'll find the following facts useful.

\begin{proposition}
[{Contractivity of Trace Norm \cite[p.\ 406]{nc}}]\label{contractive}Let
$S$\ be any quantum operation, and let $A$\ be a Hermitian matrix. \ Then%
\[
\left\Vert S\left(  A\right)  \right\Vert _{\operatorname*{tr}}\leq\left\Vert
A\right\Vert _{\operatorname*{tr}}.
\]

\end{proposition}

So in particular, for any two mixed states $\rho$\ and $\sigma$ and any
quantum operation $S$, we have%
\[
\left\Vert S\left(  \rho\right)  -S\left(  \sigma\right)  \right\Vert
_{\operatorname*{tr}}=\left\Vert S\left(  \rho-\sigma\right)  \right\Vert
_{\operatorname*{tr}}\leq\left\Vert \rho-\sigma\right\Vert
_{\operatorname*{tr}}.
\]
As an especially useful example, a superoperator that \textquotedblleft traces
out\textquotedblright\ (discards) part of its input state can never increase
trace distance.

\begin{proposition}
[Convexity of Trace Norm]\label{convextd}For all Hermitian matrices $A,B$ and
$p,q\geq0$,%
\[
\left\Vert pA+qB\right\Vert _{\operatorname*{tr}}\leq p\left\Vert A\right\Vert
_{\operatorname*{tr}}+q\left\Vert B\right\Vert _{\operatorname*{tr}}.
\]

\end{proposition}

The triangle inequality for trace distance is just a special case of the
above. \ As another useful special case, for all mixed states $\rho
,\sigma,\rho^{\prime},\sigma^{\prime}$ and probabilities $p$,%

\[
\left\Vert \left(  p\rho+\left(  1-p\right)  \sigma\right)  -\left(
p\rho^{\prime}+\left(  1-p\right)  \sigma^{\prime}\right)  \right\Vert
_{\operatorname*{tr}}\leq p\left\Vert \rho-\rho^{\prime}\right\Vert
_{\operatorname*{tr}}+\left(  1-p\right)  \left\Vert \sigma-\sigma^{\prime
}\right\Vert _{\operatorname*{tr}}.
\]

Finally, trace distance $\left\Vert \rho-\sigma\right\Vert
_{\operatorname*{tr}}$\ takes an especially simple form if $\rho=|\psi
\rangle\langle\psi|$\ and $\sigma=|\phi\rangle\langle\phi|$\ are both pure states.

\begin{proposition}
\label{puretd}For all $\left\vert \psi\right\rangle ,\left\vert \phi
\right\rangle $,%
\[
\left\Vert |\psi\rangle\langle\psi|-|\phi\rangle\langle\phi|\right\Vert
_{\operatorname*{tr}}=\sqrt{1-\left\vert \left\langle \psi|\phi\right\rangle
\right\vert ^{2}}.
\]

\end{proposition}

\subsection{Additivity of Damage\label{DAMAGE}}

In this section, we prove the extremely useful fact that, if we apply quantum
operations to a quantum state $\rho$ in succession, then we can bound the
total damage caused to $\rho$\ in trace distance by the sum of the damages
that each operation \textit{would} cause were it applied to $\rho
$\ individually. \ This fact is related to the so-called \textquotedblleft
Quantum Union Bound\textquotedblright\ (see \cite{aar:qmaqpoly,wilde}), but
it's both simpler to state and easier to prove.

\begin{lemma}
\label{damagelem}Let $\rho$\ be a mixed state, and let $S$ be any quantum
operation. \ Suppose $\left\Vert \rho^{\prime}-\rho\right\Vert
_{\operatorname*{tr}}\leq\varepsilon$, and let $\left\Vert S\left(
\rho\right)  -\rho\right\Vert _{\operatorname*{tr}}\leq\delta$. \ Then
$\left\Vert S\left(  \rho^{\prime}\right)  -\rho\right\Vert
_{\operatorname*{tr}}\leq\varepsilon+\delta$.
\end{lemma}

\begin{proof}
We have%
\begin{align*}
\left\Vert S\left(  \rho^{\prime}\right)  -\rho\right\Vert
_{\operatorname*{tr}}  &  =\left\Vert S\left(  \rho+\left(  \rho^{\prime}%
-\rho\right)  \right)  -\rho\right\Vert _{\operatorname*{tr}}\\
&  =\left\Vert S\left(  \rho^{\prime}-\rho\right)  +\left(  S\left(
\rho\right)  -\rho\right)  \right\Vert _{\operatorname*{tr}}\\
&  \leq\left\Vert \rho^{\prime}-\rho\right\Vert _{\operatorname*{tr}%
}+\left\Vert S\left(  \rho\right)  -\rho\right\Vert _{\operatorname*{tr}}\\
&  \leq\varepsilon+\delta.
\end{align*}
Here the second line used the linearity of quantum operations, and the third
used the triangle inequality for trace distance as well as Proposition
\ref{contractive} (i.e., the fact that applying a quantum operation can never
increase the trace norm).
\end{proof}

Lemma \ref{damagelem}\ has the following immediate corollary.

\begin{corollary}
\label{damagecor}Let $\rho$ be a mixed state and let $S_{1},\ldots,S_{m}$ be
quantum operations. \ Suppose that for all $i$, we have%
\[
\left\Vert S_{i}\left(  \rho\right)  -\rho\right\Vert _{\operatorname*{tr}%
}\leq\varepsilon_{i}.
\]
Then%
\[
\left\Vert S_{m}\left(  S_{m-1}\left(  \cdots S_{1}\left(  \rho\right)
\right)  \right)  -\rho\right\Vert _{\operatorname*{tr}}\leq\varepsilon
_{1}+\cdots+\varepsilon_{m}.
\]

\end{corollary}

\begin{proof}
Suppose by induction on $m$ that%
\[
\left\Vert S_{m-1}\left(  S_{m-2}\left(  \cdots S_{1}\left(  \rho\right)
\right)  \right)  -\rho\right\Vert _{\operatorname*{tr}}\leq\varepsilon
_{1}+\cdots+\varepsilon_{m-1}.
\]
Then%
\begin{align*}
\left\Vert S_{m}\left(  S_{m-1}\left(  \cdots S_{1}\left(  \rho\right)
\right)  \right)  -\rho\right\Vert _{\operatorname*{tr}}  &  \leq\left\Vert
S_{m-1}\left(  S_{m-2}\left(  \cdots S_{1}\left(  \rho\right)  \right)
\right)  -\rho\right\Vert _{\operatorname*{tr}}+\left\Vert S_{m}\left(
\rho\right)  -\rho\right\Vert _{\operatorname*{tr}}\\
&  \leq\varepsilon_{1}+\cdots+\varepsilon_{m}.
\end{align*}
by Lemma \ref{damagelem}.
\end{proof}

Corollary \ref{damagecor} is the reason why \textquotedblleft gentleness
composes\textquotedblright: that is, applying an $\alpha_{1}$-gentle
measurement to a state $\rho$, followed by an $\alpha_{2}$-gentle measurement,
yields an overall $\left(  \alpha_{1}+\alpha_{2}\right)  $-gentle measurement.
\ By contrast, it's not clear to what extent DP composes in the quantum
setting, because of the interaction between the DP requirement and damage to
the state. \ For more about this issue see Appendix \ref{CANTCOMPOSE}.

Note that, by simply specializing Corollary \ref{damagecor} to diagonal $\rho
$\ and classical operations $S_{1},\ldots,S_{m}$, we obtain an analogous
statement for classical variation distance.

In Section \ref{SHADOW}, when we analyze our shadow tomography protocol, we'll
also need a lemma that upper-bounds the damage caused by a sequence of
measurements \textit{conditional on the measurements all accepting}---or
equivalently, by a sequence of quantum operations where we normalize the final
result. \ Fortunately, the formalism of quantum operations and trace norm can
accommodate this case as well.

\begin{lemma}
[Damage Lemma]\label{facepalmlemma}Let $\rho$ be a mixed state. \ For all
$i\in\left[  m\right]  $, let $S_{i}$\ be a quantum operation, which
\textquotedblleft accepts\textquotedblright\ a state $\sigma$\ with
probability $\operatorname{Tr}\left(  S_{i}\left(  \sigma\right)  \right)
>0$, and yields the post-measurement state $\frac{S_{i}\left(  \sigma\right)
}{\operatorname{Tr}\left(  S_{i}\left(  \sigma\right)  \right)  }$ when it
does. \ Suppose that for all $i\in\left[  m\right]  $, we have%
\[
\left\Vert \frac{S_{i}\left(  \rho\right)  }{\operatorname{Tr}\left(
S_{i}\left(  \rho\right)  \right)  }-\rho\right\Vert _{\operatorname*{tr}}%
\leq\varepsilon_{i}.
\]
Let $p_{i}:=\operatorname{Tr}\left(  S_{i}\left(  \rho\right)  \right)  $\ be
the probability that $S_{i}$\ accepts the \textquotedblleft
ideal\textquotedblright\ state $\rho$, and let%
\[
q_{i}:=\frac{\operatorname{Tr}\left(  S_{i}\left(  \cdots S_{1}\left(
\rho\right)  \right)  \right)  }{\operatorname{Tr}\left(  S_{i-1}\left(
\cdots S_{1}\left(  \rho\right)  \right)  \right)  }%
\]
be the probability that $S_{i}$ accepts the state that it actually receives,
if $S_{1},\ldots,S_{i-1}$\ are first applied to $\rho$\ and if we condition on
their accepting. \ Given any subset $T\subseteq\left[  m\right]  $, let%
\[
p_{T}:=%
{\displaystyle\prod\limits_{i\in T}}
p_{i},~~~~~q_{T}:=%
{\displaystyle\prod\limits_{i\in T}}
q_{i}.
\]
Then for all $T$,%
\[
\left\vert p_{T}-q_{T}\right\vert \leq\frac{\varepsilon_{1}+\cdots
+\varepsilon_{m}}{q_{\left[  m\right]  \setminus T}}.
\]
Also,%
\[
\left\Vert \frac{S_{m}\left(  \cdots S_{1}\left(  \rho\right)  \right)
}{\operatorname{Tr}\left(  S_{m}\left(  \cdots S_{1}\left(  \rho\right)
\right)  \right)  }-\rho\right\Vert _{\operatorname*{tr}}\leq\frac
{2}{q_{\left[  m\right]  }}\left(  \varepsilon_{1}+\cdots+\varepsilon
_{m}\right)  .
\]

\end{lemma}

\begin{proof}
For all $i\in\left[  m\right]  $, let%
\[
E_{i}:=\frac{S_{i}\left(  \rho\right)  }{\operatorname{Tr}\left(  S_{i}\left(
\rho\right)  \right)  }-\rho.
\]
Then by hypothesis, $\left\Vert E_{i}\right\Vert _{\operatorname*{tr}}%
\leq\varepsilon_{i}$. \ Also,%
\[
S_{i}\left(  \rho\right)  =p_{i}\left(  \rho+E_{i}\right)  .
\]
We can now write:%
\begin{align*}
S_{1}\left(  \rho\right)   &  =p_{1}\left(  \rho+E_{1}\right)  ,\\
S_{2}\left(  S_{1}\left(  \rho\right)  \right)   &  =p_{1}\left(  S_{2}\left(
\rho\right)  +S_{2}\left(  E_{1}\right)  \right) \\
&  =p_{1}\left(  p_{2}\left(  \rho+E_{2}\right)  +S_{2}\left(  E_{1}\right)
\right)  ,\\
S_{3}\left(  S_{2}\left(  S_{1}\left(  \rho\right)  \right)  \right)   &
=p_{1}\left(  p_{2}\left(  S_{3}\left(  \rho\right)  +S_{3}\left(
E_{2}\right)  \right)  +S_{3}\left(  S_{2}\left(  E_{1}\right)  \right)
\right) \\
&  =p_{1}\left(  p_{2}\left(  p_{3}\left(  \rho+E_{3}\right)  +S_{3}\left(
E_{2}\right)  \right)  +S_{3}\left(  S_{2}\left(  E_{1}\right)  \right)
\right)  ,
\end{align*}
and so on until%
\begin{align*}
S_{m}\left(  \cdots S_{1}\left(  \rho\right)  \right)  =  &  ~p_{\left[
m\right]  }\rho+p_{\left[  m\right]  }E_{m}+p_{\left[  m-1\right]  }%
S_{m}\left(  E_{m-1}\right) \\
&  +p_{\left[  m-2\right]  }S_{m}\left(  S_{m-1}\left(  E_{m-2}\right)
\right)  +\cdots+p_{\left[  1\right]  }S_{m}\left(  \cdots S_{2}\left(
E_{1}\right)  \right)  .
\end{align*}
More generally, suppose we define%
\[
S_{i}^{\prime}\left(  \sigma\right)  :=\left\{
\begin{tabular}
[c]{ll}%
$S_{i}\left(  \sigma\right)  $ & if $i\in T$\\
$\frac{S_{i}\left(  \sigma\right)  }{q_{i}}$ & otherwise,
\end{tabular}
\ \right.
\]
so that%
\[
\operatorname{Tr}\left(  S_{m}^{\prime}\left(  \cdots S_{1}^{\prime}\left(
\rho\right)  \right)  \right)  =\frac{q_{1}\cdots q_{m}}{\prod_{i\notin
T}q_{i}}=q_{T}%
\]
is just the probability that $S_{i}$\ accepts for all $i\in T$, if
$S_{1},\ldots,S_{m}$ are applied in sequence. \ Then repeating the
manipulations above gives us the following modified equation, in which all the
products of $p_{i}$'s\ are restricted to range only over $i\in T$:%
\begin{align*}
S_{m}^{\prime}\left(  \cdots S_{1}^{\prime}\left(  \rho\right)  \right)  =  &
~p_{T}\left(  \rho+E_{m}\right)  +p_{T\cap\left[  m-1\right]  }S_{m}^{\prime
}\left(  E_{m-1}\right) \\
&  +p_{T\cap\left[  m-2\right]  }S_{m}^{\prime}\left(  S_{m-1}^{\prime}\left(
E_{m-2}\right)  \right)  +\cdots+p_{T\cap\left[  1\right]  }S_{m}^{\prime
}\left(  \cdots S_{2}^{\prime}\left(  E_{1}\right)  \right)  .
\end{align*}
Hence%
\begin{align*}
\left\Vert S_{m}^{\prime}\left(  \cdots S_{1}^{\prime}\left(  \rho\right)
\right)  -p_{T}\rho\right\Vert _{\operatorname*{tr}}\leq &  ~p_{T}\left\Vert
E_{m}\right\Vert _{\operatorname*{tr}}+p_{T\cap\left[  m-1\right]  }\left\Vert
S_{m}^{\prime}\left(  E_{m-1}\right)  \right\Vert _{\operatorname*{tr}}%
+\cdots+p_{T\cap\left[  1\right]  }\left\Vert S_{m}^{\prime}\left(  \cdots
S_{2}^{\prime}\left(  E_{1}\right)  \right)  \right\Vert _{\operatorname*{tr}%
}\\
\leq &  \left\Vert E_{m}\right\Vert _{\operatorname*{tr}}+\left\Vert
S_{m}^{\prime}\left(  E_{m-1}\right)  \right\Vert _{\operatorname*{tr}}%
+\cdots+\left\Vert S_{m}^{\prime}\left(  \cdots S_{2}^{\prime}\left(
E_{1}\right)  \right)  \right\Vert _{\operatorname*{tr}}\\
\leq &  \varepsilon_{m}+\frac{\varepsilon_{m-1}}{q_{\left\{  m\right\}
\setminus T}}+\cdots+\frac{\varepsilon_{1}}{q_{\left\{  2,\ldots,m\right\}
\setminus T}}\\
\leq &  \frac{\varepsilon_{1}+\cdots+\varepsilon_{m}}{q_{\left[  m\right]
\setminus T}},
\end{align*}
where the second line used the fact that all the products of $p_{i}$'s are
upper-bounded by $1$. \ This means that%
\[
\left\vert q_{T}-p_{T}\right\vert =\left\vert \operatorname{Tr}\left(
S_{m}^{\prime}\left(  \cdots S_{1}^{\prime}\left(  \rho\right)  \right)
\right)  -p_{T}\right\vert \leq\frac{\varepsilon_{1}+\cdots+\varepsilon_{m}%
}{q_{\left[  m\right]  \setminus T}},
\]
thereby proving the first part of the lemma.

For the second part, let us take the special case $T=\left[  m\right]  $.
\ Then $q_{\left[  m\right]  \setminus T}=1$, and the inequalities above
reduce to%
\begin{align*}
\left\Vert S_{m}\left(  \cdots S_{1}\left(  \rho\right)  \right)  -p_{\left[
m\right]  }\rho\right\Vert _{\operatorname*{tr}}  &  \leq\varepsilon
_{1}+\cdots+\varepsilon_{m},\\
\left\vert q_{\left[  m\right]  }-p_{\left[  m\right]  }\right\vert  &
\leq\varepsilon_{1}+\cdots+\varepsilon_{m}.
\end{align*}
So the triangle inequality gives%
\begin{align*}
\left\Vert S_{m}\left(  \cdots S_{1}\left(  \rho\right)  \right)  -q_{\left[
m\right]  }\rho\right\Vert _{\operatorname*{tr}}  &  \leq\left\Vert
S_{m}\left(  \cdots S_{1}\left(  \rho\right)  \right)  -p_{\left[  m\right]
}\rho\right\Vert _{\operatorname*{tr}}+\left\Vert p_{\left[  m\right]  }%
\rho-q_{\left[  m\right]  }\rho\right\Vert _{\operatorname*{tr}}\\
&  \leq\varepsilon_{1}+\cdots+\varepsilon_{m}+\left\vert p_{\left[  m\right]
}-q_{\left[  m\right]  }\right\vert \\
&  \leq2\left(  \varepsilon_{1}+\cdots+\varepsilon_{m}\right)  .
\end{align*}
Hence%
\[
\left\Vert \frac{S_{m}\left(  \cdots S_{1}\left(  \rho\right)  \right)
}{q_{\left[  m\right]  }}-\rho\right\Vert _{\operatorname*{tr}}\leq\frac
{2}{q_{\left[  m\right]  }}\left(  \varepsilon_{1}+\cdots+\varepsilon
_{m}\right)  .
\]

\end{proof}

As we'll show in Appendix \ref{CANTCOMPOSE}, Lemma \ref{facepalmlemma} implies
a limited sort of composition for quantum DP algorithms. \ Namely, we can
sequentially compose $k$ quantum DP algorithms and have the result remain
accurate and DP, so long as the total damage incurred to the quantum state (in
trace distance) is always small compared to the joint probability of the
observed outcomes $y_{1},\ldots,y_{k}$. \ We can sometimes ensure the latter
property, in turn, by using our main result, the connection between DP and gentleness.

Note that we can combine Lemmas \ref{facepalmlemma} and \ref{damagelem}, to
say that, even if we apply a final superoperator $S_{m+1}$\ \textit{after}
applying the quantum operations $S_{1},\ldots,S_{m}$\ and then conditioning on
their results, the total damage to our initial state $\rho$\ is at most
$\left\Vert S_{m+1}\left(  \rho\right)  -\rho\right\Vert _{\operatorname*{tr}%
}$\ plus the damage bound from Lemma \ref{facepalmlemma}. \ (This wouldn't be
true if we'd composed in the opposite order, since conditioning could amplify
earlier damage to $\rho$\ by an $O\left(  \frac{1}{p}\right)  $\ factor.)
\ This fact will also be used in Section \ref{SHADOW}.

\subsection{Pure vs. Mixed States\label{PUREMIXED}}

We now prove two propositions to show that, when considering differential
privacy and gentle measurements, we can restrict attention to pure states
without loss of generality; our conclusions will then automatically carry over
to mixed states.

\begin{proposition}
\label{mixeddp}If $M$ is $\varepsilon$-DP on pure product states,\ then $M$ is
$\varepsilon$-DP on mixed product states as well. \ Likewise, if $M$ is
$\varepsilon$-DP\ on all pure states, then $M$ is $\varepsilon$-DP\ on all
mixed states.
\end{proposition}

\begin{proof}
Suppose we seek to maximize the ratio%
\[
\frac{\Pr\left[  M\left(  \rho\right)  =y\right]  }{\Pr\left[  M\left(
\sigma\right)  =y\right]  }%
\]
over product states $\rho=\rho_{1}\otimes\cdots\otimes\rho_{n}$\ and
$\sigma=\sigma_{1}\otimes\cdots\otimes\sigma_{n}$\ that differ only on the
$i^{th}$\ register. \ Then holding the other $n-1$\ registers fixed, we're
maximizing over $\rho_{i}$\ and minimizing over $\sigma_{i}$. \ By convexity,
the maximum and minimum will both always be achieved by pure states. \ A
second appeal to convexity then shows that the maximum ratio is also achieved
when the other $n-1$\ registers\ are set to pure states as well.

For the second part, the argument is the same, except that we simply maximize
$\Pr\left[  M\left(  \rho\right)  =y\right]  $\ over \textit{all} $\rho$, and
minimize $\Pr\left[  M\left(  \sigma\right)  =y\right]  $\ over all $\sigma$.
\end{proof}

\begin{proposition}
\label{mixedgentle}If the measurement $M$ is $\alpha$-gentle on pure product
states, then $M$ is $\alpha$-gentle on mixed product states as well.
\ Likewise, if $M$ is $\alpha$-gentle\ on all pure states, then $M$ is
$\alpha$-gentle\ on all mixed states.
\end{proposition}

\begin{proof}
Fix an implementation of $M$; the same implementation that achieves gentleness
on pure states will also achieve gentleness on mixed states.

Suppose we apply $M$ to the product state $\rho=\rho_{1}\otimes\cdots
\otimes\rho_{n}$. \ As a first step, we can purify $\rho_{1},\ldots,\rho_{n}%
$\ to $\left\vert \psi_{1}\right\rangle ,\ldots,\left\vert \psi_{n}%
\right\rangle $\ respectively by adding registers to them. \ Then $M$ can be
seen as acting on the pure state $\left\vert \psi\right\rangle =\left\vert
\psi_{1}\right\rangle \otimes\cdots\otimes\left\vert \psi_{n}\right\rangle $,
and simply ignoring these added registers. \ By assumption, after we apply $M$
and condition on some outcome $y$, we're left with a post-measurement state
$\left\vert \psi_{y}\right\rangle $\ such that%
\[
\left\Vert |\psi_{y}\rangle\langle\psi_{y}|-|\psi\rangle\langle\psi
|\right\Vert _{\operatorname*{tr}}\leq\alpha.
\]
Likewise, let $\rho_{y}$\ be the post-measurement state if we apply $M$ to
$\rho$\ and then condition on outcome $y$. \ Observe that $\rho_{y}$\ is
\textit{also} the mixed state obtained by starting from $\left\vert \psi
_{y}\right\rangle $\ and then tracing out the added registers. \ So by
Proposition \ref{contractive}, we have $\left\Vert \rho_{y}-\rho\right\Vert
_{\operatorname*{tr}}\leq\alpha$\ as well.

For the second part, the argument is the same, except that we purify $\rho
$\ as a whole rather than $\rho_{1},\ldots,\rho_{n}$\ separately.
\end{proof}

\section{Basic Relations Among DP, Gentleness, and\ Triviality\label{BASIC}}

In this section, we prove our first connection between the differential
privacy and the gentleness of quantum measurements:

\begin{theorem}
\label{warmup}If a measurement $M$ is $\varepsilon$-DP on all states,\ then
$M$ is $O\left(  \varepsilon n\right)  $-gentle on all states. \ Conversely,
if $M$ is $\alpha$-gentle on all states for $\alpha\leq\frac{1}{4.01}$, then
$M$\ is $O\left(  \alpha\right)  $-DP\ on all states.
\end{theorem}

Unfortunately, Theorem \ref{warmup} is weaker than it might look, since as
we'll see, it relates DP to gentleness only in a regime where $M$ is
\textquotedblleft nearly trivial.\textquotedblright\ \ Later, we'll restrict
our attention to product states, which will lead to a much more interesting
connection between DP and gentleness. \ Nevertheless, Theorem \ref{warmup}
serves as an instructive warmup to our main results, and the tools used to
prove it will later be reused.

Note that all the results in this section also have classical analogues---we
simply need to replace \textquotedblleft all (mixed) states\textquotedblright%
\ by \textquotedblleft all probability distributions\textquotedblright\ in
each definition and statement---and those classical analogues might be of
interest as well.

Let's first define formally what we mean by a measurement being
\textquotedblleft nearly trivial.\textquotedblright

\begin{definition}
[Triviality]Given a set $S$ of mixed states, a measurement $M$, and a
parameter $\varepsilon\geq0$, we say $M$ is $\varepsilon$\textbf{-trivial on}
$S$ if for all states $\rho,\sigma\in S$, and all possible outcomes $y$ of
$M$, we have%
\[
\Pr\left[  M\left(  \rho\right)  \text{ outputs }y\right]  \leq e^{\varepsilon
}\Pr\left[  M\left(  \sigma\right)  \text{ outputs }y\right]  .
\]
For $M$ to be $\varepsilon$-trivial, full stop, means that $M$ is
$\varepsilon$-trivial on the set of all states.
\end{definition}

In particular, $M$ is $0$-trivial if and only if $M$'s output probabilities
are completely independent of $\rho$. \ Note also that $\varepsilon$-trivial
immediately implies $\varepsilon$-DP. \ Like $\varepsilon$-DP (but unlike
$\alpha$-gentleness), the definition of $\varepsilon$-triviality depends only
on the outcome probabilities, and not on the post-measurement states.

The following proposition gives a slightly weaker condition that already
suffices for a measurement to be $\varepsilon$-trivial.

\begin{proposition}
\label{eigprop}Given a measurement $M$ and parameter $\varepsilon\geq0$,
suppose that for every two orthogonal pure states $\left\vert \psi
\right\rangle $ and $\left\vert \phi\right\rangle $, and every possible
outcome $y$ of $M$, we have%
\begin{equation}
\Pr\left[  M\left(  \left\vert \psi\right\rangle \right)  \text{ outputs
}y\right]  \leq e^{\varepsilon}\Pr\left[  M\left(  \left\vert \phi
\right\rangle \right)  \text{ outputs }y\right]  .\nonumber
\end{equation}
Then $M$ is $\varepsilon$-trivial.
\end{proposition}

\begin{proof}
Let $E_{1}+\cdots+E_{k}=\mathbb{I}$\ be the POVM elements of $M$. \ Assume
without loss of generality that the outcome $y$ corresponds to the element
$E=E_{1}$. \ Then by assumption,%
\begin{equation}
\left\langle \psi|E|\psi\right\rangle \leq e^{\varepsilon}\left\langle
\phi|E|\phi\right\rangle \label{ee}%
\end{equation}
for all orthogonal $\left\vert \psi\right\rangle ,\left\vert \phi\right\rangle
$. \ But this means that all of $E$'s eigenvalues must be within an
$e^{\varepsilon}$ multiplicative factor of each other. \ So (\ref{ee}) holds
for \textit{all} $\left\vert \psi\right\rangle ,\left\vert \phi\right\rangle
$, not just all orthogonal $\left\vert \psi\right\rangle ,\left\vert
\phi\right\rangle $. \ By convexity, we then have%
\[
\operatorname{Tr}\left(  E\rho\right)  \leq e^{\varepsilon}\operatorname{Tr}%
\left(  E\sigma\right)
\]
for all $\rho,\sigma$\ as well.
\end{proof}

Using Proposition \ref{eigprop}, we now show that gentleness on all states
implies near-triviality.

\begin{lemma}
\label{mixtrick}Suppose $M$ is $\alpha$-gentle on all states. \ Then $M$ is
$\ln\left(  \frac{1+4\alpha}{1-4\alpha}\right)  $-trivial---so in particular,
$O\left(  \alpha\right)  $-trivial, provided $\alpha\leq\frac{1}{4.01}$.
\end{lemma}

\begin{proof}
Given mixed states $\rho$\ and $\sigma$, let's first consider the special case
where $\rho$ and $\sigma$\ are perfectly distinguishable (that is, $\left\Vert
\rho-\sigma\right\Vert _{\operatorname*{tr}}=1$). \ For any outcome $y$, let
$p:=\Pr\left[  M\left(  \rho\right)  \text{ outputs }y\right]  $\ and
$q:=\Pr\left[  M\left(  \sigma\right)  \text{ outputs }y\right]  $, and assume
without loss of generality that $p\geq q$ and $p>0$. \ Also, fix an $\alpha
$-gentle implementation of\ $M$. \ Let $\rho_{y}$\ and $\sigma_{y}$\ be the
post-measurement states for $\rho$ and $\sigma$ respectively, if the outcome
of applying $M$ is $y$. \ Now consider the mixed state $\xi:=\frac{\rho
+\sigma}{2}$. \ \textit{Its} post-measurement state is%
\[
\xi_{y}=\frac{p\rho_{y}+q\sigma_{y}}{p+q}.
\]
So let $\delta:=\frac{p}{p+q}-\frac{1}{2}$. \ Then%
\begin{align*}
\xi-\xi_{y}  &  =\frac{1}{2}\rho+\frac{1}{2}\sigma-\frac{p}{p+q}\rho_{y}%
-\frac{q}{p+q}\sigma_{y}\\
&  =\frac{p}{p+q}\left(  \rho-\rho_{y}\right)  +\frac{q}{p+q}\left(
\sigma-\sigma_{y}\right)  -\delta\left(  \rho-\sigma\right)  .
\end{align*}
So by the triangle inequality,%
\begin{align*}
\delta\left\Vert \rho-\sigma\right\Vert _{\operatorname*{tr}}  &
\leq\left\Vert \xi-\xi_{y}\right\Vert _{\operatorname*{tr}}+\frac{p}%
{p+q}\left\Vert \rho-\rho_{y}\right\Vert _{\operatorname*{tr}}+\frac{q}%
{p+q}\left\Vert \sigma-\sigma_{y}\right\Vert _{\operatorname*{tr}}\\
&  \leq2\alpha,
\end{align*}
since $\left\Vert \xi-\xi_{y}\right\Vert _{\operatorname*{tr}}$\ and
$\left\Vert \rho-\rho_{y}\right\Vert _{\operatorname*{tr}}$\ and $\left\Vert
\sigma-\sigma_{y}\right\Vert _{\operatorname*{tr}}$\ are all at most $\alpha
$\ by our gentleness assumption. \ Furthermore, by assumption, $\left\Vert
\rho-\sigma\right\Vert _{\operatorname*{tr}}=1$. \ Thus we simply get
$\delta\leq2\alpha$. \ Or%
\[
\frac{\Pr\left[  M\left(  \rho\right)  \text{ outputs }y\right]  }{\Pr\left[
M\left(  \sigma\right)  \text{ outputs }y\right]  }=\frac{p}{q}=\frac{\frac
{1}{2}+\delta}{\frac{1}{2}-\delta}\leq\frac{1+4\alpha}{1-4\alpha}.
\]
And by Proposition \ref{eigprop}, if the above holds for perfectly
distinguishable states $\rho,\sigma$\ (so in particular, for orthogonal pure
states), then it holds for all $\rho,\sigma$ as well. \ Hence $M$ is
$\ln\left(  \frac{1+4\alpha}{1-4\alpha}\right)  $-trivial.
\end{proof}

An immediate corollary of Lemma \ref{mixtrick} is this:

\begin{corollary}
\label{gentledpall}If $M$ is $\alpha$-gentle on all states, then $M$\ is
$\ln\left(  \frac{1+4\alpha}{1-4\alpha}\right)  $-DP\ on all states.
\end{corollary}

Indeed, since the reasoning applied independently to each measurement outcome
$y$, we get the following stronger conclusion, which will be useful when we
analyze shadow tomography:

\begin{corollary}
\label{gentledpall2}If $M$ is $\left(  \alpha,\delta\right)  $-gentle on all
states, then $M$\ is $\left(  \ln\left(  \frac{1+4\alpha}{1-4\alpha}\right)
,\delta\right)  $-DP\ on all states.
\end{corollary}

Notice that the central gambit in the proof of Lemma \ref{mixtrick}, namely
defining $\xi:=\frac{\rho+\sigma}{2}$, generally maps product states to
non-product states. \ It turns out that this is inherent: Lemma \ref{mixtrick}%
\ does \textit{not} have an analogue that assumes only gentleness on product
states. \ Or rather: if we assume only gentleness on product states, then we
can deduce DP (and will do so, in Lemma \ref{gentledpprod}), but will not be
able to deduce triviality. \ And this is to be expected, since there
\textit{are} nontrivial DP algorithms, and indeed our main result (Theorem
\ref{main}) shows that these algorithms lead to measurements that are gentle
on product states.

We next prove a converse to Lemma \ref{mixtrick}: that near-triviality implies gentleness.

\begin{lemma}
[Trivial$\Longrightarrow$Gentle]\label{trivgentle}Suppose $M$\ is
$\varepsilon$-trivial. \ Then $M$ is $\left(  e^{\varepsilon}-1\right)
$-gentle on all states---so in particular, $O\left(  \varepsilon\right)
$-gentle, provided (say) $\varepsilon\leq1$.
\end{lemma}

\begin{proof}
Again, let $E_{1}+\cdots+E_{k}=\mathbb{I}$\ be the POVM elements of $M$, and
recall that we can use any solutions to the equations $E_{i}=A_{i}^{\dagger
}A_{i}$\ to define the possible post-measurement states after $M$ is applied.
\ Without loss of generality, focus on $E=E_{1}$ and $A=A_{1}$.

Since $M$ is $\varepsilon$-trivial, all of $E$'s eigenvalues must be within an
$e^{\varepsilon}$ multiplicative factor of each other. \ Also, since $E$ is
Hermitian, we can diagonalize it as $U^{\dagger}DU$, where $U$ is unitary and
$D$ is a diagonal matrix of $E$'s eigenvalues. \ Let's make the choice
$A:=U^{\dagger}\sqrt{D}U$. \ Then for some constant $0<c<1$, we can write
$\sqrt{D}$ as $c\left(  I+\delta V\right)  $, where $\delta\leq
1-e^{-\varepsilon/2}$, and $V$ is a diagonal matrix whose entries are all at
most $1$ in absolute value.

Let $\left\vert \psi\right\rangle $\ be a pure state to which $M$ is applied,
and assume $\left\langle \psi|E|\psi\right\rangle >0$. \ Then conditioning on
outcome $E$ leads to the post-measurement state%
\[
\frac{U^{\dagger}\sqrt{D}U\left\vert \psi\right\rangle }{\left\Vert
U^{\dagger}\sqrt{D}U\left\vert \psi\right\rangle \right\Vert }.
\]
Therefore the post-measurement state is%
\[
\frac{cU^{\dagger}\left(  I+\delta V\right)  U\left\vert \psi\right\rangle
}{\left\Vert cU^{\dagger}\left(  I+\delta V\right)  U\left\vert \psi
\right\rangle \right\Vert }=\frac{\left\vert \psi\right\rangle +\delta
U^{\dagger}VU\left\vert \psi\right\rangle }{\left\Vert \left\vert
\psi\right\rangle +\delta U^{\dagger}VU\left\vert \psi\right\rangle
\right\Vert }%
\]
By Proposition \ref{puretd}, the trace distance between this state and
$\left\vert \psi\right\rangle $\ is at most the Euclidean distance, which in
turn is at most
\[
\frac{1+\delta}{1-\delta}-1\leq\frac{2-e^{-\varepsilon/2}}{e^{-\varepsilon/2}%
}-1\leq e^{\varepsilon}-1.
\]
Thus, we've given an implementation of $M$ that is $\left(  e^{\varepsilon
}-1\right)  $-gentle on pure states. \ By Proposition \ref{mixedgentle}, this
implies that $M$ is $\left(  e^{\varepsilon}-1\right)  $-gentle on mixed
states as well.
\end{proof}

Finally, we prove that if a measurement is $\varepsilon$-DP \textit{for
sufficiently small} $\varepsilon$, then it's nearly trivial.

\begin{proposition}
[Sufficiently DP Is Trivial]\label{dptriv}If $M$ is $\varepsilon$-DP on all
states,\ then $M$ is $2\varepsilon n$-trivial on all states.
\end{proposition}

\begin{proof}
Let $\rho,\sigma$ be any mixed states on $n$ registers. \ Also, let $S_{i}%
$\ be a superoperator that simply swaps out the $i^{th}$\ register for some
fixed state---say the maximally mixed state $\mathbb{I}/d$, if the registers
are $d$-dimensional.\footnote{Or if we preferred unitary transformations, we
could also achieve the same effect by (for example) applying a Haar-random
unitary $U$ to the $i^{th}$\ register, and then appealing to convexity.}
\ Then by applying all $n$ of the $S_{i}$'s to $\rho$\ or $\sigma$, one at a
time, we can map the entire input state to $\mathbb{I}/d^{n}$. \ Thus, for any
output possible $y$\ of $M$, if we repeatedly invoke the fact that $M$ is
$\varepsilon$-DP, once for each $S_{i}$,\ we find that%
\[
\Pr\left[  M\left(  \rho\right)  \text{ outputs }y\right]  \leq e^{\varepsilon
n}\Pr\left[  M\left(  \mathbb{I}/d^{n}\right)  \text{ outputs }y\right]  .
\]
Likewise,%
\[
\Pr\left[  M\left(  \sigma\right)  \text{ outputs }y\right]  \geq
e^{-\varepsilon n}\Pr\left[  M\left(  \mathbb{I}/d^{n}\right)  \text{ outputs
}y\right]  .
\]
Hence%
\[
\Pr\left[  M\left(  \rho\right)  \text{ outputs }y\right]  \leq
e^{2\varepsilon n}\Pr\left[  M\left(  \sigma\right)  \text{ outputs }y\right]
.
\]

\end{proof}

One can show, by a similar argument, that if $M$\ is $\varepsilon$-DP on
product states,\ then $M$ is $2\varepsilon n$-trivial on product states.
\ Again, though, this is only interesting in the regime $\varepsilon\ll
\frac{1}{n}$, whereas our results in Section \ref{MAIN}\ will be able to
handle measurements that are $\varepsilon$-DP on product states\ for
$\varepsilon$\ up to about $\frac{1}{\sqrt{n}}$.

Combining Lemma \ref{mixtrick}, Lemma \ref{dptriv}, and Lemma \ref{trivgentle}
now completes the proof of Theorem \ref{warmup}.

Again, the problem with Theorem \ref{warmup} is that, while it relates the
privacy of a measurement $M$\ to its gentleness, it does so only as an
\textquotedblleft accidental byproduct\textquotedblright\ of showing that
sufficiently private and sufficiently gentle measurements are both nearly
trivial. \ To get a more interesting connection between privacy and
gentleness, we'll need to restrict our attention to product states, as our
main result (Theorem \ref{main})\ does.

\section{Proof of Main Result\label{MAIN}}

In this section we prove Theorem \ref{main}, the two-way connection between
gentleness and differential privacy on product states. \ Unlike Theorem
\ref{warmup}, this connection will work even for measurements that are very
far from trivial.

\subsection{Gentleness Implies DP on Product States}

We'll start by proving the \textquotedblleft easy\textquotedblright%
\ direction: that gentleness on product states implies differential privacy on
product states. \ For this, we can reapply Lemma \ref{mixtrick} from the
previous section.

\begin{lemma}
[Gentleness$\Longrightarrow$DP on Product States]\label{gentledpprod}If $M$ is
$\alpha$-gentle on product states, then $M$ is $\ln\left(  \frac{1+4\alpha
}{1-4\alpha}\right)  $-DP on product states as well---so in particular,
$O\left(  \alpha\right)  $-DP on product states, provided $\alpha\leq\frac
{1}{4.01}$. \ Likewise, if $M$ is $\left(  \alpha,\delta\right)  $-gentle on
product states then $M$ is $\left(  \ln\left(  \frac{1+4\alpha}{1-4\alpha
}\right)  ,\delta\right)  $-DP on product states.
\end{lemma}

\begin{proof}
Let $\rho=\rho_{1}\otimes\cdots\otimes\rho_{n}$\ and $\sigma=\sigma_{1}%
\otimes\cdots\otimes\sigma_{n}$\ be two product states that differ only on the
$i^{th}$\ register. \ Also, fix an implementation of $M$ that is $\alpha
$-gentle on product states. \ Then for any outcome $y$, let $\rho_{y}$\ and
$\sigma_{y}$\ be the post-measurement states for\ $\rho$\ and $\sigma
$\ respectively assuming that $M$ returns outcome $y$, and let $\rho_{y,i}%
$\ and $\sigma_{y,i}$\ be the restrictions (i.e., partial traces) of $\rho
_{y}$\ and $\sigma_{y}$\ respectively to the $i^{th}$\ register. \ Then by
Proposition \ref{contractive}, together with the assumption of $\alpha
$-gentleness, we have%
\[
\left\Vert \rho_{y,i}-\rho_{i}\right\Vert _{\operatorname*{tr}}\leq\left\Vert
\rho_{y}-\rho\right\Vert _{\operatorname*{tr}}\leq\alpha,
\]
and likewise%
\[
\left\Vert \sigma_{y,i}-\sigma_{i}\right\Vert _{\operatorname*{tr}}%
\leq\left\Vert \sigma_{y}-\sigma\right\Vert _{\operatorname*{tr}}\leq\alpha.
\]
But now we can apply Lemma \ref{mixtrick}---which implies that, if we think of
$M$ as acting on the $i^{th}$\ register only, with the other $n-1$\ registers
held fixed, then $M$\ must be $\ln\left(  \frac{1+4\alpha}{1-4\alpha}\right)
$-trivial. \ Moreover, the preceding statement holds for \textit{all} $i$, and
\textit{all} settings of the other $n-1$\ registers. \ But that's simply
another way of saying that $M$\ is $\ln\left(  \frac{1+4\alpha}{1-4\alpha
}\right)  $-DP on product states.

The last part follows simply because the argument applies for each possible
output $y$\ independently.
\end{proof}

This proves part (1) of Theorem \ref{main}.

Note that as $\alpha$\ approaches $\frac{1}{4}$, the bound on the DP parameter
diverges. \ Certainly the DP parameter needs to diverge as $\alpha
$\ approaches $\frac{1}{2}$, since (for example) measuring a single qubit in
the $\left\{  \left\vert 0\right\rangle ,\left\vert 1\right\rangle \right\}  $
basis, outputting the result, and then replacing the qubit by the maximally
mixed state is $\frac{1}{2}$-gentle but preserves no privacy whatsoever. \ We
leave it as an open problem to close the gap between $\frac{1}{4}$\ and
$\frac{1}{2}$.

\subsection{DP Implies Gentleness On Product States}

For the other direction, we'll proceed in stages. \ We'll start by providing
that $\varepsilon$-DP implies $O\left(  \varepsilon\sqrt{n}\right)
$-gentleness for classical product distributions. \ Later we'll extend this
result to the quantum setting.

We'll need a claim, originally proved in \cite{drv}, that's found many uses in
classical DP.

\begin{claim}
[\cite{drv}]\label{vadhanclaim}Suppose two probability distributions,
$\mathcal{D}=\left(  p_{x}\right)  _{x\in\left[  d\right]  }$\ and
$\mathcal{D}^{\prime}=\left(  q_{x}\right)  _{x\in\left[  d\right]  }$,
satisfy%
\[
\left\vert \ln\frac{p_{x}}{q_{x}}\right\vert \leq\varepsilon
\]
for all $x\in\left[  d\right]  $. \ Then the KL-divergence,%
\[
\operatorname{KL}\left(  \mathcal{D},\mathcal{D}^{\prime}\right)  =\sum
_{x=1}^{d}p_{x}\ln\frac{p_{x}}{q_{x}},
\]
satisfies $\operatorname{KL}\left(  \mathcal{D},\mathcal{D}^{\prime}\right)
\leq2\varepsilon^{2}$.
\end{claim}

We can now prove a classical \textquotedblleft DP implies
gentleness\textquotedblright\ result. \ As noted previously, similar results
follow from the work of Dwork et al.\ \cite{dwork:adaptive2} and Rogers et
al.\ \cite{rrst} (phrased in terms of the so-called \textquotedblleft max
information\textquotedblright), but we provide a self-contained proof.

\begin{lemma}
[Classical DP$\Longrightarrow$Gentleness]\label{dpgentleclassical}Let $A$ be a
classical $\varepsilon$-DP algorithm, and let $\mathcal{D}$ be a product
distribution over databases $X$. \ Then for all possible outputs $y$\ of $A$,
the posterior distribution $\mathcal{D}_{y}$\ satisfies $\left\Vert
\mathcal{D}_{y}-\mathcal{D}\right\Vert \leq2\varepsilon\sqrt{n}$, and indeed
the stronger bound $\operatorname{KL}\left(  \mathcal{D}_{y},\mathcal{D}%
\right)  \leq2\varepsilon^{2}n$.
\end{lemma}

\begin{proof}
Fix any output $y$. \ We want to compare the prior distribution $\mathcal{D}%
=\mathcal{D}_{1}\times\cdots\times\mathcal{D}_{n}$\ over databases $X$ to the
posterior distribution $\mathcal{D}_{y}$, which is obtained by conditioning on
the event $A\left(  X\right)  =y$. \ To do this, consider a process wherein we
draw a database $X=\left(  x_{1},\ldots,x_{n}\right)  $\ from $\mathcal{D}%
_{y}$, by first drawing $x_{1}$\ from the marginal distribution over the first
entry conditioned on $A\left(  X\right)  =y$, then drawing $x_{2}$\ from the
marginal distribution over the second entry conditioned on $A\left(  X\right)
=y$\ and on $x_{1}$, and so on up to $x_{n}$.

Let's call the $i^{th}$ distribution above $\mathcal{T}_{i}$; note that
$\mathcal{T}_{i}$ depends both on $y$ and on $x_{1},\ldots,x_{i-1}$. \ These
are distributions over $\left[  d\right]  $, our \textquotedblleft data
universe.\textquotedblright\ \ We claim that, for every possible value
$v\in\left[  d\right]  $\ for $x_{i}$, the log-ratio between $v$'s
probabilities under $\mathcal{T}_{i}$\ and under $\mathcal{D}_{i}$\ must be
upper-bounded in magnitude by $\varepsilon$. \ To show this, let $W:=\left(
x_{1},\ldots,x_{i-1}\right)  $ and $Z:=\left(  x_{i+1},\ldots,x_{n}\right)  $.
\ Then%
\begin{align*}
\frac{\Pr_{\mathcal{T}_{i}}\left[  v\right]  }{\Pr_{\mathcal{D}_{i}}\left[
v\right]  }  &  =\frac{\Pr\left[  v~|~y,W\right]  }{\Pr\left[  v\right]  }\\
&  =\frac{\Pr\left[  v~|~y,W\right]  }{\Pr\left[  v~|~W\right]  }\\
&  =\frac{\Pr\left[  y~|~W,v\right]  }{\Pr\left[  y~|~W\right]  }\\
&  =\frac{\sum_{Z}\Pr\left[  Z~|~W,v\right]  \Pr\left[  y~|~W,v,Z\right]
}{\sum_{Z}\Pr\left[  Z~|~W\right]  \Pr\left[  y~|~W,Z\right]  }\\
&  =\frac{\sum_{Z}\Pr\left[  Z\right]  \Pr\left[  y~|~W,v,Z\right]  }{\sum
_{Z}\Pr\left[  Z\right]  \Pr\left[  y~|~W,Z\right]  }.
\end{align*}
Here the second and last lines used the assumption that $\mathcal{D}%
=\mathcal{D}_{1}\times\cdots\times\mathcal{D}_{n}$ is a product distribution.
\ Also, by differential privacy,%
\[
e^{-\varepsilon}\leq\frac{\Pr\left[  y~|~W,v,Z\right]  }{\Pr\left[
y~|~W,Z\right]  }\leq e^{\varepsilon}%
\]
for all $y,W,v,Z$. \ Therefore convex combinations of the above probabilities
are also within an\ $e^{\varepsilon}$\ multiplicative factor of one another,
so%
\[
\left\vert \ln\frac{\Pr_{\mathcal{T}_{i}}\left[  v\right]  }{\Pr
_{\mathcal{D}_{i}}\left[  v\right]  }\right\vert \leq\varepsilon.
\]

By Claim \ref{vadhanclaim}, this means that the \textit{expected} log-ratio
between $\Pr_{\mathcal{T}_{i}}\left[  v\right]  $\ and $\Pr_{\mathcal{D}_{i}%
}\left[  v\right]  $, with respect to $x_{i}$\ drawn from $\mathcal{T}_{i}$,
is upper-bounded by $2\varepsilon^{2}$:%
\[
\operatorname{KL}\left(  \mathcal{T}_{i},\mathcal{D}_{i}\right)
=\operatorname{E}_{v\sim\mathcal{T}_{i}}\left[  \ln\frac{\Pr_{\mathcal{T}_{i}%
}\left[  v\right]  }{\Pr_{\mathcal{D}_{i}}\left[  v\right]  }\right]
\leq2\varepsilon^{2}.
\]
Furthermore, the expected sum of the log-ratios---i.e., the KL-divergence
between $\mathcal{D}_{y}$\ and $\mathcal{D}$\ themselves---is just the sum of
the expected log-ratios:%
\[
\operatorname{KL}\left(  \mathcal{D}_{y},\mathcal{D}\right)  =\sum_{i=1}%
^{n}\operatorname{E}\left[  \operatorname{KL}\left(  \mathcal{T}%
_{i},\mathcal{D}_{i}\right)  \right]  \leq2\varepsilon^{2}n,
\]
where the expectations here are over the choices for the $\mathcal{T}_{i}$'s
(which, however, are irrelevant to the upper bound). \ So by Pinsker's
inequality (Proposition \ref{pinsker}),%
\[
\left\Vert \mathcal{D}_{y}-\mathcal{D}\right\Vert \leq\sqrt{2\operatorname{KL}%
\left(  \mathcal{D}_{y},\mathcal{D}\right)  }\leq2\varepsilon\sqrt{n}.
\]

\end{proof}

Having shown that $\varepsilon$-DP implies $O\left(  \varepsilon\sqrt
{n}\right)  $-gentleness for classical product distributions, we now begin the
task of extending the result to quantum product states.

\begin{lemma}
\label{dpgentleprojective}Suppose the measurement $M$ is $\varepsilon$-DP on
product states, and is a product-of-projectives (i.e., consists of a classical
algorithm applied to the outcomes of nonadaptive projective measurements on
the $n$ registers). \ Then $M$ is $O\left(  \varepsilon\sqrt{n}\right)
$-gentle on product states.
\end{lemma}

\begin{proof}
By Proposition \ref{mixedgentle}, it suffices to give an implementation of $M$
that is $O\left(  \varepsilon\sqrt{n}\right)  $-gentle on \textit{pure}
product states. \ Thus, let%
\[
\left\vert \psi\right\rangle =\left\vert \psi_{1}\right\rangle \otimes
\cdots\otimes\left\vert \psi_{n}\right\rangle =\sum_{X\in\left[  d\right]
^{n}}\alpha_{X}\left\vert X\right\rangle
\]
be a pure product state on $n$ registers. \ By applying suitable local
unitaries, we can assume without loss of generality that $M$ simply measures
each $\left\vert \psi_{i}\right\rangle $\ in the computational basis,
obtaining the string $\left\vert X\right\rangle =\left\vert x_{1}\right\rangle
\cdots\left\vert x_{n}\right\rangle $\ with probability $\left\vert \alpha
_{X}\right\vert ^{2}$. \ It then outputs a sample from some probability
distribution $\mathcal{M}_{X}$,\ depending on $X$, over the possible outputs
$y$.\ \ We need to show how to sample from $\mathcal{M}_{X}$\ in an $O\left(
\varepsilon\sqrt{n}\right)  $-gentle manner.

Our implementation is as follows: first map the state $\left\vert
\psi\right\rangle $\ to%
\[
\sum_{X\in\left[  d\right]  ^{n}:\Pr\left[  X\right]  >0,~y}\alpha
_{X}\left\vert X\right\rangle \sqrt{\Pr\left[  y|X\right]  }\left\vert
y\right\rangle .
\]
Note that, as long as we do not care about computational complexity, the above
mapping can always be implemented \textit{somehow}, although implementing it
efficiently requires an efficient algorithm for \textquotedblleft
QSampling\textquotedblright\ the probability distributions $\mathcal{M}_{X}$.
\ Next, measure the $\left\vert y\right\rangle $ register in the computational
basis, and condition on getting some particular result $y$.

Then by the rules of quantum mechanics and Bayes' rule, the state of the first
register is just%
\begin{align*}
\left\vert \psi_{y}\right\rangle :=  &  \frac{\sum_{X\in\left[  d\right]
^{n}:\Pr\left[  X\right]  >0}\alpha_{X}\sqrt{\Pr\left[  y|X\right]
}\left\vert X\right\rangle }{\sqrt{\sum_{X\in\left[  d\right]  ^{n}:\Pr\left[
X\right]  >0}\left\vert \alpha_{X}\right\vert ^{2}\Pr\left[  y|X\right]  }}\\
&  =\frac{\sum_{X\in\left[  d\right]  ^{n}:\Pr\left[  X\right]  >0}\alpha
_{X}\sqrt{\Pr\left[  y|X\right]  }\left\vert X\right\rangle }{\sqrt{\Pr\left[
y\right]  }}\\
&  =\sum_{X\in\left[  d\right]  ^{n}:\Pr\left[  X\right]  >0}\alpha_{X}%
\sqrt{\frac{\Pr\left[  y|X\right]  }{\Pr\left[  y\right]  }}\left\vert
X\right\rangle \\
&  =\sum_{X\in\left[  d\right]  ^{n}:\Pr\left[  X\right]  >0}\alpha_{X}%
\sqrt{\frac{\Pr\left[  X|y\right]  }{\Pr\left[  X\right]  }}\left\vert
X\right\rangle \\
&  =\sum_{X\in\left[  d\right]  ^{n}:\Pr\left[  X\right]  >0}\frac{\alpha_{X}%
}{\left\vert \alpha_{X}\right\vert }\sqrt{\Pr\left[  X|y\right]  }\left\vert
X\right\rangle .
\end{align*}
Let $\mathcal{D}$\ be the distribution over $X\in\left[  d\right]  ^{n}%
$\ defined by $\Pr_{\mathcal{D}}\left[  X\right]  =\left\vert \alpha
_{X}\right\vert ^{2}$; note that $\mathcal{D}$\ is a product distribution, to
which Lemma \ref{dpgentleclassical} applies. \ Also, let $\mathcal{D}_{y}$ be
$\mathcal{D}$\ conditioned on the event that $M$ outputs $y$. \ Then we see
above that $\left\vert \psi_{y}\right\rangle $ is a pure state that precisely
corresponds to $\mathcal{D}_{y}$---in the sense that, if we measure
$\left\vert \psi_{y}\right\rangle $ in the computational basis, we'll see a
sample from $\mathcal{D}_{y}$. \ The one complication is that $\left\vert
\psi_{y}\right\rangle $ has an additional set of degrees of freedom, namely
the unit-norm phases $\frac{\alpha_{X}}{\left\vert \alpha_{X}\right\vert }$.
\ However, even these phases go away when we calculate the inner products
$\left\langle \psi|\psi_{y}\right\rangle $ (which involve complex conjugates).
\ In more detail:%
\begin{align*}
\left\langle \psi|\psi_{y}\right\rangle  &  =\sum_{X\in\left[  d\right]
^{n}:\Pr\left[  X\right]  >0}\alpha_{X}^{\ast}\frac{\alpha_{X}}{\left\vert
\alpha_{X}\right\vert }\sqrt{\Pr\left[  X|y\right]  }\\
&  =\sum_{X\in\left[  d\right]  ^{n}}\left\vert \alpha_{X}\right\vert
\sqrt{\Pr\left[  X|y\right]  }\\
&  =\sum_{X\in\left[  d\right]  ^{n}}\sqrt{\Pr\left[  X\right]  \Pr\left[
X|y\right]  }\\
&  =1-H^{2}\left(  \mathcal{D},\mathcal{D}_{y}\right)  .
\end{align*}
Here $H^{2}\left(  \mathcal{D},\mathcal{D}_{y}\right)  $\ is the
\textit{squared Hellinger distance} between the probability distributions
$\mathcal{D}$\ and $\mathcal{D}_{y}$ (see Section \ref{PROB}). \ So in
particular, $\left\langle \psi|\psi_{y}\right\rangle $ strictly relates the
distributions $\mathcal{D}$ and $\mathcal{D}_{y}$, and has nothing further to
do with quantum mechanics.

We can now upper-bound the trace distance between $\left\vert \psi
\right\rangle $\ and $\left\vert \psi_{y}\right\rangle $---and hence, the
gentleness of $M$ on $\left\vert \psi\right\rangle $---by%
\begin{align*}
\left\Vert |\psi\rangle\langle\psi|-|\psi_{y}\rangle\langle\psi_{y}%
|\right\Vert _{\operatorname*{tr}}  &  =\sqrt{1-\left\vert \left\langle
\psi|\psi_{y}\right\rangle \right\vert ^{2}}\\
&  =\sqrt{1-\left(  1-H^{2}\left(  \mathcal{D},\mathcal{D}_{y}\right)
\right)  ^{2}}\\
&  =\sqrt{2H^{2}\left(  \mathcal{D},\mathcal{D}_{y}\right)  -H^{4}\left(
\mathcal{D},\mathcal{D}_{y}\right)  }\\
&  \leq\sqrt{2}H\left(  \mathcal{D},\mathcal{D}_{y}\right) \\
&  \leq\sqrt{2\operatorname{KL}\left(  \mathcal{D}_{y},\mathcal{D}\right)  }\\
&  \leq2\sqrt{2}\cdot\varepsilon\sqrt{n}.
\end{align*}
Here the first line used Proposition \ref{puretd}, the second-to-last line
used Proposition \ref{hkl}, and the last line used Lemma
\ref{dpgentleclassical}.
\end{proof}

Note that, if we'd upper-bounded the Hellinger distance $H\left(
\mathcal{D},\mathcal{D}_{y}\right)  $ by the square root of the variation
distance, $\left\Vert \mathcal{D}-\mathcal{D}_{y}\right\Vert =O\left(
\varepsilon\sqrt{n}\right)  $, we'd only get an upper bound of $O\left(
\sqrt{\varepsilon}n^{1/4}\right)  $, rather than the $O\left(  \varepsilon
\sqrt{n}\right)  $\ that we wanted. \ To avoid that loss, here we exploited
the fact that Lemma \ref{dpgentleclassical}\ upper-bounded the KL-divergence
rather than only the variation distance, and we also used Proposition
\ref{hkl}, which upper-bounds Hellinger distance directly in terms of
KL-divergence, bypassing variation distance.

We now prove the final lemma needed to complete the proof of Theorem
\ref{main}, by generalizing Lemma \ref{dpgentleprojective}\ from projective
measurements to POVMs.

\begin{lemma}
\label{dpgentlepovm}If $M$ is any product measurement that is $\varepsilon$-DP
on product states, then $M$ is $O\left(  \varepsilon\sqrt{n}\right)  $-gentle
on product states.
\end{lemma}

\begin{proof}
Again, by Proposition \ref{mixedgentle}, it suffices to restrict attention to
pure states. \ We will give a reduction to the situation already handled in
Lemma \ref{dpgentleprojective}. \ Suppose we start with the product state%
\[
\left\vert \psi\right\rangle :=\left\vert \psi_{1}\right\rangle \otimes
\cdots\otimes\left\vert \psi_{n}\right\rangle .
\]
Next we apply a POVM to each $\left\vert \psi_{i}\right\rangle $. \ This can
be modeled as follows: for each $i$, we apply a unitary transformation $U_{i}$
to $\left\vert \psi_{i}\right\rangle $ together with some ancilla qubits that
are initially in the state $\left\vert 0\cdots0\right\rangle $. \ This yields
a new state that we can write as%
\[
\left\vert \phi_{i}\right\rangle :=\sum_{x\in\left[  k\right]  }%
a_{x}\left\vert x\right\rangle \left\vert v_{ix}\right\rangle .
\]
Here $\left\vert x\right\rangle $ represents a classical computational basis
state that the POVM will measure, while $\left\vert v_{ix}\right\rangle $
represents \textquotedblleft garbage\textquotedblright: some normalized state
that depends only on $i$ and $x$, need not be in the computational basis, and
will not be measured.

So now we have%
\[
\left\vert \phi\right\rangle :=\left\vert \phi_{1}\right\rangle \otimes
\cdots\otimes\left\vert \phi_{n}\right\rangle .
\]
Next, we apply our classical algorithm to the basis states $\left\vert
x_{1}\right\rangle ,\ldots,\left\vert x_{n}\right\rangle $, and then we
condition on the algorithm outputting $y$. \ This yields a new state
$\left\vert \phi_{y}\right\rangle $. \ What can we say about the relation
between $\left\vert \phi\right\rangle $ and $\left\vert \phi_{y}\right\rangle
$?

Let's reorganize $\left\vert \phi\right\rangle $ by collecting $\left(
x_{1},\ldots,x_{n}\right)  $\ into a single register that we'll call $X$, and
also collecting all the $\left\vert v_{ij}\right\rangle $'s into a single
register that we'll call $\left\vert v_{X}\right\rangle $. \ We then have:%
\[
\left\vert \phi\right\rangle =\sum_{X\in\left[  k\right]  ^{n}}b_{X}\left\vert
X\right\rangle \left\vert v_{X}\right\rangle ,
\]
where $b_{X}=a_{x_{1}}\cdots a_{x_{n}}$, and $\left\vert b_{X}\right\vert
^{2}=\Pr\left[  X\right]  $ is just the probability of $X$ from the
perspective of the classical algorithm. \ By exactly the same reasoning as in
the proof of Lemma \ref{dpgentleprojective}, it follows that%
\[
\left\vert \phi_{y}\right\rangle =\sum_{X\in\left[  k\right]  ^{n}:\Pr\left[
X\right]  >0}\frac{b_{X}}{\left\vert b_{X}\right\vert }\sqrt{\Pr\left[
X|y\right]  }\left\vert X\right\rangle \left\vert v_{X}\right\rangle .
\]
Therefore,%
\[
\left\langle \phi|\phi_{y}\right\rangle =\sum_{X\in\left[  k\right]  ^{n}%
:\Pr\left[  X\right]  >0}b_{X}^{\ast}\frac{b_{X}}{\left\vert b_{X}\right\vert
}\sqrt{\Pr\left[  X|y\right]  }=\sum_{X\in\left[  k\right]  ^{n}}\sqrt
{\Pr\left[  X\right]  \Pr\left[  X|y\right]  }.
\]

So now we have exactly the same expression for the inner product that we had
in the proof of Lemma \ref{dpgentleprojective}. \ So we can use the same
argument to lower-bound the inner product by $1-\varepsilon^{2}n$, and to
upper-bound both the Hellinger distance and the trace distance between
$\left\vert \phi\right\rangle $\ and $\left\vert \phi_{y}\right\rangle $\ by
$O\left(  \varepsilon\sqrt{n}\right)  $.

Finally, recall that $\left\vert \phi\right\rangle $ was obtained by applying
a unitary transformation $U=U_{1}\otimes\cdots\otimes U_{n}$\ to $\left\vert
\psi\right\rangle $ together with some $\left\vert 0\right\rangle $ ancilla
qubits. \ Since inner products are unitarily invariant, this means that
$U^{\dagger}\left\vert \phi_{y}\right\rangle $ also has trace distance at most
$O\left(  \varepsilon\sqrt{n}\right)  $ from $U^{\dagger}\left\vert
\phi\right\rangle =\left\vert \psi\right\rangle \left\vert 0\cdots
0\right\rangle $. \ Hence we've implemented $M$ in an $O\left(  \varepsilon
\sqrt{n}\right)  $-gentle manner.
\end{proof}

Intuitively, what's going on is that the garbage register, $\left\vert
v_{X}\right\rangle $, is completely inert: it's there, but it has no effect on
the inner product.

Combining Lemma \ref{gentledpprod}\ with Lemma \ref{dpgentlepovm} now
completes the proof of Theorem \ref{main}.

\section{Separating Examples\label{EXAMPLES}}

In this section, we prove that the relationships between DP and gentleness
notions proved in the preceding two sections are essentially tight, by giving
examples of measurements that exhibit their optimality.

\subsection{Gentleness to DP\label{GDSEP}}

For all $\beta>0$, let $R_{\beta}$\ be the \textquotedblleft randomized
response\textquotedblright\ algorithm, which for each $i\in\left[  n\right]  $
separately, applies the POVM defined by the matrices%
\[
E_{\operatorname*{reject}}=\left(
\begin{array}
[c]{cc}%
\frac{1}{2}+\beta & 0\\
0 & \frac{1}{2}-\beta
\end{array}
\right)  ,~~~E_{\operatorname*{accept}}=\left(
\begin{array}
[c]{cc}%
\frac{1}{2}-\beta & 0\\
0 & \frac{1}{2}+\beta
\end{array}
\right)
\]
to the $i^{th}$ qubit and returns the result. \ In other words, the output of
$R_{\beta}$\ is an $n$-bit string, whose $i^{th}$\ bit has a\ bias of $\beta
$\ toward the value of the $i^{th}$\ qubit in the $\left\{  \left\vert
0\right\rangle ,\left\vert 1\right\rangle \right\}  $\ basis. \ The following
is immediate:

\begin{proposition}
\label{rbeta}$R_{\beta}$ is $\varepsilon$-DP for $\varepsilon=\ln\left(
\frac{1+2\beta}{1-2\beta}\right)  $, which is $O\left(  \beta\right)  $\ for
$\beta\leq\frac{1}{2.01}$, and is not $\varepsilon^{\prime}$-DP\ for any
$\varepsilon^{\prime}<\varepsilon$.
\end{proposition}

\begin{proof}
Flipping the $i^{th}$\ input bit can at worst change the probability that the
$i^{th}$\ output bit assumes some value from $\frac{1}{2}-\beta$\ to $\frac
{1}{2}+\beta$\ (or vice versa),\ while leaving the other $n-1$\ output bits unchanged.
\end{proof}

We also have:

\begin{proposition}
\label{rbeta2}Suppose $n=1$\ (i.e., there is just one qubit).
\ Then$\ R_{\beta}$ is $2\beta$-gentle.
\end{proposition}

\begin{proof}
Given a qubit in state%
\[
\rho=\left(
\begin{array}
[c]{cc}%
a & b\\
b^{\ast} & c
\end{array}
\right)  ,
\]
here is one way to implement $R_{\beta}$: with probability $1-2\beta$, return
$0$ or $1$\ with equal probabilities. \ With probability $2\beta$, measure
$\rho$ in the $\left\{  \left\vert 0\right\rangle ,\left\vert 1\right\rangle
\right\}  $\ basis and return the result. \ Suppose without loss of generality
that $a>0$\ and we condition on the output being $\left\vert 0\right\rangle $.
\ Then the post-measurement state is%
\[
\sigma:=\frac{\left(  1-2\beta\right)  \rho+2\beta a|0\rangle\langle
0|}{1-2\beta+2\beta a}.
\]
The trace distance, $\left\Vert \rho-\sigma\right\Vert _{\operatorname*{tr}}$,
can thus be calculated explicitly as%
\begin{align*}
\frac{2\beta a}{1-2\beta+2\beta a}\left\Vert \rho-|0\rangle\langle
0|\right\Vert _{\operatorname*{tr}}  &  =\frac{2\beta a\sqrt{\left\vert
b\right\vert ^{2}+c^{2}}}{1-2\beta+2\beta a}\\
&  \leq2\beta\sqrt{\left\vert b\right\vert ^{2}+c^{2}}\\
&  \leq2\beta.
\end{align*}

\end{proof}

Combining Propositions \ref{rbeta}\ and \ref{rbeta2}, we get the following corollary:

\begin{corollary}
\label{lncor}For all $\alpha\in\left(  0,1\right)  $, there exists a
measurement that is $\alpha$-gentle on arbitrary states, but not $\varepsilon
$-DP\ for any $\varepsilon<\ln\left(  \frac{1+\alpha}{1-\alpha}\right)  $,
even on product states.
\end{corollary}

\begin{proof}
Consider $R_{\alpha/2}$\ applied to the first qubit only.
\end{proof}

This shows that Corollary \ref{gentledpall} and Lemma \ref{gentledpprod} are
both tight, up to the factor of $4$ in front of the $\alpha$.

\subsection{DP to Gentleness\label{DGSEP}}

We now prove that, when we showed that $\varepsilon$-DP on arbitrary states
implies $O\left(  \varepsilon n\right)  $-gentleness on arbitrary states
(Proposition \ref{dptriv}), and that $\varepsilon$-DP on product states
implies $O\left(  \varepsilon\sqrt{n}\right)  $-gentleness on product states
for product measurements (Theorem \ref{main}), the $n$ and $\sqrt{n}$\ factors
were both asymptotically tight.

Recall the measurement $L_{\sigma}$\ from Sections \ref{GENTLE}\ and
\ref{CONNECTION}, which takes as input an $n$-qubit state and returns the
total Hamming weight, plus a Laplace noise term $\eta$\ of average magnitude
$\sigma$. \ We showed, in Proposition \ref{lsigma}, that $L_{\sigma}$ is
$\frac{1}{\sigma}$-DP---and moreover, on \textit{all} $n$-qubit states, not
merely on product states. \ By contrast, we now observe that $L_{\sigma}$\ is
far from gentle on arbitrary states:

\begin{proposition}
[Optimality of $n$ Factor]\label{notgentle}$L_{n/2}$ is not $\frac{1}{3}%
$-gentle on $n$-qubit states.
\end{proposition}

\begin{proof}
We consider $L_{\sigma}$\ applied to the mixture%
\[
\rho:=\frac{|0^{n}\rangle\langle0^{n}|+|1^{n}\rangle\langle1^{n}|}{2}.
\]
Note that the entire situation is classical, so the question of how
$L_{\sigma}$\ is implemented is irrelevant. \ Let the measurement outcome be
$y$; then%
\begin{align*}
\Pr\left[  y|0^{n}\right]   &  =\frac{1}{2\sigma}e^{-\left\vert y\right\vert
/\sigma},\\
\Pr\left[  y|1^{n}\right]   &  =\frac{1}{2\sigma}e^{-\left\vert y-n\right\vert
/\sigma}.
\end{align*}
So by Bayes' rule, the post-measurement state is%
\[
\rho_{y}=\frac{e^{-\left\vert y\right\vert /\sigma}|0^{n}\rangle\langle
0^{n}|+e^{-\left\vert y-n\right\vert /\sigma}|1^{n}\rangle\langle1^{n}%
|}{e^{-\left\vert y\right\vert /\sigma}+e^{-\left\vert y-n\right\vert /\sigma
}}.
\]
Suppose $y\leq0$. \ Then we can calculate:%
\begin{align*}
\left\Vert \rho_{y}-\rho\right\Vert _{\operatorname*{tr}}  &  =\left\vert
\frac{e^{-\left\vert y\right\vert /\sigma}}{e^{-\left\vert y\right\vert
/\sigma}+e^{-\left\vert y-n\right\vert /\sigma}}-\frac{1}{2}\right\vert \\
&  =\left\vert \frac{1}{1+e^{-n/\sigma}}-\frac{1}{2}\right\vert \\
&  =\frac{1}{2}\left(  \frac{1-e^{-n/\sigma}}{1+e^{-n/\sigma}}\right) \\
&  >\frac{1}{2}-e^{-n/\sigma}.
\end{align*}
If we now make the choice (say) $\sigma=\frac{n}{2}$, we find that this
exceeds $\frac{1}{3}$.
\end{proof}

It follows that, in going from DP on arbitrary states to gentleness on
arbitrary states, we need at least a factor of $n$ blowup in $\varepsilon$;
indeed this is true even for product-of-projectives measurements. \ Hence
Proposition \ref{dptriv}\ is essentially tight.

Likewise, in going from DP on product states to gentleness on product states,
we need at least a factor of $\sqrt{n}$\ blowup in $\varepsilon$, and this is
true even for product-of-projectives measurements. \ Hence Lemma
\ref{dpgentleprojective}\ is essentially tight. \ The example that shows this
is again $L_{\sigma}$, albeit this time with $\sigma=\sqrt{n}$:

\begin{proposition}
[Optimality of $\sqrt{n}$\ Factor]\label{notgentleprod}$L_{\sqrt{n}}$ is not
$\alpha$-gentle on $n$-qubit product states, for any $\alpha=o\left(
1\right)  $.
\end{proposition}

\begin{proof}
Let $\sigma=\sqrt{n}$, and consider $L_{\sigma}$ applied to the uniform
distribution $\mathbb{I}/2^{n}$. \ Again, since the entire situation is
classical, the question of how $L_{\sigma}$\ is implemented is irrelevant.
\ Let the measurement outcome be $y$; then by Bayes' rule, the
post-measurement state is%
\[
\rho_{y}=\frac{\sum_{X\in\left\{  0,1\right\}  ^{n}}e^{-\left\vert
y-\left\vert X\right\vert \right\vert /\sigma}|X\rangle\langle X|}{\sum
_{X\in\left\{  0,1\right\}  ^{n}}e^{-\left\vert y-\left\vert X\right\vert
\right\vert /\sigma}}.
\]
So suppose $y\leq0$, and assume without loss of generality that $n$ is odd.
\ Then we can calculate:%
\begin{align*}
\left\Vert \rho_{y}-\rho\right\Vert _{\operatorname*{tr}}  &  =\frac{1}{2}%
\sum_{X\in\left\{  0,1\right\}  ^{n}}\left\vert \frac{e^{-\left\vert
y-\left\vert X\right\vert \right\vert /\sigma}}{\sum_{Z\in\left\{
0,1\right\}  ^{n}}e^{-\left\vert y-\left\vert Z\right\vert \right\vert
/\sigma}}-\frac{1}{2^{n}}\right\vert \\
&  =\frac{1}{2}\sum_{k=0}^{n}\binom{n}{k}\left\vert \frac{e^{-\left\vert
y-k\right\vert /\sigma}}{\sum_{\ell=0}^{n}\binom{n}{\ell}e^{-\left\vert
y-\ell\right\vert /\sigma}}-\frac{1}{2^{n}}\right\vert \\
&  =\frac{1}{2}\sum_{k=0}^{n}\binom{n}{k}\left\vert \frac{e^{-k/\sigma}}%
{\sum_{\ell=0}^{n}\binom{n}{\ell}e^{-\ell/\sigma}}-\frac{1}{2^{n}}\right\vert
\\
&  \geq\frac{1}{2}\left\vert \sum_{k=0}^{\left(  n-1\right)  /2}\binom{n}%
{k}\frac{e^{-k/\sigma}}{\sum_{\ell=0}^{n}\binom{n}{\ell}e^{-\ell/\sigma}}%
-\sum_{k=\left(  n+1\right)  /2}^{n}\binom{n}{k}\frac{e^{-k/\sigma}}%
{\sum_{\ell=0}^{n}\binom{n}{\ell}e^{-\ell/\sigma}}\right\vert \\
&  =\frac{1}{2}\left\vert \sum_{k=0}^{\left(  n-1\right)  /2}\binom{n}{k}%
\frac{e^{-k/\sigma}-e^{-\left(  n-k\right)  /\sigma}}{\sum_{\ell=0}^{n}%
\binom{n}{\ell}e^{-\ell/\sigma}}\right\vert \\
&  \geq\frac{1}{2}\left\vert \sum_{k=0}^{n/2-\sqrt{n}}\binom{n}{k}%
\frac{e^{-k/\sigma}-e^{-\left(  n/2+\sqrt{n}\right)  /\sigma}}{\sum_{\ell
=0}^{n}\binom{n}{\ell}e^{-\ell/\sigma}}\right\vert \\
&  \geq\frac{1}{4}\left\vert \sum_{k=0}^{n/2-\sqrt{n}}\binom{n}{k}%
\frac{e^{-k/\sigma}}{\sum_{\ell=0}^{n}\binom{n}{\ell}e^{-\ell/\sigma}%
}\right\vert \\
&  \geq\frac{1}{4}\left\vert \sum_{k=0}^{n/2-\sqrt{n}}\binom{n}{k}\frac
{1}{2^{n}}\right\vert \\
&  =\Omega\left(  1\right)  .
\end{align*}

\end{proof}

\section{Shadow Tomography\label{SHADOW}}

Having developed the connection between DP and gentleness, we're now ready to
apply the connection to shadow tomography. \ First, in Section
\ref{ONLINE_LEARNING}, we review a recent algorithm of Aaronson et
al.\ \cite{achkn}\ for online learning of quantum states, which we'll need as
a central ingredient. \ Then, in Section \ref{ONLINESHADOW}, we present and
analyze our new \textit{Quantum Private Multiplicative Weights (QPMW)}
algorithm, which builds on the Private Multiplicative Weights (PMW) algorithm
of Hardt and Rothblum \cite{hardtrothblum}. \ QPMW proves Theorem
\ref{newshadow}: that is, it shows that it's possible to do shadow tomography
using only $O\left(  \left(  \log m\right)  ^{2}\left(  \log d\right)
^{2}/\varepsilon^{8}\right)  $\ copies of an unknown mixed state $\rho$, where
$m$\ is the number of known accept/reject measurements,\ $d$ is the dimension
of $\rho$, and $\varepsilon$\ is the accuracy with which we want to estimate
each measurement's acceptance probability---in a way that, moreover, is
\textit{online} (i.e., processes the measurements one at a time) and
$\varepsilon$-gentle (i.e., damages the copies of $\rho$\ by at most
$\varepsilon$\ in trace distance).

\subsection{Online Learning of Quantum States\label{ONLINE_LEARNING}}

Aaronson et al.\ \cite{achkn} recently defined and studied the problem of
\textit{online learning of quantum states}. \ Here we have an unknown
$d$-dimensional mixed state $\rho$, and a learner is presented with a sequence
$E_{1},E_{2},\ldots$ of two-outcome POVM measurements. \ For each measurement
$E_{t}$, the learner tries to anticipate $\operatorname{Tr}\left(  E_{t}%
\rho\right)  $, the probability that $E_{t}$\ accepts $\rho$, up to accuracy
$\pm\varepsilon$. \ Indeed, the learner maintains a \textquotedblleft
hypothesis state\textquotedblright\ $\sigma_{t}$, and on each measurement
$E_{t}$, if the hypothesis differs appreciably from the unknown state $\rho$
with respect to this measurement---that is,\ if%
\[
\left\vert \operatorname{Tr}\left(  E_{t}\rho\right)  -\operatorname{Tr}%
\left(  E_{t}\sigma_{t-1}\right)  \right\vert >\varepsilon
\]
---then we say that the learner was \textquotedblleft wrong,\textquotedblright%
\ and we allow it to \textit{update} its state by giving it an approximation
$b_{t}\in\left[  0,1\right]  $ to the correct answer, where (say) $\left\vert
\operatorname{Tr}\left(  E_{t}\rho\right)  -b_{t}\right\vert \leq
\frac{\varepsilon}{10}$. \ The learner's goal is to upper-bound the total
number of times that it's ever wrong, even assuming that the sequence of
$E_{t}$'s and $b_{t}$'s is chosen adaptively, by an adversary who sees the
learner's hypotheses.

Perhaps surprisingly, Aaronson et al.\ \cite{achkn} showed that the total
number of mistakes can be upper-bounded by $O\left(  \frac{\log d}%
{\varepsilon^{2}}\right)  $---so for example, only $O\left(  \frac
{n}{\varepsilon^{2}}\right)  $\ for a state of $n$ qubits (even though the
state space has dimension $2^{n}$).

We observe that the same bound holds even under a slight relaxation of the
update condition:\ namely, updates can also be triggered when the hypothesis
has error between $\frac{\varepsilon}{3}$ and $\varepsilon$. \ If an update is
triggered, then the learner again receives an $\frac{\varepsilon}{10}%
$-approximation to the correct answer.

\begin{theorem}
[{Variant: Online Learning of Quantum States \cite[Theorem 1]{achkn}}%
]\label{achknthm}There is an explicit procedure for online learning of quantum
states that makes at most $\ell\left(  d,\varepsilon\right)  =O\left(
\frac{\log d}{\varepsilon^{2}}\right)  $ updates, so long as updates never
occur when the hypothesis has error smaller than $\frac{\varepsilon}{3}$, and
updates always occur when the error is $\varepsilon$ or larger.
\end{theorem}

We emphasize that when the error is in the range $[\frac{\varepsilon}%
{3},\varepsilon)$, updates may or may not occur.

Aaronson et al.\ \cite{achkn} actually gave two explicit procedures that
achieve the above bound: one based on online convex optimization, the other on
matrix multiplicative weights. \ Both procedures use an amount of computation
per measurement that's polynomial in $d$.

In this work, however, we'll be able simply to use Theorem \ref{achknthm} as a
black box. \ We'll view an online learning procedure as specified by its
initialization procedure, which outputs an initial hypothesis state
$\sigma_{0}\leftarrow\mathbf{OnlineLearn}(d)$, and an update procedure used to
update the hypothesis state $\sigma_{t}\leftarrow\mathbf{OnlineUpdate}%
(\sigma_{t-1},b_{t})$.

\subsection{Online Shadow Tomography\label{ONLINESHADOW}}

Our Quantum Private Multiplicative Weights (QPMW) algorithm for gentle online
shadow tomography is presented in Figure \ref{fig:OST}.

\begin{figure}[h]
\setlength{\fboxsep}{2mm}
\par
\begin{center}
\begin{boxedminipage}{\textwidth}
{\bf Parameters:} Intended number of queries
$m\in\mathbb{N},$ gentleness and accuracy parameters $\alpha, \varepsilon, \delta > 0$ and noise magnitude $\mu > 0$ (set in the proof below, see Equation \ref{eq:noise_magnitude}).\\
\noindent {\bf Input:} $n$ and a product state $\rho = \rho_1 \otimes \cdots \otimes \rho_n$, where the $\rho_i$'s are $d$-dimensional mixed states.\\
\\
\noindent {\bf Algorithm:}\\
\noindent Initialize the online learner $\sigma_0 \leftarrow \mathbf{OnlineLearn}(d)$\\
\noindent In each round $t \leftarrow 1,2 \ldots, m$, when receiving
two-outcome measurement $E_t$, do the following:
\begin{enumerate}
\item \label{step:check}
Apply the two-outcome $\mathbf{CheckForUpdate}$ measurement:
\begin{enumerate}
\item Apply the Laplace measurement with noise magnitude $n \mu$ to the $n$ registers to compute (but {\em do not measure}) a noisy estimate $a_t$ of $\operatorname{Tr} (E_t \rho)$
\item Compute {\em and measure} the decision bit $u_t$, which is 1 if and only if $\left|{a_t} - \operatorname{Tr} (E_t \sigma_{t-1})\right| > \frac{\varepsilon}{2}$
\item Uncompute the noise and intermediate values computed in the above measurement
\end{enumerate}
\item \label{step:noupdate}
If $u_t=0$ (no update), then set $\sigma_t \leftarrow \sigma_{t-1}$ and output the answer $b_t \leftarrow \operatorname{Tr} (E_t \sigma_{t-1})$
\item \label{step:update}
Otherwise ($u_t=1$, i.e.\ an ``update round''):
\begin{enumerate}
\item If there have already been $\ell(d,\varepsilon)$ prior update rounds (see Theorem \ref{achknthm}), then abort
\item Apply the Laplace measurement with noise magnitude $n \mu$ to the $n$ registers to compute a noisy estimate $b_t$ of $\operatorname{Tr} (E_t \rho)$, uncompute the noise and intermediate values computed within this measurement
\item Run a round of online learning: $\sigma_t \leftarrow \mathbf{OnlineUpdate}(\sigma_{t-1}, b_t)$
\item output the answer $b_t$
\end{enumerate}
\end{enumerate}
\end{boxedminipage}
\end{center}
\par
\vspace{-3mm}\caption{QPMW Algorithm}%
\label{fig:OST}%
\end{figure}

\begin{theorem}
\label{ostthm}Let $\alpha,\beta,\varepsilon,\delta>0$ be gentleness and
accuracy parameters. \ There exists a setting for the noise magnitude $\mu$
for which the online shadow tomography algorithm presented in Figure
\ref{fig:OST} is $(\alpha,\delta)$-gentle. \ Moreover, given sufficiently many
copies $n$, where
\[
n=O\left(  \frac{\left(  \log^{2}m + \log\frac{1}{\delta}\right)  \cdot
\log^{2}d \cdot\log\frac{1}{\beta}}{\varepsilon^{6}\min\left\{  \alpha
,\varepsilon\right\}  ^{2}}\right)  ,
\]
the algorithm's error is bounded by $\varepsilon$ with probability at least
$1-\beta$\ over its coins and its measurements.
\end{theorem}

\begin{proof}
We first prove gentleness and then turn our attention to bounding the error
(the accuracy proof builds on the algorithm's gentleness).\bigskip

\textbf{Gentleness.} \ Note that we argue gentleness for \textit{any} product
state provided as input (i.e., for gentleness, we don't assume that the input
is $n$ copies of a single state). \ By Proposition \ref{mixedgentle}, it
suffices to consider the case where the input is a pure product state
$\left\vert \psi\right\rangle =\left\vert \psi_{1}\right\rangle \otimes
\cdots\otimes\left\vert \psi_{n}\right\rangle $. \ It is straightforward to
see that the update rounds are gentle: we run two DP measurements in each
update round, and their outcomes are gentle by Theorem \ref{main}. \ This is
stated below in Claim \ref{claim:update_gentle}. \ The non-update rounds are
certainly no less gentle than the update rounds (after all, we only run the
first measurement), but we expect to have a very large number of no-update
rounds, and so we need a much better bound. \ We obtain such a bound by
restricting our attention to the damage that can be caused by the conditioned
superoperator $\mathbf{CheckForUpdate}$, conditioned on the output being $0$
(no update). \ One important challenge is showing that the damage (conditioned
on this particular outcome) is tightly related to the probability of an
update. \ Thus, it will be highly unlikely for a sequence of rounds (even a
very long sequence!) to cause significant damage before it triggers an update.
\ The second challenge is bounding the damage that can be caused by a sequence
of \textit{conditioned} superoperators. This is done via a delicate accounting
argument, which relies on Lemma \ref{facepalmlemma}

We begin by fixing some notation. \ First, given a superoperator and a fixed
output $y$, we use the term \textit{conditioned superoperator} to refer to
running the superoperator conditioned on the output being $y$.\footnote{In the
terminology of Section \ref{MIXED}, a conditioned superoperator is a quantum
operation but where we normalize the output state.} \ The QPMW algorithm's
output in any run can be specified by $m^{\prime}\leq m$, the number of rounds
before an \textquotedblleft abort\textquotedblright\ (if any) occurs, and by a
vector $\bar{y}$ of outcomes, where for each $t\in\left\{  1,\ldots,m^{\prime
}\right\}  $, the outcome in round $t$ is $y_{t}\in\left[  0,1\right]
\cup\{\perp\}$. \ In no-update rounds the outcome is $\perp$, while in update
rounds the outcome $y_{t}=b_{t}$ is the noisy answer returned by the
algorithm. \ Note that $m^{\prime}$ and the vector $\bar{y}$ of outcomes
indeed specify all outputs of the algorithm. \ For an intermediate round
$t\in\lbrack1,m^{\prime}]$, we can also consider the vector $\bar{y}_{\leq t}$
of outcomes in the first $t$ rounds. \ Taking $\left\vert \psi\right\rangle $
to be the initial state of the algorithm, we take $\left\vert \psi_{\leq
t}\right\rangle $ to be the state after round $t$, conditioned on the outcomes
$\bar{y}_{\leq t}$ (and given the measurements $E_{1},\ldots,E_{t}$). \ The
initial state is thus $\left\vert \psi\right\rangle =\left\vert \psi_{\leq
0}\right\rangle $, and the final state is $\left\vert \psi_{\leq m^{\prime}%
}\right\rangle $.

Consider an execution of the algorithm at the beginning of the $t^{th}$ round.
\ The outcomes in previous rounds are given by $\bar{y}_{\leq t-1}$, which
determines the learned state $\sigma_{t-1}$. \ Let $E_{t}$ be the $t^{th}$
measurement. \ We define $\lambda_{t}$ to be the probability that the
$\mathbf{CheckForUpdate}\left(  E_{t},\left\vert \psi\right\rangle
,\sigma_{t-1}\right)  $ measurement returns $1$, i.e.\ the probability of an
update on the \textit{original} state $\left\vert \psi\right\rangle $.
\ Similarly, we take $\kappa_{t}$ to be the probability that
$\mathbf{CheckForUpdate}(E_{t},\left\vert \psi_{\leq t-1}\right\rangle
,\sigma_{t-1})$ returns $1$, i.e.\ the probability of an update on the real
state in the registers at the beginning of the $t^{th}$ round.

The following claims bound the damage that can occur if we run the $t^{th}$
round with a fresh copy of the original state in the registers.

\begin{claim}
\label{claim:update_gentle}Every round of the algorithm is an $O\left(
\frac{1}{\sqrt{n}\mu}\right)  $-gentle superoperator.
\end{claim}

\begin{claim}
\label{claim:noupdate_gentle} Take $n$ and $\mu$ to be set as in Equations
\ref{eq:DB_size} and \ref{eq:noise_magnitude}. \ Let $\left\vert
\psi_{\text{no}}\right\rangle $ be the state after we run
$\mathbf{CheckForUpdate}(E_{t},\left\vert \psi\right\rangle ,\sigma_{t-1})$,
and condition on the output $0$ (\textquotedblleft no
update,\textquotedblright\ which occurs with probability $1-\lambda_{t}$).
\ We have:
\[
\left\Vert |\psi\rangle\langle\psi|-|\psi_{\text{no}}\rangle\langle
\psi_{\text{no}}|\right\Vert _{\operatorname*{tr}}=O\left(  \frac{\lambda_{t}%
}{\sqrt{n}\mu}\right)  .
\]

\end{claim}

Claim \ref{claim:update_gentle} follows immediately from the differential
privacy of the Laplace measurement and from Theorem \ref{main}. \ We defer the
proof of Claim \ref{claim:noupdate_gentle}, which is technically involved and
lengthy (see below). We first show that, given this claim, the algorithm
(taken as a whole) is gentle.\bigskip

\textbf{Epoch superoperators.} \ For the analysis, we divide an execution of
the algorithm into epochs, where each epoch is comprised of one or more
rounds. \ The $k^{th}$ epoch begins in round $t_{k}$ (where $t_{1}=1$). \ The
$k^{th}$ epoch ends on the first round $t^{\prime}\geq t_{k}$ where one of the
following occurs:

\begin{enumerate}
\item[(1)] An update happens (or the epoch reaches the last round $m$).

\item[(2)] The probability of an update, if each round was run on the original
state, becomes too large:
\begin{equation}
\prod_{j=t_{k}}^{t^{\prime}}\left(  1-\lambda_{j}\right)  \leq\frac{1}{2}.
\label{condition:p}%
\end{equation}

\end{enumerate}

Naturally, the last epoch always ends on the last round $m^{\prime}$. \ The
crux of the gentleness analysis is bounding the damage done to the state
within any single epoch. \ A separate argument shows that the number of epochs
cannot be too large.

Viewing each epoch as a superoperator, it is specified by a list of
measurements $E_{t_{k}},E_{t_{k}+1},\ldots$ that would be chosen so long as no
updates occurred. \ Note that this list is indeed fixed: while the strategy
that chooses the actual measurements $E_{t_{k}},E_{t_{k}+1},\ldots$ can be
adaptive, it specifies a fixed sequence of measurements (known in advance)
that will be chosen so long as the outputs are \textquotedblleft no
update.\textquotedblright\ \ Let $t^{\prime\prime}\geq t_{k}$ be the first
round that meets Condition (\ref{condition:p}). \ The epoch processes the list
of measurements $E_{t_{k}},E_{t_{k}+1},\ldots,E_{t^{\prime\prime}}$ until an
update occurs (or the last measurement in this list is processed). \ Note that
$t^{\prime\prime}$ depends on the initial state $\left\vert \psi\right\rangle
$, but it is fixed in advance. \ Given the list of measurements $E_{t_{k}%
},E_{t_{k}+1},\ldots$, the output of the epoch superoperator is the length a
list of \textquotedblleft no update\textquotedblright\ decisions of length
$s\in\left[  0,t^{\prime\prime}-t_{k}\right]  $, followed (if an update occurs
in the final round) by the output $b_{t_{k}+s}$ of the Laplace measurement
used to approximate the value of $E_{t_{k}+s}$.

We bound the damage that can be caused to the original (product) state
$\left\vert \psi\right\rangle $ by running the epoch superoperator. We also
show that running the epoch superoperator on $\left\vert \psi\right\rangle $
triggers an update with constant probability, but with constant probability no
update occurs before round $t^{\prime\prime}$.

\begin{claim}
\label{claim:epoch}There exists a noise magnitude $\mu=O\left(  \frac{1}%
{\sqrt{n}\varepsilon^{2}}\right)  $ such that the following holds. \ Fixing
any round $t_{k}\in\left[  m\right]  $, prior measurements $E_{1}%
,\ldots,E_{t_{k}-1}$, and a history of outputs $\bar{y}_{\leq t_{k}-1}$ in the
previous rounds, define the epoch superoperator as above. \ Then:

\begin{enumerate}
\item[(1)] When we run the epoch superoperator on the state $\left\vert
\psi\right\rangle $, the probability that an update occurs is at least $0.15$.

\item[(2)] When we run the epoch superoperator on the state $\left\vert
\psi\right\rangle $, the probability that \emph{no} update occurs before the
round $t^{\prime\prime}$ is at least $0.4$.

\item[(3)] Let $\left\vert \psi^{\prime}\right\rangle $ be the state in the
registers after running this superoperator on the original state $\left\vert
\psi\right\rangle $ (including observing the epoch's outputs).
\ The damage is bounded by:
\[
\left\Vert |\psi\rangle\langle\psi|-|\psi^{\prime}\rangle\langle\psi^{\prime
}|\right\Vert _{\operatorname*{tr}}=O\left(  \frac{1}{\sqrt{n}\mu}\right)  .
\]

\end{enumerate}
\end{claim}

\begin{proof}
For any possible last round $t^{\prime}\geq t_{k}$, and any possible output
$y$ of the epoch superoperator (comprised of a sequence of no-updates, which
may or may not end with an update), we bound the damage as follows. \ We take
$t^{\prime\prime}\geq t^{\prime}$ to be the round on which the epoch always
ends (unless there is an earlier update). \ Since Condition (\ref{condition:p}%
) did not hold at the beginning of round $t^{\prime}-1$, we have:
\[
p=\prod_{j=t_{k}}^{t^{\prime}-1}\left(  1-\lambda_{j}\right)  \geq\frac{1}%
{2}.
\]
Using the fact that for any $\xi\in\left[  0,1\right]  $, we have that
$0\leq1-\xi\leq e^{-\xi}$:
\[
\frac{1}{2}\leq p=\prod_{j=t_{k}}^{t^{\prime}-1}\left(  1-\lambda_{j}\right)
\leq\prod_{j=t_{k}}^{t^{\prime}-1}e^{-\lambda_{j}}=e^{-\sum_{j=t_{k}%
}^{t^{\prime}-1}\lambda_{j}}.
\]
By taking logarithms on both sides of the above inequality we get:
\begin{equation}
\sum_{j=t_{k}}^{t^{\prime}-1}\lambda_{j}<1. \label{eq:product_to_sum}%
\end{equation}

Claim \ref{claim:noupdate_gentle} gives a bound $\varepsilon_{j}$ on the
damage when running the $j^{th}$ $\mathbf{CheckForUpdate}$ conditioned
superoperator (on the original state), conditioned on output $0$. \ Recall
that this bound $\varepsilon_{j}$ is linear in the update probability
$\lambda_{j}$. \ Claim \ref{claim:update_gentle} gives a bound $\varepsilon
_{t_{k}+s}$ on the damage caused by the conditioned superoperator run in the
last round, conditioned on any possible outcome in that round. \ Combining
these bounds with inequality (\ref{eq:product_to_sum}), we get:
\[
\sum_{j=t_{k}}^{t^{\prime}}\varepsilon_{j}=O\left(  \frac{1}{\sqrt{n}\mu
}\left(  1+\sum_{j=t_{k}}^{t^{\prime}-1}\lambda_{j}\right)  \right)  =O\left(
\frac{1}{\sqrt{n}\mu}\right)  .
\]
Define $q$ to be the probability of no update in rounds $t_{k},\ldots
,t^{\prime}-1$ in a \textquotedblleft real\textquotedblright\ execution of the
epoch superoperator on the state $\left\vert \psi\right\rangle $ (and note
that $q>0$, because we are considering an output $y$ that can actually occur).
\ Applying Lemma \ref{facepalmlemma} to the conditioned superoperator's run in
the first $t^{\prime}-t_{k}-1$ rounds, and using also the bound on $\mu$ in
the claim's statement, we get:
\begin{equation}
\left\vert p-q\right\vert \leq\sum_{j=t_{k}}^{t^{\prime}-1}\varepsilon
_{j}<0.1, \label{eq:p_and_q}%
\end{equation}
which in particular implies that $q\geq0.4$, proving item (2) above.\ \ By
Lemma \ref{facepalmlemma} (see also the remark following that lemma about
composing with a final superoperator---in our case, the $\left(  t^{\prime
}\right)  ^{th}$ round), we conclude that:
\[
\left\Vert \left\vert \psi\right\rangle \left\langle \psi\right\vert
-\left\vert \psi^{\prime}\right\rangle \left\langle \psi^{\prime}\right\vert
\right\Vert _{\operatorname*{tr}}\leq5\sum_{j=t_{k}}^{t^{\prime}}%
\varepsilon_{j}=O\left(  \frac{1}{\sqrt{n}\mu}\right)  .
\]
\bigskip

\textbf{Bounding the update probability.} \ To lower-bound the probability of
an update, observe first that if the probability $\lambda_{t^{\prime\prime}}$
of an update in the last round, when we run it on a fresh copy of the state
$\left\vert \psi\right\rangle $, satisfies $\lambda_{t^{\prime\prime}}%
\geq\frac{1}{4}$, then by gentleness of the epoch superoperator as a whole
(see above), when we run it on $\left\vert \psi\right\rangle $, the
probability of an update in the last round (run on the state $\left\vert
\psi_{\leq t^{\prime\prime}}\right\rangle $) is greater than $0.15$.

Thus, we restrict our attention to the case that $\lambda_{t^{\prime\prime}%
}<\frac{1}{4}$.\ \ Since $t^{\prime\prime}$ is the first round where Condition
(\ref{condition:p}) is violated, we know that $\prod_{j=t_{k}}^{t^{\prime
\prime}}\left(  1-\lambda_{j}\right)  \leq\frac{1}{2}$. \ I.e.,\ we have a
lower bound on the probability of an update if each round was run on a fresh
copy of $\left\vert \psi\right\rangle $. \ Since we assume $\lambda
_{t^{\prime\prime}}<\frac{1}{4}$, we in fact have an upper bound on the
probability of no update in the first $t^{\prime\prime}-1$ rounds of such an
execution:%
\[
p=\prod_{j=t_{k}}^{t^{\prime\prime}-1}\left(  1-\lambda_{j}\right)  \leq
\frac{3}{4}.
\]
By equation (\ref{eq:p_and_q}) (restricted to the case $t^{\prime}%
=t^{\prime\prime}$), we deduce a similar bound on the probability $q$ of no
update in the first $t^{\prime\prime}-1$ rounds of the actual execution (an
execution that does not get fresh copies of $\left\vert \psi\right\rangle $).
\ In particular, the probability of no update in this \textquotedblleft
actual\textquotedblright\ execution is at most $0.85$.\bigskip
\end{proof}

\textbf{Accumulated damage.} \ By Claim \ref{claim:epoch}, running each epoch
superoperator on the initial state only results in bounded damage, and
triggers an update with constant probability. \ By Lemma \ref{damagelem}
(additivity of damage), when we run a sequence of $k$ epochs, the total damage
is at worst multiplied by $k$. \ Moreover, so long as this accumulated damage
is smaller than $0.05$, each epoch still triggers an update with probability
at least $0.1$ (because the trace distance between the original state and the
state in the registers when we run the epoch is bounded). \ Under these
conditions, by Azuma's inequality, with all but $\frac{\delta}{2}$
probability, the number of epochs that occur before $\ell\left(
d,\varepsilon\right)  $ updates are triggered (and the QPMW algorithm aborts)
is at most:
\[
k=O\left(  \ell\left(  d,\varepsilon\right)  +\sqrt{\ell\left(  d,\varepsilon
\right)  \log\frac{1}{\delta}}\right)  .
\]
By Theorem \ref{achknthm} we have that $\ell\left(  d,\varepsilon\right)
=O\left(  \frac{\log d}{\varepsilon^{2}}\right)  $. \ Note that the choice of
noise parameter $\mu$ guarantees that the accumulated damage over $k$ such
rounds is indeed less than $0.05$ (in fact it is less than $\alpha$; see
equation (\ref{eq:noise_lowerbound})). \ We conclude that in this random
process, the probability that each epoch triggers an update stays above $0.1$
for the first $k$ epochs.

By Claim \ref{claim:epoch} and Lemma \ref{damagelem} (additivity of damage),
we can bound the total damage by the number of epochs times the damage per
epoch, and we get that with all but $\frac{\delta}{2}$ probability over the
coins and measurements made by QPMW:
\begin{align*}
\left\Vert |\psi\rangle\langle\psi|-|\psi_{\leq m^{\prime}}\rangle\langle
\psi_{\leq m^{\prime}}|\right\Vert _{\operatorname*{tr}}  &  =O\left(
\sum_{j=1}^{k}\frac{1}{\sqrt{n}\mu}\right) \\
&  =O\left(  \frac{\log d}{\sqrt{n}\mu\varepsilon^{2}}+\frac{\sqrt{\log
d\log\frac{1}{\delta}}}{\sqrt{n}\mu\varepsilon}\right)  .
\end{align*}
\bigskip

\textbf{Accuracy.} \ For given gentleness and accuracy parameters
$\alpha,\delta,\varepsilon>0$, we fix the noise parameter $\mu$ and then
analyze the number of copies needed to guarantee accuracy with high
probability. \ We assume without loss of generality that $\alpha
\leq\varepsilon/100$ (if a larger $\alpha$ is specified, we simply run the
algorithm with $\alpha^{\prime}=\varepsilon/100$). \ We set the parameters so
that in an \textquotedblleft ideal\textquotedblright\ run of the algorithm,
where each round is run on a fresh copy of the state $\rho^{\otimes n}$, the
algorithm is $\varepsilon/10$-accurate with all but a small constant
probability. \ We then use the algorithm's gentleness to show that this
implies accuracy in \textquotedblleft real\textquotedblright\ runs of the
algorithm: namely, we show that in a real run, the algorithm is $\varepsilon
$-accurate with all but a small constant probability. \ The error probability
can be reduced by independent repetitions.

We begin by setting the parameters so that with high probability, the total
damage to the state is bounded by $\alpha$, and recall also that we assume
$\alpha\leq\frac{\varepsilon}{100}$. \ This imposes a constraint on $\mu$:
\[
\min\left(  \alpha,\varepsilon\right)  =O\left(  \frac{\log d}{\sqrt{n}%
\mu\varepsilon^{2}}+\frac{\sqrt{\log d\log\frac{1}{\delta}}}{\sqrt{n}%
\mu\varepsilon}\right)  ,
\]
or equivalently:
\begin{equation}
n=\Omega\left(  \left(  \frac{1}{\mu\varepsilon\min\left(  \alpha
,\varepsilon\right)  }\left(  \frac{\log d}{\varepsilon}+\sqrt{\log d\log
\frac{1}{\delta}}\right)  \right)  ^{2}\right)  . \label{eq:noise_lowerbound}%
\end{equation}
Note that this setting also satisfies the conditions of Claim
\ref{claim:epoch}.

We also want to guarantee that with high probability, an ideal run of the
algorithm would give accurate answers. \ This imposes an upper bound on the
noise magnitude $\mu$. \ We analyze the accuracy by dividing the execution
into epochs, as was done in the gentleness analysis above.

\begin{claim}
[Ideal run accuracy]\label{claim:idealaccuracy} Consider an ideal run of the
algorithm (where each round is run on a fresh copy of $\rho^{\otimes n}$)
where we set:
\begin{equation}
\mu=O\left(  \frac{\varepsilon}{\log m}\right)  . \label{eq:noise_magnitude}%
\end{equation}
Consider an epoch that can run for at most $s$ rounds. The following all hold:

\begin{enumerate}
\item[(1)] With all but $\frac{s}{1000m}$ probability, there will \emph{not}
be an update in any round $t$ of the epoch where $\left\vert \operatorname{Tr}%
(E_{t}\rho)-\operatorname{Tr}(E_{t}\sigma_{t-1})\right\vert \leq
\frac{\varepsilon}{3}$.

\item[(2)] If in any round $t$ of the epoch it is the case that $\left\vert
\operatorname{Tr}(E_{t}\rho)-\operatorname{Tr}(E_{t}\sigma_{t-1})\right\vert
\geq\varepsilon$, then an update occurs in that round with all but $\frac
{1}{1000m}$ probability (note this condition can only hold on the round that
always ends the epoch).

\item[(3)] If the epoch ends in an update round, then the noisy answer $b_{t}$
is $\left(  \frac{\varepsilon}{10}\right)  $-accurate with all but $\frac
{1}{1000m}$ probability.
\end{enumerate}
\end{claim}
\end{proof}

\begin{proof}
The claim follows immediately from the exponential tails of the Laplace
distribution: in each round, for each draw of Laplace noise, with all but
$\frac{1}{1000m}$ probability, the noise magnitude is at most $\frac
{\varepsilon n}{10}$.

Recall that an epoch can end before reaching its last round. \ However, the
probability of each epoch reaching its final round is at least $1/2$ (by the
definition of the epoch superoperator). \ Thus, if an epoch can run for at
most $s$ rounds, then the \emph{expected} number of rounds is at least $s/2$.
\ We conclude that with probability at least $0.9$, the sum, over all epochs,
of the number of rounds for which each epoch can run, is at most $10m$ (by
Markov's inequality). \ By Claim \ref{claim:idealaccuracy}, taking a union
bound over all epochs, and taking $\mu$ as set as in Equation
\ref{eq:noise_magnitude}, we deduce that with all but a small constant
probability over the noise choices, the conditions of the online learning
theorem for quantum states (Theorem \ref{achknthm}) all hold in all rounds
simultaneously. \ By that theorem, we conclude that with all but small
constant probability over its coins, the QPMW algorithm does not abort, and
its answers are all $\varepsilon$-accurate.\bigskip

\textbf{How many copies do we need?} \ Before proceeding to prove that a real
run of the algorithm is also accurate, we specify the number of copies needed
to simultaneously satisfy the constraints in equations
(\ref{eq:noise_lowerbound}) and (\ref{eq:noise_magnitude}) by taking $n$ to be
large enough. \ We can do so while still guaranteeing the upper bound:
\begin{equation}
n=O\left(  \frac{\log^{2}m}{\varepsilon^{4}\min\left(  \alpha,\varepsilon
\right)  ^{2}}\cdot\left(  \frac{\log^{2}d}{\varepsilon^{2}}+\left(  \log
d\right)  \left(  \log\frac{1}{\delta}\right)  \right)  \right)  .
\label{eq:DB_size}%
\end{equation}
Note that this setting of $n$, which we use in the proof of Claim
\ref{claim:noupdate_gentle}, also guarantees that $\sqrt{n}\mu$ is a
sufficiently large constant. \ Further, this number of copies guarantees
accuracy with all but small constant probability. The error probability can be
reduced to $\beta$ by running $O(\log(1/\beta))$ independent copies of the
algorithm, and outputting the median answer in each round.

For simplicity, in the statement of Theorem \ref{ostthm} we claim a slightly
more relaxed bound of:
\[
n=O\left(  \frac{\left(  \log^{2}m+\log\frac{1}{\delta}\right)  \cdot\log
^{2}d\cdot\log\frac{1}{\beta}}{\varepsilon^{6}\min\left(  \alpha
,\varepsilon\right)  ^{2}}\right)  .
\]
\bigskip

\textbf{A hybrid execution.} \ Consider a hybrid execution, where each epoch
superoperator (see above) is run on the \textquotedblleft
real\textquotedblright\ state (with no substitutions), but after each
superoperator completes its operation, we replace the resulting state with a
fresh copy of $\rho^{\otimes n}$ before proceeding to the next epoch
superoperator. \ Since each epoch is $\alpha_{e}=O\left(  \frac{1}{\sqrt{n}%
\mu}\right)  $-gentle (Claim \ref{claim:epoch}), we can apply the Damage Lemma
(Lemma \ref{facepalmlemma}) to conclude accuracy properties for the epoch:

\begin{claim}
[Hybrid run accuracy]\label{claim:hybridaccuracy} Consider a hybrid run of the
algorithm (where each epoch is run on a fresh copy of $\rho^{\otimes n}$),
with the parameters set as in Equations \eqref{eq:noise_lowerbound},
\eqref{eq:noise_magnitude}, and \eqref{eq:DB_size}. \ Let $\alpha_{e}$ be the
bound on the gentleness of each epoch. \ Consider an epoch that can run for at
most $s$ rounds. The following all hold:

\begin{enumerate}
\item[(1)] With all but $\frac{s}{1000m}+\alpha_{e}$ probability, there will
\emph{not} be an update in any round $t$ of the epoch where $\left\vert
\operatorname{Tr}(E_{t}\rho)-\operatorname{Tr}(E_{t}\sigma_{t-1})\right\vert
\leq\frac{\varepsilon}{3}$.

\item[(2)] If in the final round of the epoch it is the case that $\left\vert
\operatorname{Tr}(E_{t}\rho)-\operatorname{Tr}(E_{t}\sigma_{t-1})\right\vert
\geq\varepsilon$, then an update occurs in that round with all but $\frac
{1}{1000m}+\alpha_{e}$ probability.

\item[(3)] If the epoch ends in an update round, then the noisy answer $b_{t}$
is $\frac{\varepsilon}{10}$-accurate with all but $\frac{1}{1000m}+\alpha_{e}$ probability.
\end{enumerate}
\end{claim}

\begin{proof}
Consider the set $I$ of rounds where $\left\vert \operatorname{Tr}(E_{t}%
\rho)-\operatorname{Tr}(E_{t}\sigma_{t-1})\right\vert \leq\frac{\varepsilon
}{3}$. \ By Claim \ref{claim:idealaccuracy}, the probability that in an ideal
execution an update occurs in one of the rounds in $I$ is at most $\frac
{s}{1000m}$. \ Applying Lemma \ref{facepalmlemma} to the epoch superoperator,
we conclude that the probability an update occurring in one of the rounds in
$I$ is at most $\frac{s}{1000m}+\alpha_{e}$. We note that in this application
of Lemma \ref{facepalmlemma}, we restrict to the subset of quantum operations
corresponding to rounds in $I$ (and condition on the \textquotedblleft no
update\textquotedblright\ outcome in those rounds). \ Claim
\ref{claim:idealaccuracy} further bounds the ideal-execution probability of no
update if in the last round $\left\vert \operatorname{Tr}(E_{t}\rho
)-\operatorname{Tr}(E_{t}\sigma_{t-1})\right\vert \geq\varepsilon$, and the
probability that the update ends in an update round but the noisy answer is
not $\frac{\varepsilon}{10}$-accurate. \ By $\alpha_{e}$-gentleness of the
epoch superoperator, we conclude that the probabilities of these two events
occurring in the hybrid execution are both bounded by $\frac{1}{1000m}%
+\alpha_{e}$.
\end{proof}

By Claim \ref{claim:epoch}, the probability that there is no update until the
last ($s^{th}$) round of an epoch is at least $0.4$. \ Thus, in the hybrid
execution, the \emph{expected} number of rounds for which an $s$-round epoch
will run is at least $0.4s$. \ Similarly to the analysis of the ideal
execution, taking a union bound over $k$ epoch superoperators and taking
$\tau>0$ to be a small constant, we conclude that with all but $\tau
+O(\alpha_{e}\cdot k)$ probability, the conditions of the online learning
theorem all hold and the answers returned are all $\varepsilon$-accurate.
Further, by the choice of parameters in Equation \eqref{eq:noise_lowerbound},
we know that with high probability, when we run QPMW and take $k^{\ast}$ to be
the number of epochs needed to process all $m$ measurements, we have
$O(\alpha_{e}\cdot k^{\ast})=O(\alpha)$. \ We conclude that with all but a
small constant probability, a hybrid execution of QPMW does not terminate
prematurely, and is $\varepsilon$-accurate on every measurement.\bigskip

\textbf{The real execution.} \ Lastly, we consider the real execution, where
the epoch superoperators are run in sequence, without any refreshing of the
state in the registers. \ We use the gentleness of the epoch superoperator to
conclude that the algorithm remains accurate in its real execution.

\begin{claim}
[Real run accuracy]\label{claim:realaccuracy}Consider a real run of the
algorithm, with the parameters set as in Equations
\eqref{eq:noise_lowerbound}, \eqref{eq:noise_magnitude}, and
\eqref{eq:DB_size}. \ With all but small constant probability over the
algorithm's coins, the following hold in every round $t$ of the algorithm (simultaneously):

\begin{enumerate}
\item[(1)] If $\left\vert \operatorname{Tr}(E_{t}\rho)-\operatorname{Tr}%
(E_{t}\sigma_{t-1})\right\vert \leq\frac{\varepsilon}{3}$, then there is no update.

\item[(2)] If $\left\vert \operatorname{Tr}(E_{t}\rho)-\operatorname{Tr}%
(E_{t}\sigma_{t-1})\right\vert \geq{\varepsilon}$, then there is an update.

\item[(3)] If $t$ is an update round, then the noisy answer $b_{t}$ is
$\frac{\varepsilon}{10}$-accurate.
\end{enumerate}
\end{claim}

\begin{proof}
Let $B$ be the (\textquotedblleft bad\textquotedblright) event that in some
round $t$ of QPMW it is either the case that:

\begin{enumerate}
\item[(i)] an update occurs even though $\left\vert \operatorname{Tr}%
(E_{t}\rho)-\operatorname{Tr}(E_{t}\sigma_{t-1})\right\vert \leq
\frac{\varepsilon}{3}$, or

\item[(ii)] no update occurs even though $\left\vert \operatorname{Tr}%
(E_{t}\rho)-\operatorname{Tr}(E_{t}\sigma_{t-1})\right\vert \geq\varepsilon$, or

\item[(iii)] $t$ is an update round, and the noisy answer $b_{t}$ is
\emph{not} $\frac{\varepsilon}{10}$-accurate.
\end{enumerate}

By the foregoing analysis, the probability of the event $B$ in the hybrid
execution is bounded by a small constant, say $\tau$. \ We would like to now
make a similar argument for a real execution, where the state is not
\textquotedblleft refreshed\textquotedblright\ between epoch superoperators.

Towards this, let $k$ be a bound on the number of epoch superoperators in a
run of QPMW, and let $\alpha_{e}$ be the bound on the gentleness of each epoch
superoperator. \ We consider further hybrids, where in the $i^{th}$ hybrid
$\mathcal{H}_{i}$, the first $i$ epochs are each run on fresh copies of
$\rho^{\otimes n}$, but there is no further refreshing after the $i^{th}$
epoch. \ Thus the first hybrid $\mathcal{H}_{1}$ equals the real execution,
and the $k^{th}$ hybrid $\mathcal{H}_{k}$ equals the hybrid execution. \ By
$\alpha_{e}$-gentleness of the epoch superoperator, we have that for every
$i$:
\[
\left\vert \Pr_{\mathcal{H}_{i}}[B]-\Pr_{\mathcal{H}_{i+1}}[B]\right\vert
\leq\alpha_{e}.
\]
This is simply because the $i^{th}$ and $(i+1)^{st}$ hybrid differ only in
running the $(i+1)^{st}$ epoch: in $\mathcal{H}_{i}$ that epoch is run on the
state in the registers after the $i^{th}$ epoch, whereas in $\mathcal{H}%
_{i+1}$ that epoch is run on a fresh copy of $\rho^{\otimes n}$. \ By the
$\alpha_{e}$-gentleness of the $i^{th}$ epoch, the trace distance between
these two states is at most $\alpha_{e}$. \ So the two hybrids only differ in
the probability that the event $B$ occurs in the $(i+1)^{st}$ epoch, and this
difference in probabilities is upper-bounded by $\alpha_{e}$.

By a hybrid argument, we conclude that the probability of the event $B$
occurring in the real execution is at most $\tau+k\cdot\alpha_{e}$. \ Further,
by the choice of parameters in Equation \eqref{eq:noise_lowerbound}, we know
that we can take $k^{\ast}$ to be a bound on the number of epochs such that
with high probability, $k^{\ast}$ epochs suffice to process all $m$
measurements, and $O(\alpha_{e}\cdot k^{\ast})=O(\alpha)$. \ We conclude that
with all but a small constant probability, the real execution of QPMW does not
terminate prematurely, and is $\varepsilon$-accurate on every measurement.
\end{proof}

Finally, we reduce the error probability to $\beta$ by running $O(\log
(1/\beta))$ independent executions and outputting the median answer in each
round. \ \medskip This completes the accuracy proof for QPMW.
\end{proof}

\begin{proof}
[Proof of Claim \ref{claim:noupdate_gentle}]We begin by assuming that the
probability $\lambda_{t}$ of an update is smaller than some sufficiently small
constant. \ If this is not the case, then the claim follows immediately from
Lemma \ref{dpgentlepovm}, because $\mathbf{CheckForUpdate}$ runs a $\frac
{1}{n\mu}$-DP classical algorithm. \ Further, we assume throughout that
$\sqrt{n}\mu$ is larger than a sufficiently large constant (see the remark
following equation (\ref{eq:DB_size})).

We follow similar reasoning to the proof of Lemma \ref{dpgentlepovm}. \ We
begin with a pure product state in the registers%
\[
\left\vert \psi\right\rangle :=\left\vert \psi_{1}\right\rangle \otimes
\cdots\otimes\left\vert \psi_{n}\right\rangle .
\]
Let $\left\vert \psi_{\text{no}}\right\rangle $ be the state after applying
the conditioned superoperator $\mathbf{CheckForUpdate}$, conditioned on
$u_{t}=0$ (\textquotedblleft no update\textquotedblright). \ The
$\mathbf{CheckForUpdate}$ superoperator applies the POVM $E_{t}$ to each
$\left\vert \psi_{i}\right\rangle $, and then runs a classical DP algorithm on
the $n$ bits observed. \ To implement it, we first apply a unitary
transformation (to the state and ancilla qubits). \ This gives a new state:%
\[
\left\vert \phi_{i}\right\rangle :=\sum_{b_{i}\in\left\{  0,1\right\}
}a_{i,x}\left\vert b_{i}\right\rangle \left\vert v_{i,b_{i}}\right\rangle .
\]
Let $X\in\left\{  0,1\right\}  ^{n}$ be the values observed when measuring the
registers $\left\vert b_{i}\right\rangle $. \ We draw a noise value $\eta$
from the Laplace distribution with magnitude $n\mu$, and output $u_{t}=0$ (no
update) whenever:
\[
\left\vert \frac{1}{n}\sum_{i=1}^{n}X_{i}+\frac{\eta}{n}-\operatorname{Tr}%
\left(  E_{t}\omega_{t-1}\right)  \right\vert \leq\frac{\varepsilon}{2}.
\]

Let $\mathcal{D}$ be the distribution over $X\in\{0,1\}^{n}$ defined by
$\Pr_{\mathcal{D}}\left[  X\right]  =\left\vert \alpha_{X}\right\vert ^{2}$,
where $\alpha_{X}=a_{i,1}\cdots a_{i,n}$, and note that $\mathcal{D}$ is
indeed a product distribution. \ Let $\mathcal{D}_{0}$ be the distribution
$\mathcal{D}$ conditioned on the event when we run the above (classical)
procedure on $X$ we get $u_{t}=0$ (no update). \ Following the proofs of
Lemmas \ref{dpgentlepovm} and \ref{dpgentleprojective}, we can implement the
$\mathbf{CheckForUpdate}$ measurement so that:
\begin{equation}
\left\Vert |\psi\rangle\langle\psi|-|\psi_{\text{no}}\rangle\langle
\psi_{\text{no}}|\right\Vert _{\operatorname*{tr}}\leq\sqrt{2\operatorname{KL}%
\left(  \mathcal{D},{\mathcal{D}_{0}}\right)  }. \label{eq:KL_bound}%
\end{equation}

At this point we diverge from the proof of Lemma \ref{dpgentlepovm}. \ There,
we considered the distribution $\mathcal{D}_{y}$ obtained by conditioning the
product distribution $\mathcal{D}$ on an outcome $y$ of a $\sigma$-DP
algorithm. \ We bounded the KL-divergence between these distributions, and
used that to bound the trace distance by $O\left(  \sigma\sqrt{n}\right)  $.
\ Here, while we know that the $\mathbf{CheckForUpdate}$ measurement is
$\sigma$-DP for $\sigma=\frac{1}{n\mu}$, when the probability $\lambda_{t}$ of
an update is much smaller than $\frac{1}{\sqrt{n}\mu}$, we want to argue that
observing a \textquotedblleft no update\textquotedblright\ answer causes much
less damage to the state.\bigskip

\noindent\textbf{Improving the DP guarantee.} The intuition is that when
$\lambda_{t}$ is small, for a \textquotedblleft typical\textquotedblright%
\ input $X$ drawn from $\mathcal{D}$, the probability of no update is quite
large: $1-\lambda_{t}$. \ For an adjacent input $X^{\prime}$, this probability
of no update is at least $1-e^{\sigma}\lambda_{t}$. \ For small $\lambda_{t}$,
the log-ratio between these two probabilities is roughly $\lambda_{t}\sigma$.
\ A compelling strategy is to try to bound the KL-divergence using this
improved bound, by following a similar argument to the proof of Lemma
\ref{dpgentlepovm}. \ For observe that that proof applies even when we focus
on any particular output $y$---in this case, \textquotedblleft no
update\textquotedblright---using the log-ratio guaranteed for that particular output.

The catch, which significantly complicates the proof, is that not all inputs
drawn from $\mathcal{D}$ are \textquotedblleft typical.\textquotedblright%
\ \ Some of these inputs have much higher update probabilities than
$\lambda_{t}$, whereas the proof of Lemma \ref{dpgentlepovm} required a
worst-case bound that applies to \textit{every} input in the support of
$\mathcal{D}$. \ On the other hand, by concentration bounds on the Hamming
weights of inputs drawn from $\mathcal{D}$, the probability of drawing an $X$
for which the update probability is significantly higher than $\lambda_{t}$ is
very small.

To obtain an improved bound, we extend the proof of Lemma \ref{dpgentlepovm}
to this case, using concentration of the (generalized) binomial distribution
(a subgaussian distribution), to show that while the contribution of
\textquotedblleft far\textquotedblright\ inputs to the KL-divergence grows,
their probability shrinks more quickly than this growth. \ To do this, we
partition the inputs into disjoint sets $\Delta_{j}$, according to the
difference between their Hamming weight and the expected Hamming weight. \ We
account for the contributions of each set in this partition to the
KL-divergence to show the claimed bound. \ The details (which can get long and
technical) follow.

\medskip\noindent\textbf{The event $\Delta_{j}$.} \ For each integer
$j\in\lbrack1,\sqrt{n}]$, we define the event $\Delta_{j}\subseteq\{0,1\}^{n}$
to consist of all inputs whose Hamming weights are at least $(j-1)\sqrt{n}%
$\ and less than $j\sqrt{n}$ away from the expectation:
\[
\Delta_{j}=\left\{  X\in\{0,1\}^{m}:\left\vert \left(  \sum_{i=1}^{n}%
X_{i}\right)  -\operatorname{E}_{Y\sim\mathcal{D}}\left[  \sum_{i}%
Y_{i}\right]  \right\vert \in\left[  \left(  j-1\right)  \sqrt{n},j\sqrt
{n}\right)  \right\}  .
\]
By Azuma's inequality, a random input drawn from $\mathcal{D}$ will with high
probability be in $\Delta_{j}$ for small $j$:%
\begin{equation}
\Pr_{X\sim\mathcal{D}}\left[  X\in\Delta_{j}\right]  \leq2\exp\left(
-\frac{(j-1)^{2}}{2}\right)  . \label{eq:Delta_prob}%
\end{equation}
In particular, for a random input $X\sim\mathcal{D}$, the expected value of
the $j$ such that $X\in\Delta_{j}$ is small:%
\begin{equation}
\sum_{j=1}^{\infty}j\cdot\Pr_{\mathcal{D}}\left[  \Delta_{j}\right]  <2.
\label{eq:Delta_std_dev}%
\end{equation}
Similarly, we can also bound higher moments of this function. \ Since the
distribution over the Hamming weight of $X$ is subgaussian with standard
deviation $O\left(  \sqrt{n}\right)  $, we also have:
\begin{equation}
\sum_{j=1}^{\infty}j^{2}\cdot\Pr_{\mathcal{D}}\left[  \Delta_{j}\right]
=O(1). \label{eq:Delta_moment}%
\end{equation}

We use $\left(  \mathcal{D}|\Delta_{j}\right)  $ to denote the distribution
$\mathcal{D}$ conditioned on the event $\Delta_{j}$ (and similarly for
$\mathcal{D}_{0}$). \ We proceed with a sequence of technical propositions,
which will be used to bound the KL-divergence between $\mathcal{D}$ and
$\mathcal{D}_{0}$.

\begin{proposition}
\label{prop:good_DB_DP}Let $\Delta_{j}$ be as defined above. \ For every
$j\geq1$, every $X\in\Delta_{j}$, and every $b\in\{0,1\}$, we have:
\[
\left\vert \ln\frac{\Pr_{Y\sim\mathcal{D}}[u_{t}=b~|~Y]}{\Pr[u_{t}%
=b~|~X]}\right\vert \leq\frac{2j+5}{\sqrt{n}\mu}.
\]

\end{proposition}

\begin{proof}
The intuition is that the probability that $u_{t}=b$ (by $\mathcal{D}$) is
dominated by the probability that this event occurs for inputs whose Hamming
weights are close to the expectation. \ By the differential privacy of the
Laplace noise mechanism, the log-ratio of probabilities for inputs close to
the expectation and inputs in $\Delta_{j}$ is upper-bounded by $\frac{j}%
{\sqrt{n}\mu}$ in magnitude. \ We show one direction (an upper bound); the
lower bound follows similarly:
\begin{align*}
\frac{\Pr_{Y\sim\mathcal{D}}[u_{t}=b~|~Y]}{\Pr[u_{t}=b~|~X]}  &  =\sum
_{k=1}^{\sqrt{n}}\Pr_{\mathcal{D}}[\Delta_{k}]\cdot\frac{\Pr_{\mathcal{D}%
}[u_{t}=b~|~\Delta_{k}]}{\Pr[u_{t}=b~|~X]}\\
&  \leq\sum_{k=1}^{\sqrt{n}}\Pr_{\mathcal{D}}[\Delta_{k}]\cdot e^{\frac
{k+2j}{\sqrt{n}\mu}}\\
&  =e^{\frac{2j}{\sqrt{n}\mu}}\left(  \sum_{k=1}^{\sqrt{n}\mu}\Pr
_{\mathcal{D}}[\Delta_{k}]\cdot e^{\frac{k}{\sqrt{n}\mu}}+\sum_{k=\sqrt{n}%
\mu+1}^{\sqrt{n}}\Pr_{\mathcal{D}}[\Delta_{k}]\cdot e^{\frac{k}{\sqrt{n}\mu}%
}\right) \\
&  \leq e^{\frac{2j}{\sqrt{n}\mu}}\left(  \sum_{k=1}^{\sqrt{n}\mu}%
\Pr_{\mathcal{D}}[\Delta_{k}]\cdot\left(  1+\frac{2k}{\sqrt{n}\mu}\right)
+\sum_{k=\sqrt{n}\mu+1}^{\sqrt{n}}\Pr_{\mathcal{D}}[\Delta_{k}]\cdot
e^{\frac{k}{\sqrt{n}\mu}}\right) \\
&  \leq e^{\frac{2j}{\sqrt{n}\mu}}\left(  1+\frac{2}{\sqrt{n}\mu}\sum
_{k=1}^{\sqrt{n}\mu}\Pr_{\mathcal{D}}[\Delta_{k}]k+\sum_{k=\sqrt{n}\mu
+1}^{\sqrt{n}}\Pr_{\mathcal{D}}[\Delta_{k}]\cdot e^{\frac{k}{\sqrt{n}\mu}%
}\right) \\
&  \leq e^{\frac{2j}{\sqrt{n}\mu}}\left(  1+\frac{4}{\sqrt{n}\mu}%
+\sum_{k=\sqrt{n}\mu+1}^{\sqrt{n}}\Pr_{\mathcal{D}}[\Delta_{k}]\cdot
e^{\frac{k}{\sqrt{n}\mu}}\right) \\
&  <e^{\frac{2j}{\sqrt{n}\mu}}\left(  1+\frac{5}{\sqrt{n}\mu}\right)  .
\end{align*}
Here the second line follows from the differential privacy of the Laplace
noise mechanism, as well as the fact that the Hamming distance between inputs
$Y\in\Delta_{k}$ and $X\in\Delta_{j}$ is at most $(k+2j)\sqrt{n}$. \ The
second-to-last line uses equation (\ref{eq:Delta_std_dev}), while the final
line uses equation (\ref{eq:Delta_prob}), and can be seen as follows:
\begin{align*}
\sum_{k=\sqrt{n}\mu+1}^{\sqrt{n}}\Pr_{\mathcal{D}}[\Delta_{k}]\cdot
e^{k/\left(  \sqrt{n}\mu\right)  }  &  \leq\sum_{k=\sqrt{n}\mu}^{\sqrt{n}}%
\exp\left(  \frac{-k^{2}}{2}+\frac{k+1}{\sqrt{n}\mu}\right) \\
&  <4e^{-n\mu^{2}/2}\\
&  <\frac{1}{\sqrt{n}\mu},
\end{align*}
where the last two inequalities hold so long as $\sqrt{n}\mu$ is a
sufficiently large constant.
\end{proof}

\begin{proposition}
\label{prop:good_DBs}Let $\Delta_{j}$ be as defined above. \ For every
$j\geq1$, every $X\in\Delta_{j}$, and every input $X^{\prime}$ that differs
from $X$ in a single coordinate, we have
\[
\left\vert \ln\frac{\Pr\left[  u_{t}=0~|~X\right]  }{\Pr\left[  u_{t}%
=0~|~X^{\prime}\right]  }\right\vert \leq\frac{4\min\left\{  1,e^{\frac
{2j+5}{\sqrt{n}\mu}}\cdot\lambda_{t}\right\}  }{{n\mu}}.
\]

\end{proposition}

\begin{proof}
First, since we add Laplace noise of magnitude $n\mu$ before checking for an
update, for \textit{every} pair of adjacent inputs $X,X^{\prime}\in
\{0,1\}^{n}$, the log-ratio between the probabilities of $u_{t}=0$ is at most
$\frac{1}{n\mu}$. \ When the probability of an update is smaller, we can
improve this bound as follows. \ Define $q_{t}$ to be the probability of an
update (i.e., $u_{t}=1$) given the input $X$. \ By Proposition
\ref{prop:good_DB_DP}, we have $q_{t}\leq e^{\frac{2j+5}{\sqrt{n}\mu}}%
\cdot\lambda_{t}$.

Take the count on $X$ to be $k=\sum_{i}X_{i}$. \ An update is triggered when
the difference between the noisy count and $\operatorname{Tr}(E_{t}%
\sigma_{t-1})$ is too large---or\ equivalently, when the noisy count passes a
threshold $h>k+1$.\footnote{This is without loss of generality: the case
$h<k-1$ can be handled similarly. \ The case where $k\in\left[
h-1,h+1\right]  $ cannot occur because then $\lambda_{t}$ would be much larger
than say $\frac{1}{100}$, whereas we assumed $\lambda_{t}$ was sufficiently
small.} \ Thus, $q_{t}=\Pr\left[  k+\eta>h\right]  $. \ Similarly, the
probability $q_{t}^{\prime}$ of an update on $X^{\prime}$ is $\Pr\left[
k+\eta+1>h\right]  $. \ (The case where the count on $X^{\prime}$ is smaller
than on $X$ is handled similarly.) \ By the definition of the Laplace
distribution, these probabilities are given by:
\begin{align*}
q_{t}  &  =\frac{1}{2}\exp\left(  -\frac{h-k}{n\mu}\right)  ,\\
q_{t}^{\prime}  &  =\frac{1}{2}\exp\left(  -\frac{h-k-1}{n\mu}\right)
\end{align*}
Now by standard manipulations we get:
\begin{align*}
\frac{\Pr\left[  u_{t}=0~|~X\right]  }{\Pr\left[  u_{t}=0~|~X^{\prime}\right]
}  &  =\frac{1-q_{t}}{1-q_{t}^{\prime}}\\
&  =\frac{1-\frac{1}{2}\exp\left(  -\frac{h-k}{n\mu}\right)  }{1-\frac{1}%
{2}\exp\left(  -\frac{h-k-1}{n\mu}\right)  }\\
&  =1+\frac{\exp\left(  -\frac{h-k-1}{n\mu}\right)  -\exp\left(  -\frac
{h-k}{n\mu}\right)  }{2-\exp\left(  -\frac{h-k-1}{n\mu}\right)  }\\
&  \leq1+\exp\left(  -\frac{h-k}{n\mu}\right)  \left(  \exp\left(  \frac
{1}{n\mu}\right)  -1\right) \\
&  =1+2q_{t}\left(  \exp\left(  \frac{1}{n\mu}\right)  -1\right) \\
&  \leq\exp\left(  \frac{4q_{t}}{n\mu}\right)  .
\end{align*}
Here the last line uses the fact that $n\mu$ is a sufficiently large constant.
\ (Note also that, in the case we're analyzing, the ratio of probabilities is
larger than $1$, so we only need to prove an \textit{upper} bound.)
\ Proposition \ref{prop:good_DBs}\ follows, recalling that by its conditions
$q_{t}\leq e^{\frac{2j+5}{\sqrt{n}\mu}}\cdot\lambda_{t}$.
\end{proof}

\begin{proposition}
\label{prop:KL_j}Let $\mathcal{D}$, $\mathcal{D}_{0}$ and $\Delta_{j}$ be as
defined above. \ Then for every $j\geq1$:
\begin{align*}
\operatorname{KL}\left(  (\mathcal{D}|{\Delta_{j}}),(\mathcal{D}_{0}%
|\Delta_{j})\right)   &  \leq\frac{\min\left\{  2,32e^{\frac{4j+10}{\sqrt
{n}\mu}}\cdot\lambda_{t}^{2}\right\}  }{n\mu^{2}},\\
\operatorname{KL}\left(  (\mathcal{D}_{0}|{\Delta_{j}}),(\mathcal{D}%
|\Delta_{j})\right)   &  \leq\frac{\min\left\{  2,32e^{\frac{4j+10}{\sqrt
{n}\mu}}\cdot\lambda_{t}^{2}\right\}  }{n\mu^{2}}.
\end{align*}

\end{proposition}

\begin{proof}
We employ a variant of the proof of Lemma \ref{dpgentleclassical}. \ We spell
out the bound in the first direction, the second direction follows similarly.
\ Recall that $(\mathcal{D}|{\Delta_{j}})$ is the product distribution
$\mathcal{D}$, conditioned on the event $\Delta_{j}$ (the difference between
the Hamming weight of $X$ and its expectation is in the interval
$[(j-1)\sqrt{n},j\sqrt{n})$). \ We can sample an input $X=\left(  x_{1}%
,\ldots,x_{n}\right)  $\ from this distribution by sampling $x_{1}$ from the
marginal distribution over the first entry of $\mathcal{D}$ conditioned on
$\Delta_{j}$, then drawing $x_{2}$ from the marginal distribution over the
second entry, conditioned on $x_{1}$ and $\Delta_{j}$, and so on up to $x_{n}%
$. \ \ Call the $i^{th}$ distribution $\mathcal{S}_{i}$; note that
$\mathcal{S}_{i}$ depends on $x_{1},\ldots,x_{i-1}$ (and on $\Delta_{j}$).
\ Similarly, we can also consider a conditional distribution $(\mathcal{D}%
_{0}|\Delta_{j})$, where we condition both on $u_{t}=0$ (no update) and on the
event $\Delta_{j}$ occurring. \ We can sample from this second distribution by
first drawing $x_{1}$\ from the marginal distribution over the first entry
conditioned on $u_{t}=0$ and on $\Delta_{j}$, then drawing $x_{2}$\ from the
marginal distribution over the second entry conditioned on no update, on
$\Delta_{j}$, and on $x_{1}$, and so on up to $x_{n}$. \ Call the $i^{th}$
distribution in this second process $\mathcal{T}_{i}$; note that
$\mathcal{T}_{i}$ depends on $x_{1},\ldots,x_{i-1}$ (as well as on the set
$\Delta_{j}$ and the event $u_{t}=0$). \ The marginal distributions
$\mathcal{S}_{i}$ and $\mathcal{T}_{i}$ are over $\{0,1\}$.

We note that for any setting of the first $i-1$ variables, the supports of the
random variables $\mathcal{S}_{i}$ and $\mathcal{T}_{i}$ are identical: a
given prefix might make the event $\Delta_{j}$ impossible for a certain fixing
of the $i^{th}$ variable, but in this case the forbidden fixing has weight $0$
both in $\mathcal{S}_{i}$ and in $\mathcal{T}_{i}$. \ By Proposition
\ref{prop:good_DBs} and by Bayes' rule, for every $i\in\left\{  1,\ldots
,n\right\}  $, for every setting $x_{1},\ldots,x_{i-1}\in\left\{  0,1\right\}
$ for the first $i-1$ input coordinates, and for every value $v\in\left\{
0,1\right\}  $ such that $v$ has nonzero probability by $\mathcal{S}_{i}$, the
magnitude of the log-ratio between $v$'s probabilities under $\mathcal{S}_{i}%
$\ and under $\mathcal{T}_{i}$\ is bounded as follows:%
\[
\left\vert \ln\frac{\Pr_{\mathcal{S}_{i}}\left[  v\right]  }{\Pr
_{\mathcal{T}_{i}}\left[  v\right]  }\right\vert =\left\vert \ln\frac
{\Pr\left[  u_{t}=0~|~x_{1},\ldots,x_{i-1},G\right]  }{\Pr\left[
u_{t}=0~|~x_{1},\ldots,x_{i-1},X_{i}=v,G\right]  }\right\vert \leq\frac
{\min\left\{  1,4e^{\frac{2j+5}{\sqrt{n}\mu}}\cdot\lambda_{t}\right\}  }{n\mu}%
\]

By Claim \ref{vadhanclaim}, this means that the \textit{expected} log-ratio
between $\Pr_{\mathcal{S}_{i}}\left[  v\right]  $\ and $\Pr_{\mathcal{T}_{i}%
}\left[  v\right]  $, with respect to $x_{i}$\ drawn from $\mathcal{S}_{i}$,
is upper-bounded by%
\[
2\left(  \frac{\min\left\{  1,4e^{\frac{2j+5}{\sqrt{n}\mu}}\cdot\lambda
_{t}\right\}  }{n\mu}\right)  ^{2}.
\]
As in the proof of Lemma \ref{dpgentleclassical}, we conclude:%
\[
\operatorname{KL}\left(  (\mathcal{D}|{\Delta_{j}}),(\mathcal{D}_{0}%
|\Delta_{j})\right)  =\sum_{i=1}^{n}\operatorname{E}\left[  \operatorname{KL}%
\left(  \mathcal{S}_{i},\mathcal{T}_{i}\right)  \right]  \leq\frac
{2\min\left\{  1,16e^{\frac{4j+10}{\sqrt{n}\mu}}\cdot\lambda_{t}^{2}\right\}
}{n\mu^{2}}%
\]

\end{proof}

\begin{proposition}
\label{prop:log_ratio_j}Let $\mathcal{D}$, $\mathcal{D}_{0}$ and $\Delta_{j}$
be as defined above. \ Partition the line $[1,\sqrt{n}]$ into the following
three segments:
\begin{align*}
J_{1}  &  =\left[  1,\frac{\sqrt{n}\mu-5}{2}\right) \\
J_{2}  &  =\left[  \frac{\sqrt{n}\mu-5}{2},\frac{\sqrt{n}\mu(\ln(\frac
{1}{\lambda_{t}})-1)-5}{2}\right) \\
J_{3}  &  =\left[  \frac{\sqrt{n}\mu(\ln(\frac{1}{\lambda_{t}})-1)-5}{2}%
,\sqrt{n}\right]  .
\end{align*}
Then the following hold:

\begin{itemize}
\item For every integer $j \in J_{1}$:
\begin{align*}
\left\vert \ln\frac{\Pr_{\mathcal{D}}[\Delta_{j}]}{\Pr_{\mathcal{D}_{0}%
}[\Delta_{j}]} \right\vert \leq\frac{\lambda_{t}(4j+10)}{\sqrt{n} \mu}.
\end{align*}

\item For every integer $j\in J_{2}$:
\[
\left\vert \ln\frac{\Pr_{\mathcal{D}}[\Delta_{j}]}{\Pr_{\mathcal{D}_{0}%
}[\Delta_{j}]}\right\vert \leq2\lambda_{t}e^{\frac{2j+5}{\sqrt{n}\mu}}.
\]

\item For every integer $j\in J_{3}$:
\[
\left\vert \ln\frac{\Pr_{\mathcal{D}}[\Delta_{j}]}{\Pr_{\mathcal{D}_{0}%
}[\Delta_{j}]}\right\vert \leq\frac{2j+5}{\sqrt{n}\mu}.
\]

\end{itemize}
\end{proposition}

\begin{proof}
First, by Bayes' rule, for every $j$ we have:
\[
\frac{\Pr_{\mathcal{D}}[\Delta_{j}]}{\Pr_{\mathcal{D}_{0}}[\Delta_{j}]}%
=\frac{\Pr_{\mathcal{D}}[u_{t}=0]}{\Pr_{\mathcal{D}}[u_{t}=0~|~\Delta_{j}]}.
\]
Further, by Proposition \ref{prop:good_DB_DP} we have that for every integer
$j\geq1$:
\begin{equation}
\left\vert \ln\frac{\Pr_{\mathcal{D}}[u_{t}=b]}{\Pr_{\mathcal{D}}%
[u_{t}=b~|~\Delta_{j}]}\right\vert \leq\frac{2j+5}{\sqrt{n}\mu}.
\label{eq:DP_binomial}%
\end{equation}
(The Proposition asserts this for every $X\in\Delta_{j}$; the claim when
conditioning on $\Delta_{j}$ follows by a standard argument.)\bigskip

\textbf{Case analysis.} \ We proceed to analyze each of the cases separately,
beginning with the case $j\in J_{1}$. \ Recall that $\Pr_{\mathcal{D}}%
[u_{t}=1]=\lambda_{t}$. \ By equation (\ref{eq:DP_binomial}), the probability
of $u_{t}=1$ under $\Delta_{j}$ can differ from this by at most an
$e^{\frac{(2j+5)\sqrt{n}}{n\mu}}$ multiplicative factor. \ We conclude that:
\begin{align*}
\frac{\Pr_{\mathcal{D}}[u_{t}=0]}{\Pr_{\mathcal{D}}[u_{t}=0~|~\Delta_{j}]}  &
\leq\frac{1-\lambda_{t}}{1-\lambda_{t}e^{\frac{2j+5}{\sqrt{n}\mu}}}\\
&  =\frac{1-\lambda_{t}e^{\frac{2j+5}{\sqrt{n}\mu}}}{1-\lambda_{t}%
e^{\frac{2j+5}{\sqrt{n}\mu}}}+\frac{\lambda_{t}\left(  e^{\frac{2j+5}{\sqrt
{n}\mu}}-1\right)  }{1-\lambda_{t}e^{\frac{2j+5}{\sqrt{n}\mu}}}\\
&  \leq1+\frac{\lambda_{t}\left(  e^{\frac{2j+5}{\sqrt{n}\mu}}-1\right)
}{1/2}\\
&  \leq e^{\frac{\lambda_{t}(4j+10)}{\sqrt{n}\mu}}.
\end{align*}
Here the second-to-last line holds because for this range of $j$ we have
$2j+5\leq\sqrt{n}\mu$, and thus $\lambda_{t}e^{\frac{2j+5}{\sqrt{n}\mu}}%
<\frac{1}{2}$. \ The last line holds because for the same range of $j$ we have%
\[
e^{\frac{2j+5}{\sqrt{n}\mu}}-1\leq\frac{4j+10}{\sqrt{n}\mu}.
\]
To conclude the analysis of the first case, observe that for similar reasons
also in the other direction we have:
\[
\frac{\Pr_{\mathcal{D}}[u_{t}=0]}{\Pr_{\mathcal{D}}[u_{t}=0~|~\Delta_{j}]}\geq
e^{\frac{-\lambda_{t}(4j+10)}{\sqrt{n}\mu}}.
\]

For the second case, $j\in J_{2}$, we have $\lambda_{t}e^{\frac{2j+5}{\sqrt
{n}\mu}}\leq\frac{1}{e}$ and thus:
\begin{align*}
\frac{\Pr_{\mathcal{D}}[u_{t}=0]}{\Pr_{\mathcal{D}}[u_{t}=0~|~\Delta_{j}]}  &
\leq\frac{1-\lambda_{t}}{1-\lambda_{t}e^{\frac{2j+5}{\sqrt{n}\mu}}}\\
&  \leq\frac{1}{1-\lambda_{t}e^{\frac{2j+5}{\sqrt{n}\mu}}}\\
&  \leq e^{2\lambda_{t}e^{(2j+5)/\left(  \sqrt{n}\mu\right)  }}.
\end{align*}
Here the last line holds because for all $z\in(0,e^{-1}]$ we have $\frac
{1}{1-z}\leq e^{-2z}$. \ In the other direction:
\begin{align*}
\frac{\Pr_{\mathcal{D}}[u_{t}=0]}{\Pr_{\mathcal{D}}[u_{t}=0~|~\Delta_{j}]}  &
\geq\Pr_{\mathcal{D}}[u_{t}=0]\\
&  =1-\lambda_{t}\\
&  \geq e^{-2\lambda_{t}}.
\end{align*}

The third case, $j\in J_{3}$, follows immediately from equation
(\ref{eq:DP_binomial}) (or Proposition \ref{prop:log_ratio_j}), which holds
for every possible value of $j$.\bigskip

\textbf{Bounding the KL-divergence.} \ We now proceed to bound the
KL-divergence between $\mathcal{D}$ and $\mathcal{D}_{0}$
\begin{align*}
\operatorname{KL}\left(  \mathcal{D},\mathcal{D}_{0}\right)   &  =\sum
_{X\in\{0,1\}^{n}}\Pr_{\mathcal{D}}[X]\ln\frac{\Pr_{\mathcal{D}}[X]}%
{\Pr_{\mathcal{D}_{0}}[X]}\\
&  =\sum_{j}\left(  \sum_{X\in\Delta_{j}}\Pr_{\mathcal{D}}[X]\ln\frac
{\Pr_{\mathcal{D}}[X]}{\Pr_{\mathcal{D}_{0}}[X]}\right) \\
&  =\sum_{j}\left(  \sum_{X\in\Delta_{j}}\Pr_{\mathcal{D}}[\Delta_{j}%
]\Pr_{(\mathcal{D}|\Delta_{j})}[X]\ln\frac{\Pr_{\mathcal{D}}[\Delta_{j}%
]\Pr_{(\mathcal{D}|{\Delta_{j}})}[X]}{\Pr_{\mathcal{D}_{0}}[\Delta_{j}%
]\Pr_{(\mathcal{D}_{0}|\Delta_{j})}[X]}\right) \\
&  =\sum_{j}\left(  \Pr_{\mathcal{D}}[\Delta_{j}]\ln\frac{\Pr_{\mathcal{D}%
}[\Delta_{j}]}{\Pr_{\mathcal{D}_{0}}[\Delta_{j}]}+\Pr_{\mathcal{D}}[\Delta
_{j}]\sum_{X\in\Delta_{j}}\Pr_{(\mathcal{D}|\Delta_{j})}[X]\ln\frac
{\Pr_{(\mathcal{D}|{\Delta_{j}})}[X]}{\Pr_{(\mathcal{D}_{0}|\Delta_{j})}%
[X]}\right) \\
&  =\sum_{j}\left(  \Pr_{\mathcal{D}}[\Delta_{j}]\ln\frac{\Pr_{\mathcal{D}%
}[\Delta_{j}]}{\Pr_{\mathcal{D}_{0}}[\Delta_{j}]}+\Pr_{\mathcal{D}}[\Delta
_{j}]\operatorname{KL}\left(  (\mathcal{D}|{\Delta_{j}}),(\mathcal{D}%
_{0}|{\Delta_{j}})\right)  \right)
\end{align*}
And similarly:
\[
\operatorname{KL}\left(  \mathcal{D}_{0},\mathcal{D}\right)  =\sum_{j}\left(
\Pr_{\mathcal{D}_{0}}[\Delta_{j}]\ln\frac{\Pr_{\mathcal{D}_{0}}[\Delta_{j}%
]}{\Pr_{\mathcal{D}}[\Delta_{j}]}+\Pr_{\mathcal{D}_{0}}[\Delta_{j}%
]\operatorname{KL}\left(  (\mathcal{D}_{0}|{\Delta_{j}}),(\mathcal{D}%
|{\Delta_{j}})\right)  \right)
\]
Using the nonnegativity of KL-divergence, together with the bound in
Proposition \ref{prop:KL_j}, we conclude that:
\begin{align*}
\operatorname{KL}\left(  \mathcal{D},\mathcal{D}_{0}\right)  \leq &
\operatorname{KL}\left(  \mathcal{D},\mathcal{D}_{0}\right)
+\operatorname{KL}\left(  \mathcal{D}_{0},\mathcal{D}\right) \\
=  &  \sum_{j}(\Pr_{\mathcal{D}_{0}}[\Delta_{j}]-\Pr_{\mathcal{D}}[\Delta
_{j}])\ln\frac{\Pr_{\mathcal{D}_{0}}[\Delta_{j}]}{\Pr_{\mathcal{D}}[\Delta
_{j}]}\\
&  +\sum_{j}\left(  \Pr_{\mathcal{D}}[\Delta_{j}]+\Pr_{\mathcal{D}_{0}}%
[\Delta_{j}]\right)  \frac{2\min\left\{  1,16e^{\frac{4j+10}{\sqrt{n}\mu}%
}\cdot\lambda_{t}^{2}\right\}  }{n\mu^{2}}.
\end{align*}
Below, we show that each of these two sums is bounded by $O\left(
\frac{\lambda_{t}^{2}}{n\mu^{2}}\right)  $. \ We conclude that
\[
\left\Vert |\psi\rangle\langle\psi|-|\psi_{\text{no}}\rangle\langle
\psi_{\text{no}}|\right\Vert _{\operatorname*{tr}}\leq\sqrt{2\operatorname{KL}%
\left(  \mathcal{D},{\mathcal{D}_{0}}\right)  }=O\left(  \frac{\lambda_{t}%
}{\sqrt{n}\mu}\right)  ,
\]
which completes the proof of Claim \ref{claim:noupdate_gentle}.\bigskip
\end{proof}

\medskip\noindent\textbf{Bounding the first sum.} \ We divide the sum over $j$
into the three segments defined in Proposition \ref{prop:log_ratio_j}, and use
the bound on the log-ratio to bound the sum over each of the segments.
\ Starting with the first segment $J_{1}$:
\begin{align*}
\sum_{j\in J_{1}}(\Pr_{\mathcal{D}_{0}}[\Delta_{j}]-\Pr_{\mathcal{D}}%
[\Delta_{j}])\ln\frac{\Pr_{\mathcal{D}_{0}}[\Delta_{j}]}{\Pr_{\mathcal{D}%
}[\Delta_{j}]}  &  \leq\sum_{j\in J_{1}}\Pr_{\mathcal{D}}[\Delta_{j}]\left(
e^{\frac{\lambda_{t}(4j+10)}{\sqrt{n}\mu}}-1\right)  \frac{\lambda_{t}%
(4j+10)}{\sqrt{n}\mu}\\
&  \leq8\sum_{j\in J_{1}}\Pr_{\mathcal{D}}[\Delta_{j}]\frac{\lambda_{t}%
^{2}(4j+10)^{2}}{n\mu^{2}}\\
&  =8\frac{\lambda_{t}^{2}}{n\mu^{2}}\sum_{j\in J_{1}}\Pr_{\mathcal{D}}%
[\Delta_{j}](4j+10)^{2}\\
&  =\Theta\left(  \frac{\lambda_{t}^{2}}{n\mu^{2}}\right)  .
\end{align*}
Here the first line uses Proposition \ref{prop:log_ratio_j}; the second line
uses the fact that $4j+10\leq2\sqrt{n}\mu$\ for all $j\in J_{1}$; and the last
line uses a moment bound on the distribution of $j$, namely inequality
(\ref{eq:Delta_moment}).

For the segment $J_{2}$ we have:
\begin{align*}
\sum_{j\in J_{2}}(\Pr_{\mathcal{D}_{0}}[\Delta_{j}]-\Pr_{\mathcal{D}}%
[\Delta_{j}])\ln\frac{\Pr_{\mathcal{D}_{0}}[\Delta_{j}]}{\Pr_{\mathcal{D}%
}[\Delta_{j}]}  &  \leq\sum_{j\in J_{2}}\Pr_{\mathcal{D}}[\Delta_{j}]\left(
e^{2\lambda_{t}e^{(2j+5)/\left(  \sqrt{n}\mu\right)  }}-1\right)  2\lambda
_{t}e^{\frac{2j+5}{\sqrt{n}\mu}}\\
&  \leq8\lambda_{t}^{2}\sum_{j\in J_{2}}\Pr_{\mathcal{D}}[\Delta_{j}%
]e^{\frac{4j+10}{\sqrt{n}\mu}}\\
&  \leq16\lambda_{t}^{2}\sum_{j\in J_{2}}\exp\left(  \frac{4j+10}{\sqrt{n}\mu
}-\frac{(j-1)^{2}}{2}\right) \\
&  =16\lambda_{t}^{2}\sum_{j\in J_{2}}\exp\left(  -\Theta\left(  j^{2}\right)
\right) \\
&  =16\lambda_{t}^{2}\exp\left(  -\Theta\left(  n\mu^{2}\right)  \right) \\
&  =O\left(  \frac{\lambda_{t}^{2}}{n\mu^{2}}\right)  .
\end{align*}
Here the first line uses Proposition \ref{prop:log_ratio_j}, the second uses
the fact that $\frac{2j+5}{\sqrt{n}\mu}\leq\ln(\frac{1}{\lambda_{t}})-1$\ for
all $j\in J_{2}$; the third uses inequality (\ref{eq:Delta_prob})
(concentration of $j$); and the fourth and fifth use the facts that
$j=\Omega(\sqrt{n}\mu)$ for all $j\in J_{2}$\ and that$\sqrt{n}\mu$ is a
sufficiently large constant.

For the segment $J_{3}$ we have:
\begin{align*}
\sum_{j\in J_{3}}(\Pr_{\mathcal{D}_{0}}[\Delta_{j}]-\Pr_{\mathcal{D}}%
[\Delta_{j}])\ln\frac{\Pr_{\mathcal{D}_{0}}[\Delta_{j}]}{\Pr_{\mathcal{D}%
}[\Delta_{j}]}  &  \leq\sum_{j\in J_{3}}\Pr_{\mathcal{D}}[\Delta_{j}]\left(
e^{\frac{2j+5}{\sqrt{n}\mu}}-1\right)  \frac{2j+5}{\sqrt{n}\mu}\\
&  \leq2\sum_{j\in J_{3}}e^{-(j-1)^{2}/2}\left(  e^{\frac{2j+5}{\sqrt{n}\mu}%
}-1\right)  \frac{2j+5}{\sqrt{n}\mu}\\
&  =\sum_{j\in J_{3}}\exp\left(  -\Theta(j^{2})\right) \\
&  =\exp\left(  -\Theta(n\mu^{2}\ln^{2}\frac{1}{\lambda_{t}}\right) \\
&  =O\left(  \frac{\lambda_{t}^{2}}{n\mu^{2}}\right)  .
\end{align*}
Here the first line uses Proposition \ref{prop:log_ratio_j}; the second uses
inequality (\ref{eq:Delta_prob}) (concentration of $j$); and the third and
fourth use the facts that $j=\Omega(\sqrt{n}\mu\ln\frac{1}{\lambda_{t}})$ for
all $j\in J_{3}$\ and that $\sqrt{n}\mu$\ is a sufficiently large
constant.\bigskip

\medskip\noindent\textbf{Bounding the second sum.} \ Similarly to the first
sum, we divide the sum over $j$ into the three segments defined in Proposition
\ref{prop:log_ratio_j}, and use the bound on the log-ratio to bound the sum
over each of the segments. \ Starting with the first segment $J_{1}$:
\begin{align*}
\sum_{j\in J_{1}}\left(  \Pr_{\mathcal{D}}[\Delta_{j}]+\Pr_{\mathcal{D}_{0}%
}[\Delta_{j}]\right)  \frac{\min\left\{  2,32e^{\frac{4j+10}{\sqrt{n}\mu}%
}\cdot\lambda_{t}^{2}\right\}  }{n\mu^{2}}  &  \leq\sum_{j\in J_{1}}%
\Pr_{\mathcal{D}}[\Delta_{j}]\left(  1+e^{\frac{\lambda_{t}(4j+10)}{\sqrt
{n}\mu}}\right)  \frac{32e^{\frac{4j+10}{\sqrt{n}\mu}}\cdot\lambda_{t}^{2}%
}{n\mu^{2}}\\
&  =\Theta\left(  \frac{\lambda_{t}^{2}}{n\mu^{2}}\right)  .
\end{align*}
Here the first inequality uses Proposition \ref{prop:log_ratio_j}, while the
second uses the fact that $4j+10\leq2\sqrt{n}\mu$\ for all $j\in J_{1}$.

For the segment $J_{2}$ we have:
\begin{align*}
\sum_{j\in J_{2}}\left(  \Pr_{\mathcal{D}}[\Delta_{j}]+\Pr_{\mathcal{D}_{0}%
}[\Delta_{j}]\right)  \frac{\min\left\{  2,32e^{\frac{4j+10}{\sqrt{n}\mu}%
}\cdot\lambda_{t}^{2}\right\}  }{n\mu^{2}}  &  \leq\sum_{j\in J_{2}}%
\Pr_{\mathcal{D}}[\Delta_{j}]\left(  1+e^{2\lambda_{t}e^{(2j+5)/\left(
\sqrt{n}\mu\right)  }}\right)  \frac{32e^{\frac{4j+10}{\sqrt{n}\mu}}%
\cdot\lambda_{t}^{2}}{n\mu^{2}}\\
&  \leq\frac{128\lambda_{t}^{2}}{n\mu^{2}}\sum_{j\in J_{2}}\Pr_{\mathcal{D}%
}[\Delta_{j}]e^{\frac{4j+10}{\sqrt{n}\mu}}\\
&  \leq\frac{256\lambda_{t}^{2}}{n\mu^{2}}\sum_{j\in J_{2}}\exp\left(
\frac{4j+10}{\sqrt{n}\mu}-\frac{(j-1)^{2}}{2}\right) \\
&  =O\left(  \frac{\lambda_{t}^{2}}{n\mu^{2}}\right)  .
\end{align*}
Here the first line uses Proposition \ref{prop:log_ratio_j}; the second uses
the fact that $\frac{2j+5}{\sqrt{n}\mu}\leq\frac{1}{e\lambda_{t}}$\ for all
$j\in J_{2}$; the third uses equation (\ref{eq:Delta_prob}); and the last uses
the facts that $j=\Omega(\sqrt{n}\mu)$ for all $j\in J_{2}$\ and that
$\sqrt{n}\mu$ is a sufficiently large constant.

Finally, for the segment $J_{3}$ we have:
\begin{align*}
\sum_{j\in J_{3}}\left(  \Pr_{\mathcal{D}}[\Delta_{j}]+\Pr_{\mathcal{D}_{0}%
}[\Delta_{j}]\right)  \frac{\min\left\{  2,32e^{\frac{4j+10}{\sqrt{n}\mu}%
}\cdot\lambda_{t}^{2}\right\}  }{n\mu^{2}}  &  \leq\sum_{j\in J_{3}}%
\Pr_{\mathcal{D}}[\Delta_{j}]\left(  1+e^{\frac{2j+5}{\sqrt{n}\mu}}\right)
\frac{2}{n\mu^{2}}\\
&  \leq\frac{4}{n\mu^{2}}\sum_{j\in J_{3}}\exp\left(  \frac{2j+5}{\sqrt{n}\mu
}-\frac{(j-1)^{2}}{2}\right) \\
&  =\frac{4}{n\mu^{2}}\exp\left(  -\Theta\left(  n\mu^{2}\ln^{2}\frac
{1}{\lambda_{t}}\right)  \right) \\
&  =O\left(  \frac{\lambda_{t}^{2}}{n\mu^{2}}\right)  .
\end{align*}
Here the first line uses Proposition \ref{prop:log_ratio_j}; the second uses
equation (\ref{eq:Delta_prob}); and the third uses the facts that
$j=\Omega\left(  \sqrt{n}\mu\ln\frac{1}{\lambda_{t}}\right)  $\ for all $j\in
J_{3}$\ and that $\sqrt{n}\mu$ is a sufficiently large constant.
\end{proof}

\subsection{Lower Bounds for Shadow Tomography\label{LBGENTLE}}

To recap, the QPMW algorithm lets us do shadow tomography on a $d$-dimensional
state $\rho$, with respect to two-outcome measurements $E_{1},\ldots,E_{m}$
and with accuracy $\pm\varepsilon$, in a way that moreover is online and
gentle, by measuring $O\left(  \left(  \log m\right)  ^{2}\left(  \log
d\right)  ^{2}/\varepsilon^{8}\right)  $ copies of $\rho$. \ How close to
optimal is this upper bound?

The only known general lower bound for shadow tomography, due to Aaronson
\cite{aar:shadow}, says that $\Omega\left(  \frac{\min\left\{  d^{2},\log
m\right\}  }{\varepsilon^{2}}\right)  $\ copies of $\rho$\ are needed, for
information-theoretic reasons. \ Aaronson \cite{aar:shadow} also shows that,
in the special case where the states and measurements are entirely classical,
$\Theta\left(  \frac{\min\left\{  d,\log m\right\}  }{\varepsilon^{2}}\right)
$\ copies are necessary and sufficient.\footnote{The original conference
version of \cite{aar:shadow}\ proved only a weaker lower bound: namely,
$\Omega\left(  \frac{\log m}{\varepsilon^{2}}\right)  $\ when $d$ can be
arbitrarily large (including for the classical special case). \ However, the
most recent arXiv version includes the stated bounds, the ones that explicitly
incorporate dependence on the dimension $d$.} \ In the general, quantum
setting, it remains open whether there could exist a shadow tomography
procedure that used only $\left(  \log m\right)  ^{O\left(  1\right)  }$
copies, independent of the dimension $d$.

In this section, we won't resolve that problem. \ However, as yet another
application of our connection between DP and gentleness, we'll observe a lower
bound on the sample complexity of \textit{gentle} shadow tomography, which
applies even to offline algorithms---i.e., ones that see all the measurements
in advance. \ And conversely, by using the connection to adaptive data
analysis, we'll use known results in that setting to give a lower bound for
\textit{online} shadow tomography, which applies even to non-gentle algorithms.

We stress that, while these lower bounds use nontrivial recent results, they
have nothing to do with quantum mechanics: all of them apply even to the
\textquotedblleft classical special case\textquotedblright\ of shadow
tomography, wherein the input consists of i.i.d.\ samples from a single
distribution and the \textquotedblleft measurements\textquotedblright\ are all
in the computational basis. \bigskip

\noindent\textbf{Gentle shadow tomography.} \ The first result we state is a
lower bound for \emph{gentle} shadow tomography, even in the offline setting:

\begin{theorem}
[Lower Bound for Gentle Shadow Tomography]\label{sqrtlogdlb}Any shadow
tomography procedure that is $\left(  \alpha,\frac{1}{n^{1+\tau}}\right)
$-gentle\ for a constant $\tau>0$ on all product states, and is also
$\varepsilon$-accurate on states of the form $\rho^{\otimes n}$, requires
\[
n=\widetilde{\Omega}\left(  \frac{\left(  \log m\right)  \sqrt{\log d}%
}{\varepsilon^{2}\alpha}\right)
\]
samples.
\end{theorem}

In other words, as long as we insist that our shadow tomography procedure be
$\left(  \alpha,\delta\right)  $-gentle for small $\delta$---with gentleness
applying for all product states, as usual in this paper---the sample
complexity of the QPMW algorithm is optimal up to a polynomial factor.

We'll deduce Theorem \ref{sqrtlogdlb}\ as a corollary of the following result
of Bun, Ullman, and Vadhan \cite{buv}:

\begin{theorem}
[Bun et al.\ \cite{buv}]\label{buvthm}For all $m,n,d$, there exist $m$ Boolean
functions $f_{1},\ldots,f_{m}:\left[  d\right]  \rightarrow\left\{
0,1\right\}  $, such that no $\left(  \gamma,\frac{1}{10n}\right)  $-DP
algorithm can, for all databases $X=\left(  x_{1},\ldots,x_{n}\right)
\in\left[  d\right]  ^{n}$, estimate $\operatorname{E}_{j\in\left[  n\right]
}\left[  f_{i}\left(  x_{j}\right)  \right]  $\ to within additive error
$\pm\varepsilon$, for every $i\in\left[  m\right]  $ and with success
probability at least $2/3$, unless%
\[
n=\widetilde{\Omega}\left(  \frac{\left(  \log m\right)  \sqrt{\log d}%
}{\varepsilon^{2}\gamma}\right)  .
\]

\end{theorem}

The proof of Theorem \ref{buvthm}\ uses so-called \textit{fingerprinting
codes} to construct the functions $f_{1},\ldots,f_{m}$. \ We omit the details;
see for example Vadhan \cite[Section 5.3]{vadhan:dp}\ for further discussion
of this technique.

Recall Lemma \ref{gentledpprod}, which said that any measurement that is
$\left(  \alpha,\delta\right)  $-gentle on product states is also $\left(
\ln\left(  \frac{1+4\alpha}{1-4\alpha}\right)  ,\delta\right)  $-DP on product
states. \ In the classical special case, the latter simply means $\left(
\ln\left(  \frac{1+4\alpha}{1-4\alpha}\right)  ,\delta\right)  $-DP in the
usual sense. \ By just combining this implication with Theorem \ref{buvthm},
we immediately obtain a lower bound on the sample complexity of \textit{some}
form of gentle shadow tomography, even in the classical special case.
\ However, there is still a difficulty. \ Namely, the lower bound that we get
will apply only to shadow tomography algorithms that remain accurate in what
we call the \textit{diverse-state setting}. \ This is the setting where the
algorithm is given a sample from a product distribution $\mathcal{D}%
=\mathcal{D}_{1}\times\cdots\times\mathcal{D}_{n}$---or in the quantum case, a
product state $\rho_{1}\otimes\cdots\otimes\rho_{n}$---and its goal is to
estimate the acceptance probability of each of the two-outcome measurements
$E_{1},\ldots,E_{m}$\ on the average state%
\[
\widehat{\rho}:=\frac{\rho_{1}+\cdots+\rho_{n}}{n}.
\]
By contrast, we defined shadow tomography for what we call the
\textit{identical-state setting}: that is, the setting where we're
additionally promised that $\rho_{1}=\cdots=\rho_{n}$, so that the input state
has the special form $\rho^{\otimes n}$. \ All of the shadow tomography
procedures that we know, including QPMW, are accurate even in the more general
diverse-state setting. \ But it's not obvious that \textit{lower bounds} in
the diverse-state setting carry over to the identical-state setting, so there
is still a gap to close. \ We close the gap using the simple claim below,
which translates accuracy in the identical-state setting to accuracy in the
diverse-state setting, with only a small loss in the differential privacy
parameters.\footnote{We note that it might be possible to obtain a lower bound
similar to that of Theorem \ref{buvthm} that directly applies to the
identical-state setting (see, e.g., Steinke and Ullman \cite[Corollary
15]{SteinkeU17a}). \ Still, the transformation we outline incurs only a small
loss in the parameters, and works more generally.}

\begin{claim}
\label{claim:arbitrary2iid}Fix a data universe $[d]$, functions $f_{1}%
,\ldots,f_{m}:\left[  d\right]  \rightarrow\left\{  0,1\right\}  $, and a
database size $n$. \ Let $\mathcal{A}$ be a classical algorithm that's
$(\alpha,\delta)$-DP in the usual classical sense, and satisfies the following
accuracy guarantee: for any distribution $\mathcal{D}$ over $[d]$, with all
but $\beta$ probability over $X$ drawn from $\mathcal{D}^{n}$ (and the
algorithm's coins), $\mathcal{A}$'s answers are all within $\pm\varepsilon$ of
the correct answers $\operatorname{E}_{j\in\left[  n\right]  }\left[
f_{i}\left(  X_{j}\right)  \right]  $.

Then there exists another algorithm $\mathcal{B}$, which runs in $O(|X|)$ time
using a single oracle call to $\mathcal{A}$, such that $\mathcal{B}$ is
$(\frac{3\alpha\log n}{\log\log n},\delta\cdot n^{o(1)}+\frac{1}{n^{2}})$-DP,
and for any fixed database $X\in\lbrack d]^{n}$, with all but $\beta$
probability over the $\mathcal{B}$'s coins, $\mathcal{B}$'s answers are all
within $\varepsilon$ of the correct answers $\operatorname{E}_{j\in\left[
n\right]  }\left[  f_{i}\left(  X_{j}\right)  \right]  $.
\end{claim}

\begin{proof}
Given an input database $X$, the algorithm $\mathcal{B}$ operates by taking
$n$ i.i.d.\ samples, with replacement, from the distribution $\mathcal{D}_{X}$
that is uniform over the entries of $X$ (a distribution whose support has size
at most $n$). \ It then runs $\mathcal{A}$ on the resulting database
$X^{\prime}$ and outputs the results.

Accuracy follows because we are running $\mathcal{A}$ on a sample from
$\mathcal{D}_{X}^{n}$, so with all but $\beta$ probability over the samples
and $\mathcal{A}$'s coins, the answers will all be within $\pm\varepsilon$ of
the correct expectations over $\mathcal{D}_{X}$, which are the correct answers
on the database $X$.

For privacy, fix adjacent databases $X$ and $Y$ that differ only in the
$i^{th}$ entry. \ For fixed coins used to choose i.i.d.\ samples, let
$X^{\prime}$ and $Y^{\prime}$ be the databases produced by sampling from $X$
or from $Y$ respectively. \ All entries in $X^{\prime}$ and $Y^{\prime}$ will
be identical, except those that are copies of the $i^{th}$ entry. \ By a balls
and bins argument, with all but $1/n^{2}$ probability, the number of copies of
the $i^{th}$ entry is at most $\frac{3\log n}{\log\log n}$. \ Whenever this is
the case, the group privacy guarantees that follow from the differential
privacy of $\mathcal{A}$ imply that the probability of any event differs by at
most a $\frac{3\alpha\log n}{\log\log n}$ multiplicative factor and a
$\delta\cdot n^{o(1)}$ additive error.
\end{proof}

We can now complete the proof of Theorem \ref{sqrtlogdlb}.

\begin{proof}
[Proof of Theorem \ref{sqrtlogdlb}]Let $\mathcal{P}$ be a\ shadow tomography
procedure that is $\left(  \alpha,\frac{1}{n^{1+\tau}}\right)  $-gentle\ on
product states $\rho_{1}\otimes\cdots\otimes\rho_{n}$, for small $\alpha$ and
fixed $\tau>0$. \ By Lemma \ref{gentledpprod}, this $\mathcal{P}$ is also
$\left(  \gamma,\delta\right)  $-DP on product states, for $\gamma=\ln\left(
\frac{1+4\alpha}{1-4\alpha}\right)  =O\left(  \alpha\right)  $.

Henceforth, we restrict attention to $\mathcal{P}$'s behavior on classical
inputs $X\in\left[  d\right]  ^{n}$. \ Here, being DP on product states simply
reduces to the usual notion of DP.

Now suppose further that $\mathcal{P}$ is $\varepsilon$-accurate in the
identical-state setting. \ Then by Claim \ref{claim:arbitrary2iid}, we can
obtain a new classical procedure $\mathcal{P}^{\prime}$ that is $\left(
\gamma^{\prime}=\frac{3\gamma\log n}{\log\log n},\delta^{\prime o(1)}%
=\delta+\frac{1}{n^{2}}\right)  $-DP, and that moreover is $\varepsilon
$-accurate for any given database $X\in\left[  d\right]  ^{n}$. \ But this
means that $n$ must satisfy the bound of Theorem \ref{buvthm}. \ We use here
the fact that for $\delta=\frac{1}{n^{1+\tau}}$ we get $\delta^{\prime
}=o(\frac{1}{n})$. \ We note that the $O(\frac{\log n}{\log\log n})$
deterioration in the privacy guarantee of $\mathcal{P}^{\prime}$ (compared to
$\mathcal{P}$) is accounted for by the tilde in the $\widetilde{\Omega}$.
\end{proof}

As noted above, Theorem \ref{sqrtlogdlb} applies even to the \textquotedblleft
classical special case\textquotedblright\ of shadow tomography. \ In that
special case, the Chernoff bound immediately implies a procedure with
$O\left(  \frac{\log m}{\varepsilon^{2}}\right)  $\ sample complexity. \ Thus,
one implication of Theorem \ref{sqrtlogdlb} is that such a procedure
necessarily violates gentleness---where \textquotedblleft
gentleness,\textquotedblright\ here, means a bound on the damage in variation
distance caused by classical Bayesian updating.\bigskip

\noindent\textbf{Online shadow tomography.} \ The second result we state is a
lower bound for \emph{online} shadow tomography, even without gentleness:

\begin{theorem}
[Lower Bound for Online Shadow Tomography \cite{nsssu}]\label{adatoST}Any
online shadow tomography procedure that is $\varepsilon$-accurate requires
sample complexity%
\[
n=\widetilde{\Omega}\left(  \frac{\sqrt{\min\left\{  m,\log d\right\}  }%
}{\varepsilon}\right)  .
\]

\end{theorem}

Combining Theorem \ref{adatoST} with the $\Omega\left(  \frac{\min\left\{
d,\log m\right\}  }{\varepsilon^{2}}\right)  $ lower bound of Aaronson
\cite{aar:shadow}, we can conclude that online shadow tomography requires%
\[
n=\widetilde{\Omega}\left(  \frac{\log m+\sqrt{\log d}}{\varepsilon}\right)
\]
copies of $\rho$\ unless $n=m^{\Omega\left(  1\right)  }$\ or $n=d^{\Omega
\left(  1\right)  }$. \ Hence QPMW achieves the optimal sample complexity for
online shadow tomography up to polynomial factors.

Theorem \ref{adatoST} again has nothing to do with quantum mechanics, and
follows immediately from known lower bounds for classical adaptive data
analysis. \ There, an algorithm processes a collection of states that are
drawn i.i.d.\ from an underlying distribution, and the goal is to provide
accurate answers with respect to the underlying distribution---and in
particular, to avoid overfitting to the specific sample. \ Adaptive data
analysis is thus a special case of online shadow tomography in the
identical-state setting. \ Hardt and Ullman \cite{hardtullman} and Steinke and
Ullman \cite{steinkeullman} showed sample complexity lower bounds and
computational hardness results for this setting. \ Theorem \ref{adatoST} is a
restatement, in our setting, of a recent result of Nissim et al.\ \cite{nsssu}.

\section{Computational Efficiency\label{COMPUT}}

So far, our results have been purely information-theoretic. \ When we talked,
for example, about a gentle \textquotedblleft implementation\textquotedblright%
\ of a measurement $M$, we were concerned only about whether such an
implementation existed, not about its time complexity. \ Likewise, the QPMW
procedure for shadow tomography was efficient in sample complexity, but we
weren't concerned to bound its computation time.

Now, at last, we consider to what extent our constructions are (or can be
made) computationally efficient. \ In Section \ref{COMPUTDP}, we'll explain
why gentle measurements \textit{can} be implemented in polynomial time,
provided we have an efficient way to uncompute garbage, and we'll give several
classes of examples where this can be done. \ Then, in Section \ref{COMPUTST},
we'll use our results from Section \ref{COMPUTDP}\ to examine the
computational complexity of the QPMW procedure. \ Finally, in Section
\ref{PROMISEBQP}, we'll turn things around, and observe how gentle
measurements like the ones in this paper, whether derived from DP algorithms
or not, can be applied to the safe implementation of subroutines in quantum algorithms.

\subsection{Efficiency of DP and Gentle Measurements\label{COMPUTDP}}

Let's start with Theorem \ref{main}, the connection between gentleness and DP.
\ For part (1) of the theorem, namely that $\alpha$-gentleness implies
$O\left(  \alpha\right)  $-DP for small $\alpha$, there's no issue of
computational efficiency. \ This is because the very same measurement
procedure that achieves $\alpha$-gentleness also achieves $O\left(
\alpha\right)  $-DP---the latter being solely a property of the output
probabilities, which has nothing to do with the post-measurement states.

On the other hand, for part (2) of the theorem, namely that $\varepsilon$-DP
on product states implies $O\left(  \varepsilon\sqrt{n}\right)  $%
-gentleness\ on product states for small $\varepsilon$ (and product
measurements), there \textit{is} a computational issue. \ Namely: even if our
original $\varepsilon$-DP measurement $M$\ could be implemented by a
polynomial-size circuit, the proof of Theorem \ref{main}\ might return an
implementation of $M$ that is $O\left(  \varepsilon\sqrt{n}\right)  $-gentle
but that does \textit{not} correspond to any polynomial-size circuit. \ Yet,
while this is a problem in principle, fortunately it turns out not to be a
problem for any of the measurements that have concerned us in this paper,
including the ones used in our shadow tomography procedure.

The potential computational issue occurs in the proof of Lemma
\ref{dpgentleprojective}. \ There, given a classical DP algorithm $A$, we
needed to map the state%
\[
\left\vert \psi\right\rangle =\sum_{X\in\left[  d\right]  ^{n}}\alpha
_{X}\left\vert X\right\rangle
\]
to%
\[
\sum_{X\in\left[  d\right]  ^{n}:\Pr\left[  X\right]  >0,~y}\alpha
_{X}\left\vert X\right\rangle \sqrt{\Pr\left[  y|X\right]  }\left\vert
y\right\rangle ,
\]
where the $y$'s are the possible outcomes of running $A$ on the input database
$X$. \ Assuming that $A$ itself is computationally efficient, it's easy to
prepare a state of the form%
\[
\sum_{X\in\left[  d\right]  ^{n}:\Pr\left[  X\right]  >0,~y}\alpha
_{X}\left\vert X\right\rangle \sqrt{\Pr\left[  y|X\right]  }\left\vert
y\right\rangle \left\vert g_{X,y}\right\rangle ,
\]
where $\left\vert g_{X,y}\right\rangle $\ is \textquotedblleft
garbage\textquotedblright\ entangled with the $\left\vert X\right\rangle $ and
$\left\vert y\right\rangle $\ registers (for example, the outcomes of coin
flips made by $A$). \ The entire difficulty lies in uncomputing the
$\left\vert g_{X,y}\right\rangle $\ register. \ If we fail to uncompute, then
the effect on $\left\vert \psi\right\rangle $\ might no longer be gentle.

As we mentioned in Section \ref{CONNECTION}, an equivalent way to say this is
that our reduction from DP to gentleness preserves efficiency if, and only if,
we have an efficient algorithm to \textquotedblleft QSample\textquotedblright%
\ the output distribution of the DP algorithm $A$, meaning to prepare the
superposition%
\[
\left\vert \phi_{X}\right\rangle :=\sum_{y}\sqrt{\Pr\left[  y|X\right]
}\left\vert y\right\rangle
\]
for a given input $X$. \ In practice, many fast sampling algorithms do give
rise to fast QSampling algorithms, but this need not always be the case.
\ Indeed, as pointed out by Aharonov and Ta-Shma \cite{at} in 2003, if fast
sampling always implied fast QSampling, then we'd immediately get
polynomial-time quantum algorithms for graph isomorphism, breaking
lattice-based cryptosystems, and all other problems in the class
$\mathsf{SZK}$ (Statistical Zero Knowledge). \ Closely related to that, the
collision lower bound of Aaronson \cite{aar:col}\ implies that, in the
black-box setting, fast sampling does \textit{not} imply fast QSampling.

But what about the specific measurements considered in this paper? \ Let's
start with the following observation:

\begin{proposition}
[Efficient Implementation of $L_{\sigma}$]\label{lsigmafast}There is an
$O\left(  n\right)  $-size quantum circuit to implement $L_{\sigma}$, the
Laplace noise\ measurement on $n$ qubits, to $\frac{1}{\exp\left(  n\right)
}$ accuracy, so long as $\sigma=\exp\left(  O\left(  n\right)  \right)  $.
\end{proposition}

\begin{proof}
We simply use the procedure for implementing $L_{\sigma}$\ described in
Section \ref{GENTLE}: the one where, given a superposition over $\left\vert
X\right\rangle $'s, we first prepare a Laplace noise register%
\[
\left\vert \eta\right\rangle :=\frac{1}{Z}\sum_{k=-C}^{C}e^{-\left\vert
k\right\vert /\left(  2\sigma\right)  }\left\vert k\right\rangle ,
\]
for some cutoff $C=\exp\left(  O\left(  n\right)  \right)  $ and normalization
$Z$, then calculate $\left\vert \left\vert X\right\vert +k\right\rangle $, and
finally use $\left\vert X\right\rangle $\ and $\left\vert \left\vert
X\right\vert +k\right\rangle $\ together to uncompute the noise $\left\vert
k\right\rangle $. \ What makes this work is that, in $L_{\sigma}$, the noise
is entirely additive, and addition of integers is an easy operation to invert.

Also, as long as $\sigma=\exp\left(  O\left(  n\right)  \right)  $, a cutoff
of the form $C=\exp\left(  O\left(  n\right)  \right)  $\ suffices for
exponential accuracy. \ Moreover, one can check that DP, and hence gentleness,
still hold even after we impose the cutoff.

It remains only to verify that there are $O\left(  n\right)  $-size quantum
circuits to add and subtract $O\left(  n\right)  $-bit integers, and to
prepare $\left\vert \eta\right\rangle $. \ The one interesting part is
preparing $\left\vert \eta\right\rangle $. \ Omitting normalization and
restricting to $k\geq0$ for simplicity, we observe that%
\[
\sum_{k=0}^{2^{n}-1}e^{-\gamma k}\left\vert k\right\rangle =\left(  \left\vert
0\right\rangle +e^{-\gamma2^{n-1}}\left\vert 1\right\rangle \right)
\otimes\left(  \left\vert 0\right\rangle +e^{-\gamma2^{n-2}}\left\vert
1\right\rangle \right)  \otimes\cdots\otimes\left(  \left\vert 0\right\rangle
+e^{-\gamma}\left\vert 1\right\rangle \right)  ,
\]
from which a linear-size circuit to prepare $\left\vert \eta\right\rangle $ follows.
\end{proof}

Note that the algorithm from Proposition \ref{lsigmafast} is \textquotedblleft
maximally gentle,\textquotedblright\ in the sense that for every possible
state $\rho$ of the $n$ input registers (including non-product states), the
\textit{only} damage that running the algorithm causes to $\rho$, is the
damage that necessarily results from learning the desired output.

We now prove a much more general result, though one that's formally
incomparable to Proposition \ref{lsigmafast}. \ We start with a
trivial-seeming proposition.

\begin{proposition}
\label{twopart}Suppose we have two polynomial-time quantum algorithms: an
algorithm $A$ that, given a classical string $X$, prepares a state $\left\vert
\zeta_{X}\right\rangle $, and an algorithm $B$ that, for some $k=n^{O\left(
1\right)  }$, maps $\left\vert \zeta_{X}\right\rangle ^{\otimes k}$ to
$\left\vert \zeta_{X}\right\rangle ^{\otimes k}\left\vert \phi_{X}%
\right\rangle $, to $\frac{1}{n^{O\left(  1\right)  }}$\ accuracy. \ Then
there's also a polynomial-time quantum algorithm $Q$ that maps $\left\vert
X\right\rangle $\ to $\left\vert X\right\rangle \left\vert \phi_{X}%
\right\rangle $, to $\frac{1}{n^{O\left(  1\right)  }}$\ accuracy.
\end{proposition}

\begin{proof}
We first run $A$ sequentially $k$ times, to map $\left\vert X\right\rangle
$\ to $\left\vert X\right\rangle \left\vert \zeta_{X}\right\rangle ^{\otimes
k}$. \ We next run $B$, to map $\left\vert X\right\rangle \left\vert \zeta
_{X}\right\rangle ^{\otimes k}$\ to $\left\vert X\right\rangle \left\vert
\zeta_{X}\right\rangle ^{\otimes k}\left\vert \phi_{X}\right\rangle $ (to
$\frac{1}{n^{O\left(  1\right)  }}$\ accuracy). \ Finally we run $A^{\dagger}%
$\ sequentially $k$ times, to map $\left\vert X\right\rangle \left\vert
\zeta_{X}\right\rangle ^{\otimes k}\left\vert \phi_{X}\right\rangle $\ to
$\left\vert X\right\rangle \left\vert \phi_{X}\right\rangle $.
\end{proof}

Despite its simplicity, Proposition \ref{twopart}\ lets us efficiently
implement a large class of gentle measurements: namely, any gentle measurement
that admits an efficient \textquotedblleft two-part
algorithm,\textquotedblright\ wherein the first part prepares states
$\left\vert \zeta_{X}\right\rangle $ (which might include unwanted garbage),
and the second part maps the $\left\vert \zeta_{X}\right\rangle $\ states to a
desired output state $\left\vert \phi_{X}\right\rangle $%
\ that---crucially---is nearly unentangled with the $\left\vert \zeta
_{X}\right\rangle $'s, depending only on the original input $X$.

Let's give an example.

\begin{theorem}
[Fast QSampling of Sparse Distributions]\label{qsamplethm}For each input $X$,
suppose the state $\left\vert \phi_{X}\right\rangle $\ has the form%
\[
\left\vert \phi_{X}\right\rangle =\sum_{y\in S_{X}}\sqrt{\Pr\left[
y|X\right]  }\left\vert y\right\rangle ,
\]
where the support sets $S_{X}\subset\left\{  0,1\right\}  ^{m}$\ all satisfy
$\left\vert S_{X}\right\vert \leq\ell$, for some $\ell=n^{O\left(  1\right)
}$ (i.e., the $S_{X}$'s\ are sparse). \ Suppose also that there's an efficient
quantum algorithm $A$\ that, for each $X$, samples---but does \textit{not}
necessarily QSample---the distribution $\mathcal{D}_{X}$\ over $y$ conditional
on $X$. \ Then there's also an efficient quantum algorithm $Q$ that QSamples
$\mathcal{D}_{X}$: that is, maps $\left\vert X\right\rangle \left\vert
0\cdots0\right\rangle $\ to $\left\vert X\right\rangle \left\vert \phi
_{X}\right\rangle $ for each $X$ (up to $\frac{1}{n^{O\left(  1\right)  }}%
$\ error in trace distance).
\end{theorem}

\begin{proof}
As in Proposition \ref{twopart}, the algorithm $Q$\ first runs $A$
sequentially $k$ times, for some sufficiently large $k=n^{O\left(  1\right)
}$. \ It thereby produces the state $\left\vert \zeta_{X}\right\rangle
^{\otimes k}$, where%
\[
\left\vert \zeta_{X}\right\rangle =\sum_{y\in S_{X}}\sqrt{\Pr\left[
y|X\right]  }\left\vert y\right\rangle \left\vert g_{X,y}\right\rangle
\]
is a superposition over samples from $\mathcal{D}_{X}$, possibly entangled
with garbage. \ Next, $Q$ simulates a standard-basis measurement on the
$\left\vert y\right\rangle $\ registers of the $\left\vert \zeta
_{X}\right\rangle $\ states, in order to estimate an empirical frequency for
each possible output string $y\in\left\{  0,1\right\}  ^{m}$. \ (Of course,
all but $n^{O\left(  1\right)  }$\ strings will have an empirical frequency of
$0$ in the sample; for the sake of efficiency, the $0$-frequency strings are
not explicitly recorded.) \ Then, using these empirical frequencies,
$Q$\ prepares the state $\left\vert \phi_{X}\right\rangle $\ to $\frac
{1}{n^{O\left(  1\right)  }}$\ accuracy. \ The efficiency of the preparation
procedure follows from the fact that $\left\vert \phi_{X}\right\rangle $\ has
support of size $\ell=n^{O\left(  1\right)  }$.\footnote{Since we only care
about $\frac{1}{n^{O\left(  1\right)  }}$\ accuracy, in this case we do not
even need the Solovay-Kitaev Theorem (see \cite{nc}).} \ Meanwhile, accuracy
follows by a Chernoff bound and union bound, together with the assumption that
$k$\ was a sufficiently large polynomial compared to $\ell$. \ As the final
step, $Q$ uses $A^{\dagger}$\ to uncompute the $\left\vert \zeta
_{X}\right\rangle $'s.
\end{proof}

As a small special case of Theorem \ref{qsamplethm}, take $\ell=2$ and $m=1$.
\ Then each $\left\vert \phi_{X}\right\rangle $\ has the form $\alpha
_{X}\left\vert 1\right\rangle +\beta_{X}\left\vert 2\right\rangle $, so the
algorithm $A$ could be seen as a $\mathsf{P{}romiseBQP}$\ decision procedure,
which accepts an input $X$\ with probability $\left\vert \beta_{X}\right\vert
^{2}$ (not necessarily bounded away from $\frac{1}{2}$). \ We have shown that
a probabilistic oracle for this decision procedure can be safely implemented
up to $\frac{1}{p\left(  n\right)  }$\ accuracy in polynomial time, for any
polynomial $p$. \ A reasonable interpretation of this\footnote{That is, for
some reasonable definition of what it means to query a $\mathsf{P{}romiseBQP}%
$\ oracle on a superposition of inputs.} is that $\mathsf{BQP}^{\mathsf{P{}%
romiseBQP}}=\mathsf{BQP}$, generalizing the result of Bennett et
al.\ \cite{bbbv}\ that $\mathsf{BQP}^{\mathsf{BQP}}=\mathsf{BQP}$.

Note that, for some DP algorithms, given an input $X\in\left\{  0,1\right\}
^{n}$\ we can just explicitly \textit{calculate} a classical description of
the desired output state $\left\vert \phi_{X}\right\rangle $, to $\frac
{1}{\exp\left(  n\right)  }$\ precision, deterministically and in time
polynomial in $n$. \ If that description also gives rise to a small quantum
circuit to prepare $\left\vert \phi_{X}\right\rangle $, then we can
short-circuit the estimation procedure above, and can improve its accuracy
from $\frac{1}{n^{O\left(  1\right)  }}$\ to $\frac{1}{\exp\left(  n\right)
}$. \ As an example, suppose again that each desired output state $\left\vert
\phi_{X}\right\rangle $\ is a superposition over a sparse set of basis states,
$S_{X}\subset\left\{  0,1\right\}  ^{m}$\ with $\left\vert S_{X}\right\vert
=n^{O\left(  1\right)  }$. But now suppose that, given $X$, we can calculate
both $S_{X}$\ (as a list of elements), and $\left\langle \phi_{X}%
|y\right\rangle $ for each $y\in S_{X}$ to $\frac{1}{\exp\left(  n\right)  }%
$\ precision, in polynomial time. \ Then by using the Solovay-Kitaev Theorem
(see \cite{nc}), we can clearly \textit{prepare} the states $\left\vert
\phi_{X}\right\rangle $---i.e., QSample---in polynomial time as well.

It is not clear how to generalize the above techniques to superpositions
$\left\vert \phi_{X}\right\rangle $\ over exponentially many basis states (or
rather, to do so in any useful generality), even in cases where the individual
amplitudes $\left\langle \phi_{X}|y\right\rangle $\ and probabilities
$\left\vert \left\langle \phi_{X}|y\right\rangle \right\vert ^{2}$\ are
computable in polynomial time.

\subsection{Efficiency of Shadow Tomography\label{COMPUTST}}

What does all of this mean for the computational complexity of shadow
tomography? \ In the QPMW algorithm of Section \ref{SHADOW}, recall that we
needed two types of measurements: threshold measurements\ on all rounds, and
$L_{\sigma}$\ (Hamming weight plus Laplace noise) type measurements on update
rounds. \ Proposition \ref{lsigmafast} has shown that the $L_{\sigma}$
measurements can be implemented in quantum polynomial time, provided the
underlying POVMs $E_{1},\ldots,E_{m}$\ can be implemented in quantum
polynomial time. \ Since a threshold measurement just consists of an
$L_{\sigma}$\ measurement, followed by a binary threshold decision, followed
by uncomputing of garbage, it follows that the threshold measurements can be
implemented in quantum polynomial time as well, again assuming efficient
procedures for the $E_{i}$'s.

Unfortunately, this doesn't mean that QPMW runs in polynomial time overall.
\ The first issue is just the sheer number of measurements $m$. \ Since QPMW
needs one round per measurement, if $m$ is exponentially large then QPMW will
of course need exponential time.

The second issue is the need to maintain, and to do computations on, a
classical description of the current hypothesis state $\sigma_{t}$, in the
online learning procedure \cite{achkn}\ that QPMW uses as a subroutine. \ If
$\sigma_{t}$\ is stored explicitly, as a $d\times d$\ Hermitian matrix, then
this takes $d^{2}$\ space, which is prohibitive if $d$ is exponentially large.
\ However, even if $\sigma_{t}$\ is stored only implicitly, say by a list of
constraints that it satisfies, estimating expectation values
$\operatorname{Tr}\left(  E_{i}\sigma_{t}\right)  $ will still take
$d^{\Theta\left(  1\right)  }$\ time in general.

In summary, if we ignore various low-order contributions, then the running
time of QPMW is roughly $O\left(  mL\right)  +d^{O\left(  1\right)  }$, where
$L$\ is an upper bound on the time needed to implement a single measurement
$E_{i}$. \ By comparison, Aaronson's previous shadow tomography procedure
\cite{aar:shadow} used roughly $O\left(  mL\right)  +d^{O\left(  \log\log
d\right)  }$\ time. \ Thus, QPMW improves the dependence on $d$\ from
quasipolynomial to polynomial.

There is also later work by Brand\~{a}o et al.\ \cite{bkllsw}, which connects
shadow tomography to semidefinite programming and Gibbs states. \ Brand\~{a}o
et al.\ gave a shadow tomography procedure with the same sample complexity as
Aaronson's, and running time $O\left(  \sqrt{m}L\right)  +d^{O\left(
1\right)  }$. \ Here the improvement from $m$ to $\sqrt{m}$\ came from, in
essence, repeatedly doing Grover search over $E_{1},\ldots,E_{m}$ to find an
informative $E_{i}$. \ Thus, if we compare to Brand\~{a}o et al., QPMW matches
the improvement from $d^{O\left(  \log\log d\right)  }$\ to $d^{O\left(
1\right)  }$, but not the improvement from $m$ to $\sqrt{m}$. \ However, this
is to be expected: unlike Aaronson's or Brand\~{a}o et al.'s, our new shadow
tomography procedure is \textit{online}, which necessitates taking time linear
in the number of measurements.

It's natural to wonder: is there some inherent barrier ruling out a shadow
tomography procedure that runs in $\left(  \log d\right)  ^{O\left(  1\right)
}$\ time, avoiding the polynomial dependence on Hilbert space dimension $d$?
\ We now show that there \textit{is} such a barrier---at least if we insist
that the shadow tomography procedure be online, or alternatively, that it be
gentle. \ Our proof will use recent cryptographic lower bounds for
differential privacy and for answering adaptively chosen queries, as well as
our result that gentleness implies DP.\bigskip

\textbf{Hardness for gentle (even offline) shadow tomography.} \ We use a
result of Ullman \cite{Ullman16}, which shows that under plausible
cryptographic assumptions, computing differentially private answers to more
than $\widetilde{\Theta}\left(  n^{2}\right)  $ queries (where $n$ is the
database size) requires time $d^{\Omega(1)}$. \ This hardness result extends
to quantum algorithms, under plausible cryptographic assumptions about their
power. \ Moreover, the result constructs a single distribution $\mathcal{D}$
over $[d]$, such that it's hard for DP algorithms to compute accurate answers
on databases that are drawn i.i.d.\ from $\mathcal{D}$. \ Using our result that
gentleness implies DP, we derive a similar hardness result for \emph{gentle}
shadow tomography.

\begin{theorem}
[Ullman \cite{Ullman16}, quantum variant]\label{ullmanthm}Suppose there exists
a symmetric-key encryption scheme that, for keys of length $\kappa$, is
semantically secure against $2^{\Omega(\kappa)}$-time quantum adversaries.
\ Then there is no quantum algorithm $\mathcal{A}$, running in time
$d^{o(1)}\cdot\mathit{poly}(m)$, that receives as input a database $X$
comprised of $n$ items from $[d]$, and a set of $m=\widetilde{\Theta}\left(
n^{2}\right)  $ queries $E_{1},\ldots,E_{m}$, such that:

\begin{enumerate}
\item[(1)] $\mathcal{A}$ is $(1,\frac{1}{10n})$-DP.

\item[(2)] For any distribution $\mathcal{D}$ over $[d]$, if $X$'s entries are
drawn i.i.d.\ from $\mathcal{D}$, then with all but a small constant
probability over $\mathcal{A}$'s coins and the choice of $X$, for every
$j\in\lbrack m]$, the $j^{th}$ answer $a_{j}$ computed by $\mathcal{A}$
satisfies:%
\[
\left\vert a_{j}-\frac{\sum_{i\in\lbrack n]}{E}_{j}(X_{i})}{n}\right\vert
<\frac{1}{2}.
\]

\end{enumerate}

Moreover, the queries $E_{1},\ldots,E_{m}$ are each computable in
$\mathit{poly}(n,\log d)$ time.
\end{theorem}

Using the fact that gentleness implies differential privacy (Theorem
\ref{main}), we conclude that gentle shadow tomography is hard.

\begin{corollary}
\label{ullmancor}Suppose there exists a symmetric-key encryption scheme that,
for keys of length $\kappa$, is semantically secure against $2^{\Omega
(\kappa)}$-time quantum adversaries. \ Then there is no quantum shadow
tomography procedure that is gentle on product states and runs in $d^{o\left(
1\right)  }\cdot\mathit{poly}\left(  m\right)  $\ time. \ Moreover, this holds
even for the classical special case of shadow tomography.
\end{corollary}

Corollary \ref{ullmancor} applies even to the offline setting, and to
algorithms that are accurate only in the identical-state setting where the
algorithm's input is a state of the form $\rho^{\otimes n}$. \ Moreover, it
applies even for classical data and classical queries. \ We note that Theorem
\ref{ullmanthm} and Corollary \ref{ullmancor} extend to milder cryptographic
assumptions, with a milder conclusion on the possible running time for gentle
shadow tomography. \ Essentially, symmetric key encryption that is hard for
time-$T(\kappa)$ quantum algorithms translates into hardness of differentially
private data analysis for quantum algorithms that run in time $O\left(
T(\kappa)^{\tau}\right)  $, for a fixed constant $\tau>0$. \ Similarly to
Corollary \ref{ullmancor}, the existence of such encryption schemes rules out
gentle shadow tomography in time $T(\log d)^{o(1)}\cdot\mathit{poly}\left(
m,n\right)  $.

Finally, we remark that Theorem \ref{ullmanthm} (and Corollary \ref{ullmancor}%
) do not rule out efficient gentle algorithms that are tailored to fixed
classes of queries---even for exponentially large fixed
classes.\footnote{Theorem \ref{ullmanthm} does not apply because, for the
specific queries used to instantiate the lower bound, the time needed to
compute the queries grows with the database size. In particular, Theorem
\ref{ullmanthm} does not rule out an efficient DP algorithm for answering all
queries that can be computed by $\mathit{poly}(d)$-size circuits. More
generally, for any fixed query family, it does not rule out the possibility of
obtaining an efficient algorithm that is accurate so long as the database is
large enough, and in particular larger than the representation of queries in
the family.} \ Until recently, known DP hardness results for fixed query
families, such as \cite{DworkNRRV09,BonehZ17,KowalczykMUZ16}, relied on
assumptions for which we have no quantum-secure candidate instantiation, such
as bilinear maps or indistinguishability obfuscation. \ A recent result of
Kowalczyk et al.\ \cite{KowalczykMUW18} presents a candidate query family
based on the existence of one-way functions. \ These results may also extend
to gentle shadow tomography.\bigskip

\textbf{Hardness for online (even non-gentle) shadow tomography.} \ We use a
result of Steinke and Ullman \cite{steinkeullman} (building on earlier work by
Hardt and Ullman \cite{HardtU14}), showing that under plausible cryptographic
assumptions, given $n$ i.i.d.\ samples from a distribution $\mathcal{D}$ over
$[d]$, it is computationally hard to answer more than $O(n^{2})$
adaptively-chosen queries accurately. \ Under appropriate assumptions, this
result extends to quantum algorithms, and shows hardness for time
$d^{\Omega(1)}$:

\begin{theorem}
[Steinke and Ullman \cite{steinkeullman}, quantum variant]\label{suthm}Suppose
there exists a symmetric-key encryption scheme that, for keys of length
$\kappa$, is semantically secure against $2^{\Omega(\kappa)}$-time quantum
adversaries. \ Then there is no quantum algorithm, running in $d^{o\left(
1\right)  }\cdot m^{O\left(  1\right)  }$\ time, that takes as input
$n$\ independent samples from a distribution $\mathcal{D}$\ over $\left[
d\right]  $, as well as $m=O\left(  n^{2}\right)  $\ efficiently computable
counting queries $E_{1},\ldots,E_{m}$ that are chosen adversarially and
adaptively, and correctly estimates $\operatorname{E}_{x\sim\mathcal{D}%
}\left[  E_{i}\left(  x\right)  \right]  $\ to within a fixed constant error
for each $i\in\left[  m\right]  $ in an online manner.
\end{theorem}

Theorem \ref{suthm} has the following as an immediate corollary.

\begin{corollary}
\label{sucor}Suppose there exists a symmetric-key encryption scheme that, for
keys of length $\kappa$, is semantically secure against $2^{\Omega(\kappa)}%
$-time quantum adversaries. \ Then there is no shadow tomography procedure
that is online and runs in $d^{o\left(  1\right)  }\cdot\mathit{poly}(m)$\ time.
\end{corollary}

Note that Corollary \ref{sucor} applies even to online algorithms that are not
gentle, and that work only in the \textquotedblleft identical-state
setting\textquotedblright\ (i.e., when the algorithm's input has the form
$\rho^{\otimes n}$). \ Moreover, it applies even for the classical special
case of shadow tomography. \ Finally, we note that just like Corollary
\ref{ullmancor}, Corollary \ref{sucor} extends to milder cryptographic
assumptions, albeit with milder conclusions for the complexity of gentle
shadow tomography.

\subsection{Quantum Complexity Implication\label{PROMISEBQP}}

We now observe that gentle measurements, whether or not derived from DP
algorithms, have potentially useful applications in quantum algorithms and
complexity. \ In particular, whenever we have an efficient implementation of a
gentle measurement, we can turn it into a safe and efficient way to run an
associated class of estimation subroutines on superpositions of inputs,
without generating unwanted garbage.

As an example, let's now prove Theorem \ref{promisebqpthm} from Section
\ref{APP}. \ In other words, let's show that without loss of generality, a
$\mathsf{BQP}$ machine can coherently query an oracle that takes as input a
description of a quantum circuit $C$, and that outputs an estimate of
$\Pr\left[  C\text{ accepts}\right]  $\ to within $\pm\varepsilon$, or a
superposition over such estimates, for any desired additive error
$\varepsilon=\frac{1}{n^{O\left(  1\right)  }}$. \ (In the sense that, for
every $\mathsf{BQP}$ machine that queries such an oracle, there is another
$\mathsf{BQP}$ machine that simulates the oracle on its own.) \ While this
might seem obvious, we would not know how to prove it without a gentle
measurement procedure of some kind.

\begin{proof}
[Proof of Theorem \ref{promisebqpthm}]Let%
\[
\sum_{g,C}\alpha_{g,C}\left\vert g,C\right\rangle
\]
be a state of the $\mathsf{BQP}$\ machine, where $g$ is garbage that we don't
care about and $C$\ is a description of a quantum circuit whose acceptance
probability (say, on the $\left\vert 0\cdots0\right\rangle $\ state) we'd like
to estimate. \ Then as a first step, we map the above state to%
\[
\sum_{g,C}\alpha_{g,C}\left\vert g,C\right\rangle \left(  C\left\vert
0\cdots0\right\rangle \right)  ^{\otimes\ell},
\]
for some suitable $\ell=n^{O\left(  1\right)  }$. \ Next we use the efficient
implementation of the Laplace noise measurement $L_{\sigma}$ (with $\sigma
\gg\sqrt{\ell}$), from Proposition \ref{lsigmafast}, to map the above to some
state%
\begin{equation}
\left\vert \psi\right\rangle \approx\sum_{g,C}\alpha_{g,C}\left\vert
g,C\right\rangle \left(  C\left\vert 0\cdots0\right\rangle \right)
^{\otimes\ell}\left\vert p_{C}\right\rangle . \label{psiapprox}%
\end{equation}
Here $p_{C}$\ is an estimate of $\Pr\left[  C\left\vert 0\cdots0\right\rangle
\text{ accepts}\right]  $ to within $\pm\eta$ additive error---or more
precisely, a Laplace superposition over estimates, one with the property that%
\[
\Pr\left[  \left\vert p_{C}-\Pr\left[  C\left\vert 0\cdots0\right\rangle
\text{ accepts}\right]  \right\vert >K\eta\right]  \leq\frac{1}{\exp\left(
\Omega\left(  K\right)  \right)  }%
\]
for all $K$. \ The equality (\ref{psiapprox})\ is only approximate because in
reality, the $\left\vert p_{C}\right\rangle $\ register is slightly entangled
with the $\left(  C\left\vert 0\cdots0\right\rangle \right)  ^{\otimes\ell}%
$\ registers. \ However, recall from Corollary \ref{lsigmagentle}\ that
$L_{\sigma}$\ is $\alpha$-gentle on product states for some $\alpha
=O(\sqrt{\ell}/\sigma)$. \ Thus, the damage to the $\left(  C\left\vert
0\cdots0\right\rangle \right)  ^{\otimes\ell}$\ registers in trace distance
can be upper-bounded by $\alpha$, and the equality (\ref{psiapprox}) also
holds up to error $\alpha$. \ So as a final step, we can simply uncompute the
$C\left\vert 0\cdots0\right\rangle $\ registers, to produce a state that is
$\alpha$-close in trace distance to%
\[
\sum_{g,C}\alpha_{g,C}\left\vert g,C\right\rangle \left\vert p_{C}%
\right\rangle .
\]
If we want to ensure that the above, in turn, is $\alpha$-close\ to a
superposition such that%
\[
\left\vert p_{C}-\Pr\left[  C\left\vert 0\cdots0\right\rangle \text{
accepts}\right]  \right\vert \leq\varepsilon
\]
with \textit{certainty}, where $\varepsilon$\ is our original accuracy bound,
then it suffices to choose $\eta$\ such that $K\eta\leq\varepsilon$ for some
$K=O\left(  \log\frac{1}{\alpha}\right)  $. \ Working backwards, a calculation
shows that it suffices to set%
\[
\ell=\Theta\left(  \frac{1}{\alpha^{2}\eta^{2}}\right)  =\Theta\left(
\frac{\log^{2}\frac{1}{\alpha}}{\alpha^{2}\varepsilon^{2}}\right)  .
\]
In turn, if our $\mathsf{BQP}$\ machine was going to make $T=n^{O\left(
1\right)  }$ such queries in sequence, it would suffice to set $\alpha
=\Theta\left(  \frac{1}{T}\right)  $ for each of them, to ensure that the
final output has trace distance at most (say) $\frac{1}{10}$\ from what we'd
obtain using an ideal oracle for approximating $\Pr\left[  C\left\vert
0\cdots0\right\rangle \text{ accepts}\right]  $.
\end{proof}

Though Theorem \ref{promisebqpthm} is not particularly shocking, it serves as
a model for a large number of results that could now be proven, using gentle
measurement procedures derived from DP algorithms. \ I.e., for every DP
algorithm that can be implemented coherently and in polynomial time, along the
lines of Proposition \ref{lsigmafast}, we get another way that quantum
algorithms can be safely invoked as subroutines by other quantum algorithms.

One might wonder about the difference between Theorem \ref{promisebqpthm}\ and
our results from Section \ref{COMPUTDP}. \ In particular, why was the Laplace
noise measurement $L_{\sigma}$\ needed for Theorem \ref{promisebqpthm}, but
\textit{not} needed for Theorem \ref{qsamplethm}? \ The key point is that, in
Theorem \ref{promisebqpthm}, we wanted outputs that were explicit estimates of
$\Pr\left[  C\text{ accepts}\right]  $. \ And even if two estimates $p\neq
p^{\prime}$\ are extremely close, the states $\left\vert p\right\rangle $\ and
$\left\vert p^{\prime}\right\rangle $\ will still be orthogonal. \ This is
what necessitated using a gentle measurement, to break the entanglement
between the output and computation registers, and thereby allow safe
uncomputing. \ In Section \ref{COMPUTDP}, by contrast, we were content with
outputs that were superpositions $\left\vert \phi_{X}\right\rangle $, with our
estimates of probabilities implicitly encoded in $\left\vert \phi
_{X}\right\rangle $'s amplitude vector. \ As a result, a slight error in
estimating those probabilities would yield a state $\left\vert \phi^{\prime
}\right\rangle $\ such that $\left\langle \phi^{\prime}|\phi_{X}\right\rangle
\approx1$, and gentle measurement techniques were not needed (even if the
results were \textit{useful} for efficient implementation of gentle measurements).

Here is an interesting question that we leave open. \ Suppose a quantum
algorithm has a polynomial-time quantum subroutine $C$, which on each input
$X$, generates a sample from a probability distribution $\mathcal{D}_{X}%
$\ supported on a sparse set $S_{X}\subset\left\{  0,1\right\}  ^{m}$\ with
$\left\vert S_{X}\right\vert =n^{O\left(  1\right)  }$. \ Suppose also that
the output we want, on each input $X$, is a polynomial-size approximate
\textit{description} of $\mathcal{D}_{X}$: that is, a string $z_{X}$ that
lists approximations to those $\Pr_{\mathcal{D}_{X}}\left[  y\right]
$\ values that are far from zero, or some other representation from which
$\mathcal{D}_{X}$\ could be efficiently sampled. \ Is there then, necessarily,
an efficient way to implement a mapping of the form%
\[
\sum_{X}\alpha_{X}\left\vert X\right\rangle \rightarrow\sum_{X}\alpha
_{X}\left\vert X\right\rangle \left\vert z_{X}\right\rangle ,
\]
with no garbage?

In the special case where $C$ is a classical randomized algorithm, we can do
this by first picking a single polynomial-size random string $r$, and then
using $r$ as $C$'s randomness for \textit{every} input $X$\ in the
superposition, relying on amplification and the union bound to ensure that
$C$\ succeeds on every $X$ with overwhelming probability over the choice of
$r$. \ This is an instance of the well-known \textquotedblleft Adleman's
trick\textquotedblright\ \cite{adleman:bpp} from complexity theory, as used
for example to prove the containment $\mathsf{BPP}\subset\mathsf{P/poly}$.
\ The use of a single $r$ avoids any unwanted entanglement between $r$\ and
the $\left\vert X\right\rangle $\ and $\left\vert z_{X}\right\rangle $ registers.

But what about the general case, where $C$ is a quantum algorithm? \ Here
Adleman's trick clearly won't work, so a different idea is needed: perhaps the
use of a more sophisticated DP algorithm than the Laplace algorithm used to
prove Theorem \ref{promisebqpthm}.

\section{Open Problems\label{OPEN}}

This paper established a new bridge between the fields of differential privacy
and quantum measurement. \ But we've barely begun to explore what this bridge
can carry. \ Here are a few of our favorite open problems.\bigskip

\textbf{Basic Questions}

\begin{enumerate}
\item[(1)] Can we generalize our main result, to show that $\varepsilon$-DP on
product states implies $O\left(  \varepsilon\sqrt{n}\right)  $-gentleness on
product states for \textit{any} quantum measurement, rather than only for
product measurements? \ One natural first step would be to prove this for LOCC
measurements. \ Another would be to show that $\varepsilon$-triviality on
product states implies $O\left(  \varepsilon\right)  $-gentleness (or even
just $O\left(  \varepsilon\sqrt{n}\right)  $-gentleness) on product states.
\ Note that there are two questions here: first, given a measurement $M$
that's $\varepsilon$-DP on product states, can we implement $M$ (meaning,
produce the correct output probabilities on \textit{all} states, not just
product states), in a way that happens to be $O\left(  \varepsilon\sqrt
{n}\right)  $-gentle\ when restricted to product states? \ And second, can we
implement some \textit{other} measurement $M^{\prime}$ that has essentially
the same output probabilities as $M$ on product states,\footnote{If $M$ is
$\varepsilon$-trivial, then to get a nontrivial question here, we demand
relative error on product states that's \textit{less} than $\varepsilon$.} and
that's also $O\left(  \varepsilon\sqrt{n}\right)  $-gentle\ on product states,
but that could be arbitrarily different from $M$ on entangled states?

\item[(2)] In this paper, we used our DP/gentleness connection, together with
known results from DP, to design and analyze a new quantum measurement
procedure of independent interest (namely, QPMW). \ Can we also go in the
opposite direction, and use known results from quantum measurement theory to
say anything new about classical differential privacy?

\item[(3)] Does $\alpha$-gentleness\ imply $O\left(  \alpha\right)  $-DP not
merely for all $\alpha\ll\frac{1}{4}$, but for all $\alpha\ll\frac{1}{2}$?

\item[(4)] In quantum differential privacy, how much can we do in the
\textquotedblleft local model,\textquotedblright\ wherein $n$ users\ are each
individually responsible for ensuring the privacy of their respective states
$\rho_{i}$, by submitting an obscured state $\widetilde{\rho}_{i}$\ to the
database? \ Also, how does the local model relate to the model wherein we can
only perform measurements on the $n$ states separately, for example because of
experimental limitations?\bigskip
\end{enumerate}

\textbf{Shadow Tomography}

\begin{enumerate}
\item[(5)] What is the true sample complexity of shadow tomography? \ Recall
that this paper's upper bound had the form $\left(  \log m\right)  ^{2}\left(
\log d\right)  ^{2}/\varepsilon^{O\left(  1\right)  }$, where $m$ is the
number of measurements and $d$ is the Hilbert space dimension. \ By contrast,
the best known lower bound is $\Omega\left(  \frac{\min\left\{  d^{2},\log
m\right\}  }{\varepsilon^{2}}\right)  $ \cite{aar:shadow}. \ Is any dependence
on $d$ needed? \ Theorem \ref{sqrtlogdlb}\ showed that, if a shadow tomography
procedure is also \textit{gentle} on product states, then it needs
$\widetilde{\Omega}\left(  \sqrt{\log m}\left(  \log d\right)  ^{1/4}\right)
$ samples. \ Meanwhile, Theorem \ref{adatoST}\ showed that if the procedure is
online, then it needs $\widetilde{\Omega}\left(  \sqrt{\min\left\{  m,\log
d\right\}  }\right)  $\ samples. \ But what if we drop these additional
requirements, or relax to gentleness on states of the form $\rho^{\otimes n}$?
\ We stress that any lower bound will need to be \textquotedblleft inherently
quantum,\textquotedblright\ since classically, in the offline and non-gentle
setting, an $O\left(  \frac{\log m}{\varepsilon^{2}}\right)  $\ upper bound
holds independent of $d$ \cite{aar:shadow}.

\item[(6)] Is it possible to do shadow tomography using incoherent
measurements (i.e., measuring each copy of $\rho$\ separately)? \ If so, this
would bring shadow tomography much closer to experimental feasibility.\bigskip
\end{enumerate}

\textbf{Composition}

\begin{enumerate}
\item[(7)] What can we say about the composition of quantum DP algorithms (see
Appendix \ref{CANTCOMPOSE} for further discussion)? \ In the regime where DP
implies gentleness, but where the probabilities of outcomes are too small for
Lemma \ref{facepalmlemma} to apply, can we compose DP algorithms in a way that
preserves not only accuracy,\ but also a multiplicative privacy guarantee?
\ Also, outside the regime where DP implies gentleness, is there any way to
get around the counterexample of Appendix \ref{CANTCOMPOSE}, and compose
quantum DP algorithms in a way that preserves accuracy (to say nothing about
privacy)? \ For example, what about \textquotedblleft
non-black-box\textquotedblright\ composition methods?

\item[(8)] Does an \textquotedblleft advanced composition
theorem\textquotedblright\ (see \cite{drv}) hold for gentleness, or at least
for the particular gentle measurements that arise from our connection between
gentleness and DP? \ In other words, if we perform $\alpha$-gentle
measurements $k$ times in sequence, then can we say that with high probability
over the measurement outcomes, our states have been damaged by only
$O(\alpha\sqrt{k})$\ in trace distance, rather than $O\left(  \alpha k\right)
$? \ If so, we could likely improve the sample complexity of our QPMW shadow
tomography procedure, say from $\left(  \log m\right)  ^{2}\left(  \log
d\right)  ^{2}/\varepsilon^{O\left(  1\right)  }$\ to $\left(  \log m\right)
^{2}\left(  \log d\right)  /\varepsilon^{O\left(  1\right)  }$.\bigskip
\end{enumerate}

\textbf{Computational Complexity}

\begin{enumerate}
\item[(9)] Is there any example of a polynomial-time classical randomized
algorithm that is $\varepsilon$-DP\ for some $\varepsilon\ll\frac{1}{\sqrt{n}%
}$, but does \textit{not} give rise to a gentle measurement on product states
that can be implemented in polynomial time, because of the issue with the
computational complexity of QSampling discussed in Section \ref{COMPUT}? \ If
so, are there any \textquotedblleft natural\textquotedblright\ examples of
such DP algorithms? \ It would be of interest to give such examples either
conditionally (say, based on a cryptographic assumption), or unconditionally
in the black-box model.

\item[(10)] Can we show, under some plausible cryptographic assumption, that
$d^{\Omega\left(  1\right)  }$\ computation time is needed for shadow
tomography, without the additional constraints that the procedure be online or gentle?

\item[(11)] Can we generalize Theorem \ref{promisebqpthm}, to give more
examples of how quantum algorithms can be safely invoked as subroutines by
other quantum algorithms using gentle measurement procedures? \ What about the
problem mentioned at the end of Section \ref{PROMISEBQP}?
\end{enumerate}

\section{Acknowledgments}

We thank Lijie Chen for insightful comments, including catching an error in a
previous analysis of QPMW;\ Thomas Steinke, Uri Stemmer, and Jon Ullman for
helpful conversations about lower bounds and hardness results for differential
privacy and adaptive data analysis; Andris Ambainis, Mark Bun, Dana
Moshkovitz, and Fabio Sciarrino for helpful conversations; and David Mestel
and the anonymous reviewers for their comments.

\bibliographystyle{plain}
\bibliography{thesis}

\section{Appendix: DP, Gentleness, and Triviality on Separable versus
Entangled States\label{SEPVSENT}}

What is the relationship between a measurement's being differentially private
(or trivial, or gentle) on product states, and its having those same
properties on arbitrary states? \

In this appendix, we'll give examples of measurements $M$ on $n$ qubits that are

\begin{enumerate}
\item[(1)] $\frac{1}{2^{\left(  n-1\right)  /2}}$-trivial, $\frac
{1}{2^{\left(  n-1\right)  /2}}$-DP, and $\frac{1}{2^{\left(  n-1\right)  /2}%
}$-gentle on all product states (and indeed, on all separable mixed states),
and yet

\item[(2)] extremely far from being trivial, private, or gentle on certain
entangled states.
\end{enumerate}

In some sense, this will answer our question \textquotedblleft for
complexity-theoretic purposes\textquotedblright: doing nothing whatsoever on
separable states, to some fixed exponential precision, is compatible with
enormous departures from DP, gentleness, and triviality on entangled states.

Nevertheless, we'll then show that there's \textit{some} level of triviality,
DP, and gentleness on product states that implies the same properties on
arbitrary states---but strikingly, that this would be false in quantum
mechanics over $\mathbb{R}$ rather than over $\mathbb{C}$.

\subsection{Separations\label{SEP2}}

Our first example separates DP on product states from DP on arbitrary states.

\begin{proposition}
\label{dpprodnotdp}There exists an $n$-qubit measurement $M$ that's $O\left(
2^{-n/2}\right)  $-trivial (and hence, $O\left(  2^{-n/2}\right)  $-DP) on
product states, but not $\varepsilon$-DP for any $\varepsilon$\ on arbitrary states.
\end{proposition}

\begin{proof}
For simplicity, let $n$ be odd,\ and group the first $n-1$ qubits into
$\frac{n-1}{2}$ pairs. \ Then the measurement $M$ will first project each of
these pairs onto the Bell pair $\frac{\left\vert 00\right\rangle +\left\vert
11\right\rangle }{\sqrt{2}}$. \ If all $\frac{n-1}{2}$ projections succeed,
then $M$ measures the $n^{th}$ qubit in the $\left\{  \left\vert
0\right\rangle ,\left\vert 1\right\rangle \right\}  $ basis and returns the
result. \ Otherwise $M$ returns a uniformly random bit.

Clearly, on states of the form%
\[
\frac{\left\vert 00\right\rangle +\left\vert 11\right\rangle }{\sqrt{2}%
}\otimes\cdots\otimes\frac{\left\vert 00\right\rangle +\left\vert
11\right\rangle }{\sqrt{2}}\otimes\left\vert v\right\rangle ,
\]
this measurement is not $\varepsilon$-DP for any $\varepsilon$, since (for
example) it completely leaks whether $\left\vert v\right\rangle =\left\vert
0\right\rangle $\ or $\left\vert v\right\rangle =\left\vert 1\right\rangle $.

On the other hand, we claim that $M$ is $O\left(  2^{-n/2}\right)  $-DP on
product states. \ To see this, observe that every $2$-qubit product state\ has
at most $\frac{1}{\sqrt{2}}$\ projection onto the Bell pair $\frac{\left\vert
00\right\rangle +\left\vert 11\right\rangle }{\sqrt{2}}$. \ So when we apply
$M$ to an $n$-qubit product state, the $\frac{n-1}{2}$ projections all succeed
with probability at most $2^{-\left(  n-1\right)  /2}$---and if at least one
projection fails, then $M$'s output is random. \ Thus, if $\rho$\ and $\sigma
$\ are any two product states, then for all $y\in\left\{  0,1\right\}  $,%
\[
\frac{\Pr\left[  M\left(  \rho\right)  \text{ outputs }y\right]  }{\Pr\left[
M\left(  \sigma\right)  \text{ outputs }y\right]  }\leq\frac{\frac{1}%
{2}+2^{-\left(  n-1\right)  /2}}{\frac{1}{2}-2^{-\left(  n-1\right)  /2}%
}=1+O\left(  2^{-n/2}\right)  .
\]

\end{proof}

As a bonus, we can adapt Proposition \ref{dpprodnotdp} to separate DP on
product states from DP on arbitrary states, even in the special case where the
measurement $M$\ is mixture-of-products.

\begin{proposition}
\label{loccbonus}There exists an $n$-qubit mixture-of-products measurement $M$
that's $\frac{1}{\exp\left(  n\right)  }$-trivial (or equivalently, $\frac
{1}{\exp\left(  n\right)  }$-DP) on product states, but is not $\varepsilon
$-DP for any $\varepsilon<\exp\left(  n\right)  $\ on arbitrary states.
\end{proposition}

\begin{proof}
We simply modify the measurement $M$ from the proof of Proposition
\ref{dpprodnotdp}, so that now $M$ tries to use each of the $\frac{n-1}{2}$
qubit pairs to violate a Bell inequality---say, by playing the so-called
\textit{CHSH game} \cite{chsh}, which can be won with probability $\cos
^{2}\frac{\pi}{8}\approx0.85$\ using\ the entangled state $\frac{\left\vert
00\right\rangle +\left\vert 11\right\rangle }{\sqrt{2}}$, but with at most
$\frac{3}{4}$\ probability using any unentangled state.

If $M$ wins at the CHSH game, on (say) at least an $0.8$ fraction of the
$\frac{n-1}{2}$\ qubit pairs, then $M$ returns the result of measuring the
$n^{th}$ qubit in the $\left\{  \left\vert 0\right\rangle ,\left\vert
1\right\rangle \right\}  $\ basis. \ Otherwise, $M$ returns a uniformly random bit.

Again, on states of the form%
\[
\frac{\left\vert 00\right\rangle +\left\vert 11\right\rangle }{\sqrt{2}%
}\otimes\cdots\otimes\frac{\left\vert 00\right\rangle +\left\vert
11\right\rangle }{\sqrt{2}}\otimes\left\vert v\right\rangle ,
\]
this measurement is not $\varepsilon$-DP for any $\varepsilon<\exp\left(
n\right)  $, since it leaks whether $\left\vert v\right\rangle =\left\vert
0\right\rangle $\ or $\left\vert v\right\rangle =\left\vert 1\right\rangle $
with all but exponentially small probability.

But again, on product states, we claim that $M$ is $\frac{1}{\exp\left(
n\right)  }$-trivial. \ For by a Chernoff bound, whenever $M$ is applied to a
product state, the $n^{th}$\ qubit is measured with at most $\frac{1}%
{\exp\left(  n\right)  }$\ probability.
\end{proof}

The measurements $M$\ from Propositions \ref{dpprodnotdp}\ and \ref{loccbonus}
don't have product form, so we can't apply Theorem \ref{main}\ to them to
conclude automatically that they're $\frac{1}{\exp\left(  n\right)  }$-gentle
on product states. \ Nevertheless, it's not hard to verify directly that they
\textit{are} $\frac{1}{\exp\left(  n\right)  }$-gentle on product states, and
even on separable mixed states.

By contrast, Corollary \ref{gentledpall} says that, if $M$ is $\alpha$-gentle
on all states, then $M$\ is $\ln\left(  \frac{1+4\alpha}{1-4\alpha}\right)
$-DP\ on all states. \ But $M$ is \textit{not} $\varepsilon$-DP on all states,
for any $\varepsilon>0$\ (in the case of Proposition \ref{dpprodnotdp}) or for
any $\varepsilon<\exp\left(  n\right)  $\ (in the case of Proposition
\ref{loccbonus}). \ So summarizing, we obtain the following corollary of
Propositions \ref{dpprodnotdp}\ and \ref{loccbonus}, which dramatically
separates \textit{gentleness} on product states from \textit{gentleness} on
all states:

\begin{corollary}
\label{gentlesep}There exists an $n$-qubit measurement $M$ that's $\frac
{1}{\exp\left(  n\right)  }$-gentle (and indeed, $\frac{1}{\exp\left(
n\right)  }$-trivial) on product states and indeed on separable mixed states,
but not $\alpha$-gentle for any $\alpha<\frac{1}{4.01}$\ on arbitrary states.
\ We can even take this measurement to be mixture-of-products.
\end{corollary}

From Proposition \ref{notgentle}, together with Lemma \ref{dpgentleprojective}%
, we already get that the measurement $L_{n/2}$\ is $O\left(  1/\sqrt
{n}\right)  $-gentle on product states\ despite not being $\frac{1}{3}$-gentle
on arbitrary states. \ However, Corollary \ref{gentlesep}\ gives an
exponentially more dramatic separation between gentleness on product states
and gentleness on arbitrary states.

It will follow from Corollary \ref{prodtoalleps}, proved in Section
\ref{COMPLEX}, that these exponential separations, between triviality, DP, and
gentleness on product states and the same parameters on arbitrary states,\ are
the largest separations possible, up to the exact value of the exponential
scaling factor.

Note also that the following is an immediate consequence of convexity and of
Proposition \ref{convextd}:

\begin{proposition}
\label{sepmixed}If $M$ is $\varepsilon$-trivial or $\varepsilon$-DP on all
product states, then $M$ is also $\varepsilon$-trivial or $\varepsilon$-DP
respectively on all separable mixed states.
\end{proposition}

Beware that $\alpha$-gentleness on product states does \textit{not}
automatically imply $\alpha$-gentleness on separable mixed states (even though
in the examples above the two happened to go together); the measurement
$L_{\sigma}$\ is a counterexample.

As a final remark, one might wonder whether the counterexamples of
Propositions \ref{dpprodnotdp}\ and \ref{loccbonus} and Corollary
\ref{gentlesep} have classical probabilistic analogues. \ In other words, is
there a separation between DP on product distributions, and DP on arbitrary
distributions? \ Or the analogous question for triviality? \ We now observe
that the answer is no. \ Indeed, this is just a special case of Proposition
\ref{sepmixed} above. \ Every probability distribution can be written as a
convex combination of product distributions (indeed, point distributions), and
DP and triviality are both closed under convex combinations.\footnote{Again,
gentleness is the outlier, failing to be closed under convex combinations.
\ It's not hard to show, by a classical analogue of Lemma \ref{mixtrick}, that
the only classical algorithms that are gentle on arbitrary distributions
$\mathcal{D}$ are close to trivial. \ But \emph{every} algorithm is, or can be
made, gentle on classical computational basis states.}

Why is the quantum case different? \ Because, while DP is closed under convex
combinations, it's \textit{not} closed under superpositions. \ The CHSH game
provides one example of this: a certain measurement\ has a behavior on the
Bell pair $\frac{\left\vert 00\right\rangle +\left\vert 11\right\rangle
}{\sqrt{2}}$\ that's \textit{not} a convex combination of its behaviors on the
components $\left\vert 00\right\rangle $\ and $\left\vert 11\right\rangle
$---so that the measurement can fail to be DP on the superposition, despite
being DP on the components. \ Thus, the separation between DP on product
states and DP on arbitrary states is a quantum phenomenon.

\subsection{Relationships\label{COMPLEX}}

We'll now show that, despite the separating examples in the last section, a
measurement's being $\varepsilon$-trivial\ on product states for
\textit{extremely} small values of $\varepsilon$\ (say, $\varepsilon\ll
\frac{1}{\left(  2\sqrt{2}\right)  ^{n}}$), really does imply its being nearly
trivial on arbitrary states (and hence DP and gentle as well). \ Intriguingly,
we'll also show that this\ depends on the fact that amplitudes in quantum
mechanics can be complex rather than only real.

Our first claim is that \textit{any measurement }$M$\textit{\ that accepts
every product state with the same probability }$p$\textit{, in fact accepts
every state with probability }$p$\textit{.} \ We do not know whether this was
known before; in any case, we cannot resist including a strikingly simple
proof for completeness. \ Our proof uses the following result of Braunstein et
al. \cite{bcjlps}:

\begin{theorem}
[Braunstein et al.\ \cite{bcjlps}]\label{densitythm}In any finite-dimensional
tensor product Hilbert space (on any number of registers), the separable mixed
states have positive density within the set of all mixed states.
\end{theorem}

We observe the following consequence.

\begin{theorem}
\label{prodtoall}Suppose a measurement $M$ is $0$-trivial (or equivalently,
$0$-DP or $0$-gentle) on all product states. \ Then $M$ is $0$-trivial on all states.
\end{theorem}

\begin{proof}
If $M$\ is $0$-trivial on product states, then for each possible outcome $y$,
there is some constant $p$\ such that, for all product states $\rho=\rho
_{1}\otimes\cdots\otimes\rho_{n}$,%
\[
\Pr\left[  M\left(  \rho\right)  \text{ outputs }y\right]  =p.
\]
So by convexity, the above holds as well for all convex combinations of
product states: i.e., separable mixed states. \ Now%
\[
\Pr\left[  M\left(  \rho\right)  \text{ outputs }y\right]  =\operatorname{Tr}%
\left(  E\rho\right)
\]
for some Hermitian operator $E$. \ By Theorem \ref{densitythm}, this means
that the linear function $f\left(  \rho\right)  :=\operatorname{Tr}\left(
E\rho\right)  $ equals $p$ on a subset of positive density. \ But any linear
function that's constant on a subset of positive density is constant
everywhere, so $\operatorname{Tr}\left(  E\rho\right)  =p$\ for all $\rho$.
\end{proof}

Why did this depend on amplitudes being complex numbers? \ In quantum
mechanics over $\mathbb{R}$, the result of Braunstein et al.\ \cite{bcjlps} is
known to be false. \ Let us now show that Theorem \ref{prodtoall} is false as
well. \ Consider the $2$-outcome measurement on $2$ \textquotedblleft
rebits\textquotedblright\ (i.e., real-amplitude qubits) that accepts $\rho
$\ with probability $\operatorname{Tr}\left(  E\rho\right)  $, where%
\[
E=\frac{1}{2}\left(
\begin{array}
[c]{cccc}%
1 & 0 & 0 & -1\\
0 & 1 & 1 & 0\\
0 & 1 & 1 & 0\\
-1 & 0 & 0 & 1
\end{array}
\right)  .
\]
One can check that, for every $2$-rebit pure product state $\left\vert
\psi\right\rangle =\left\vert \psi_{A}\right\rangle \otimes\left\vert \psi
_{B}\right\rangle $, we have%
\[
\operatorname{Tr}\left(  E|\psi\rangle\langle\psi|\right)  =\left\langle
\psi|E|\psi\right\rangle =\frac{1}{2},
\]
and hence the same is true for every $2$-rebit separable mixed state.
\ Nevertheless, this measurement accepts the entangled rebit state
$\frac{\left\vert 01\right\rangle +\left\vert 10\right\rangle }{\sqrt{2}}%
$\ with certainty, and rejects $\frac{\left\vert 00\right\rangle +\left\vert
11\right\rangle }{\sqrt{2}}$\ with certainty. \ This is a rare example of a
quantum information phenomenon that's fundamentally different for qubits and
rebits.\footnote{In the same spirit: in complex quantum mechanics, one can
recover the POVM $E$ if one knows $\operatorname{Tr}\left(  E\rho\right)
$\ for all product states $\rho$; but in real quantum mechanics, one
can't---by the same counterexample $E$, which the product states $\rho$\ of
$2$ rebits fail to distinguish from the $\mathbb{I}/2$ POVM that accepts every
state with probability $1/2$. \ This fact is a \textquotedblleft
dual\textquotedblright\ to the well-known fact that a mixed state $\rho$\ is
uniquely determined by the values of $\operatorname{Tr}\left(  E\rho\right)
$\ on all product measurements $E$ (i.e., Hardy's \textquotedblleft local
tomography axiom\textquotedblright\ \cite{hardy}\ holds), in complex quantum
mechanics but not in real quantum mechanics. \ The \textquotedblleft
duality\textquotedblright\ between the two facts can be seen by interchanging
the roles of the Hermitian matrices $E$\ and $\rho$\ in the expression
$\operatorname{Tr}\left(  E\rho\right)  $.}

In ordinary (complex) quantum mechanics, we can even obtain a weak
\textit{quantitative} connection between DP, gentleness, and triviality on
product states and the same notions on arbitrary states, by using the
following result due to Gurvits and Barnum \cite{gurvitsbarnum}.

\begin{theorem}
[\cite{gurvitsbarnum}]\label{gurvitsthm}Let $\rho$\ be any mixed state on $n$
registers, each $d$-dimensional. \ Then the state $\left(  1-\delta\right)
\frac{\mathbb{I}}{d^{n}}+\delta\rho$ is separable, for all $\delta\leq
\frac{2^{-n/2}}{d^{n}}$.
\end{theorem}

Theorem \ref{gurvitsthm}\ has the following corollary.

\begin{corollary}
\label{prodtoalleps}Suppose the measurement $M$, on $n$ registers of $d$
dimensions each, is $\varepsilon$-trivial on product states, for some
$\varepsilon\leq\frac{1}{2(\sqrt{2}d)^{n}}$. \ Then $M$ is $O\left(  (\sqrt
{2}d)^{n}\varepsilon\right)  $-trivial on all states.
\end{corollary}

\begin{proof}
Fix some measurement outcome $y$ corresponding to the POVM element $E$. \ Then
let $p=\operatorname{Tr}\left(  E\frac{\mathbb{I}}{d^{n}}\right)  $ be the
probability that $M$ outputs $y$ on the maximally mixed state. \ Set
$\delta:=\frac{1}{\left(  \sqrt{2}d\right)  ^{n}}$, so that $\varepsilon
\leq\frac{\delta}{2}$. \ Let $\rho$ be an arbitrary state, and let%
\[
\sigma:=\left(  1-\delta\right)  \frac{\mathbb{I}}{d^{n}}+\delta\rho.
\]
Then $\sigma$\ is separable by Theorem \ref{gurvitsthm}. \ So since $M$ is
$\varepsilon$-trivial on product states,%
\[
pe^{-\varepsilon}\leq\operatorname{Tr}\left(  E\sigma\right)  \leq
pe^{\varepsilon}.
\]
Now,%
\begin{align*}
\operatorname{Tr}\left(  E\sigma\right)   &  =\operatorname{Tr}\left(
E\left(  \left(  1-\delta\right)  \frac{\mathbb{I}}{d^{n}}+\delta\rho\right)
\right) \\
&  =\left(  1-\delta\right)  p+\delta\operatorname{Tr}\left(  E\rho\right)  .
\end{align*}
Solving for $\operatorname{Tr}\left(  E\rho\right)  $, we find that%
\[
\left(  1-\frac{1-e^{-\varepsilon}}{\delta}\right)  p\leq\operatorname{Tr}%
\left(  E\rho\right)  \leq\left(  1+\frac{e^{\varepsilon}-1}{\delta}\right)
p
\]
This implies that $M$ is $\beta$-trivial\ on all states, for%
\[
\beta=\ln\left(  \frac{1+\frac{e^{\varepsilon}-1}{\delta}}{1-\frac
{1-e^{-\varepsilon}}{\delta}}\right)  =O\left(  \frac{\varepsilon}{\delta
}\right)  =O\left(  (\sqrt{2}d)^{n}\varepsilon\right)  .
\]
Here we've used the fact that $\varepsilon\leq\frac{\delta}{2}$.
\end{proof}

\section{Appendix: General Neighbor Relations\label{GENNEIGHBOR}}

Given two states $\rho,\sigma$\ on $n$ registers each, we called $\rho$\ and
$\sigma$\ \textit{neighbors} if it's possible to reach $\sigma$\ from $\rho$,
or $\sigma$\ from $\rho$, by applying some superoperator to a single register
only. \ In the special case where $\rho=\rho_{1}\otimes\cdots\otimes\rho_{n}%
$\ and $\sigma=\sigma_{1}\otimes\cdots\otimes\sigma_{n}$\ are both product
states, this is simply equivalent to saying that we can reach $\sigma$\ from
$\rho$ by changing\ a single $\rho_{i}$. \ For correlated or entangled states,
by contrast, it's not obvious that we should favor this definition over
various alternatives.

Thus, call $\rho$\ and $\sigma$\ \textit{superoperator neighbors} if they're
neighbors in the sense above. \ Call them \textit{unitary neighbors} if it's
possible to reach $\sigma$\ from $\rho$, or equivalently $\rho$\ from $\sigma
$, by applying some unitary transformation $U$ to a single register only.
\ And call them \textit{conditioned neighbors} if it's possible to reach one
from the other by applying a conditioned superoperator (i.e., a normalized
quantum operation) to a single register. \ Clearly, all unitary neighbors are
also superoperator neighbors, and all superoperator neighbors are also
conditioned neighbors. \ But for general states, the three notions are easily
seen to form a strict hierarchy. \ For example, $\frac{1}{\sqrt{2}}\left(
\left\vert 0^{n}\right\rangle +\left\vert 1^{n}\right\rangle \right)  $\ and
$\frac{1}{2}\left(  |0^{n}\rangle\langle0^{n}|+|1^{n}\rangle\langle
1^{n}|\right)  $\ are superoperator neighbors but not unitary neighbors, while
$\frac{1}{\sqrt{2}}\left(  \left\vert 0^{n}\right\rangle +\left\vert
1^{n}\right\rangle \right)  $\ and $\left\vert 0^{n}\right\rangle $\ are
conditioned neighbors but not superoperator neighbors.

Nevertheless, we now prove that, for the task of defining $\varepsilon$-DP,
switching from superoperator neighbors to unitary neighbors would change
nothing of substance, while switching to conditioned neighbors would collapse
our framework to triviality.

\begin{proposition}
If $M$ is $\varepsilon$-DP with respect to unitary neighbors,\ then $M$ is
also $2\varepsilon$-DP with respect to superoperator neighbors (regardless of
whether we mean DP on product states or on all states).
\end{proposition}

\begin{proof}
Let $\rho$\ and $\sigma$\ be superoperator neighbors, which differ only on the
$i^{th}$\ register. \ Let $\xi$ be the state obtained by starting from either
$\rho$\ \textit{or} $\sigma$, and then applying a Haar-random unitary
transformation $U$ to the $i^{th}$\ register (which has the effect of putting
that register into the maximally mixed state, $\mathbb{I}/d$). \ Then
averaging over the possible $U$'s and applying convexity, we have%
\[
\Pr\left[  M\left(  \rho\right)  =y\right]  \leq e^{\varepsilon}\Pr\left[
M\left(  \xi\right)  =y\right]
\]
and likewise%
\[
\Pr\left[  M\left(  \sigma\right)  =y\right]  \geq e^{-\varepsilon}\Pr\left[
M\left(  \xi\right)  =y\right]  .
\]
Hence%
\[
\Pr\left[  M\left(  \rho\right)  =y\right]  \leq e^{2\varepsilon}\Pr\left[
M\left(  \sigma\right)  =y\right]  .
\]

\end{proof}

\begin{proposition}
If $M$ is $\varepsilon$-DP on all states with respect to postselected
neighbors,\ then $M$ is $4\varepsilon$-trivial.
\end{proposition}

\begin{proof}
Let $\rho,\sigma$\ be any two mixed states on $n$ registers each. \ Let
$|1\rangle\langle1|\otimes\rho^{\prime}$\ be the result of measuring the first
register of $\rho$ in the $\left\vert 1\right\rangle ,\ldots,\left\vert
d\right\rangle $\ basis and getting the outcome $\left\vert 1\right\rangle $,
and let $|2\rangle\langle2|\otimes\sigma^{\prime}$\ be the result of measuring
the first register of $\sigma$ in the $\left\vert 1\right\rangle
,\ldots,\left\vert d\right\rangle $\ basis and getting the outcome $\left\vert
2\right\rangle $. \ Then by assumption, for all possible outcomes $y$ of $M$,%
\begin{align*}
\Pr\left[  M\left(  \rho\right)  =y\right]   &  \leq e^{\varepsilon}\Pr\left[
M\left(  |1\rangle\langle1|\otimes\rho^{\prime}\right)  =y\right]  ,\\
\Pr\left[  M\left(  \sigma\right)  =y\right]   &  \geq e^{-\varepsilon}%
\Pr\left[  M\left(  |2\rangle\langle2|\otimes\sigma^{\prime}\right)
=y\right]  .
\end{align*}
But the state%
\[
\xi:=\frac{|1\rangle\langle1|\otimes\rho^{\prime}+|2\rangle\langle
2|\otimes\sigma^{\prime}}{2}%
\]
is a postselected neighbor of both $|1\rangle\langle1|\otimes\rho^{\prime}%
$\ and $|2\rangle\langle2|\otimes\sigma^{\prime}$,\ since measuring the first
register of $\xi$\ in the $\left\vert 1\right\rangle ,\ldots,\left\vert
d\right\rangle $\ basis can yield either. \ Hence%
\begin{align*}
\Pr\left[  M\left(  |1\rangle\langle1|\otimes\rho^{\prime}\right)  =y\right]
&  \leq e^{\varepsilon}\Pr\left[  M\left(  \xi\right)  =y\right]  ,\\
\Pr\left[  M\left(  |2\rangle\langle2|\otimes\sigma^{\prime}\right)
=y\right]   &  \geq e^{-\varepsilon}\Pr\left[  M\left(  \xi\right)  =y\right]
.
\end{align*}
Chaining together the inequalities now yields%
\[
\Pr\left[  M\left(  \rho\right)  =y\right]  \leq e^{4\varepsilon}\Pr\left[
M\left(  \sigma\right)  =y\right]  .
\]

\end{proof}

\section{Appendix: Differential Privacy Beyond Product and LOCC
Measurements\label{NOTPRODMEAS}}

In this appendix, we'll give an example of a measurement $M$ on $n$ qubits,
which is differentially private on all states, but which is provably
\textit{not} a product measurement, or even a mixture-of-products measurement.
\ In other words, there's no way to implement $M$ (even approximately) by
measuring each qubit in a separately chosen basis, with none of the bases
depending on the outcomes of measuring previous qubits. \ This rules out the
possibility of a \textquotedblleft structure theorem\textquotedblright%
\ showing that all DP measurements can be put into the restricted form that we
mainly studied in the body of this paper.

Going further, we'll also give a second measurement $M^{\prime}$ that's
differentially private on all $n$-qubit states, but which we conjecture is not
even LOCC. \ That is, we conjecture that there's no way to implement
$M^{\prime}$ using local operations and classical communication (even allowing
adaptivity), and that entangling measurements on the qubits are needed.

To construct $M$, we'll use the following lemma.

\begin{lemma}
\label{noprod}There is no $2$-qubit mixture-of-products measurement that
accepts the states $\left\vert 0\right\rangle \left\vert 0\right\rangle $ and
$\left\vert 1\right\rangle \left\vert +\right\rangle $ with certainty, and
that rejects $\left\vert 0\right\rangle \left\vert 1\right\rangle $ and
$\left\vert 1\right\rangle \left\vert -\right\rangle $\ with certainty.
\end{lemma}

\begin{proof}
It suffices to show that there's no product measurement; the lemma then
follows by convexity.

A product measurement can be written $\left\{  E_{i}\otimes F_{j}\right\}
_{i\in\left[  k\right]  ,j\in\left[  \ell\right]  }$, for some one-qubit POVMs
$E_{1}+\cdots+E_{k}=\mathbb{I}$\ and $F_{1}+\cdots+F_{\ell}=\mathbb{I}$.

Suppose we knew that measuring the first qubit in the $\left\{  \left\vert
0\right\rangle ,\left\vert 1\right\rangle \right\}  $\ basis yielded the
outcome $\left\vert 0\right\rangle $. \ Then we'd need to accept with
certainty if the second qubit was $\left\vert 0\right\rangle $, and reject
with certainty if the second qubit was $\left\vert 1\right\rangle $. \ But
since the $F_{i}$'s must be Hermitian and positive semidefinite, the only
POVMs $\left\{  F_{1},\ldots,F_{\ell}\right\}  $\ on the second qubit that
achieve that objective are equivalent under trivial changes (i.e., relabelings
and adding \textquotedblleft dummy\textquotedblright\ POVM elements) to%
\[
F_{1}=\left(
\begin{array}
[c]{cc}%
1 & 0\\
0 & 0
\end{array}
\right)  ,~~~~~F_{2}=\left(
\begin{array}
[c]{cc}%
0 & 0\\
0 & 1
\end{array}
\right)
\]
---in other words, simply measuring the second qubit in the $\left\{
\left\vert 0\right\rangle ,\left\vert 1\right\rangle \right\}  $\ basis.
\ Likewise, if we knew that measuring the first qubit in the $\left\{
\left\vert 0\right\rangle ,\left\vert 1\right\rangle \right\}  $\ basis
yielded the outcome $\left\vert 1\right\rangle $, then we'd need to accept
with certainty if the second qubit was $\left\vert +\right\rangle $, and
reject with certainty if the second qubit was $\left\vert -\right\rangle $.
\ The only POVMs that achieve \textit{that} objective are equivalent under
trivial changes to%
\[
F_{1}^{\prime}=\frac{1}{2}\left(
\begin{array}
[c]{cc}%
1 & 1\\
1 & 1
\end{array}
\right)  ,~~~~~F_{2}^{\prime}=\frac{1}{2}\left(
\begin{array}
[c]{cc}%
1 & -1\\
-1 & 1
\end{array}
\right)
\]
---in other words, measuring the second qubit in the $\left\{  \left\vert
+\right\rangle ,\left\vert -\right\rangle \right\}  $\ basis. \ But since we
don't get to choose $\left\{  F_{1},\ldots,F_{\ell}\right\}  $\ based on the
outcome of measuring the first qubit, we can't achieve both objectives simultaneously.

Finally, if the first qubit was \textit{not} measured in the $\left\{
\left\vert 0\right\rangle ,\left\vert 1\right\rangle \right\}  $\ basis, but
in some other basis, then the situation is \textquotedblleft even
worse,\textquotedblright\ since some outcome $E_{i}$\ of measuring the first
qubit will be compatible with the first qubit having been $\left\vert
0\right\rangle $\ \textit{or} with its having been $\left\vert 1\right\rangle
$. \ So even fixing $E_{i}$, we'll again need POVM elements equivalent to
$F_{1},F_{2}$\ \textit{and} POVM elements equivalent to $F_{1}^{\prime}%
,F_{2}^{\prime}$, which contradicts $F_{1}+\cdots+F_{\ell}=\mathbb{I}$.
\end{proof}

By compactness considerations, a corollary of Lemma \ref{noprod} is that there
must be some constant $\eta>0$\ (we have not worked out its value) such that
no mixture-of-products measurement can distinguish $\left\vert 0\right\rangle
\left\vert 0\right\rangle $ and $\left\vert 1\right\rangle \left\vert
+\right\rangle $\ from $\left\vert 0\right\rangle \left\vert 1\right\rangle $
and $\left\vert 1\right\rangle \left\vert -\right\rangle $\ even with success
probability $1-\eta$.

Using Lemma \ref{noprod}, we now prove the main result.

\begin{theorem}
[Existence of Non-Product Quantum DP Measurements]\label{noprodthm}There
exists a measurement $M$ on $n$ qubits that's $O\left(  \frac{\log n}%
{n}\right)  $-DP on all states, but that cannot be approximated (say, to
$\frac{1}{3}$\ variation distance in the distribution over measurement
outcomes) by any mixture-of-products measurement.
\end{theorem}

\begin{proof}
Set $k:=C\log n$\ for some constant $C$. \ Then the measurement $M=M_{\sigma}$
does the following:

\begin{enumerate}
\item[(1)] Group the $n$ qubits into $n/k$ blocks $B_{1},\ldots,B_{n/k}$, each
of size $k$

\item[(2)] Within each block $B_{i}$:

\begin{itemize}
\item Group the qubits into pairs

\item Measure each pair in the basis $\left\{  \left\vert 0\right\rangle
\left\vert 0\right\rangle ,\left\vert 0\right\rangle \left\vert 1\right\rangle
,\left\vert 1\right\rangle \left\vert +\right\rangle ,\left\vert
1\right\rangle \left\vert -\right\rangle \right\}  $

\item Count the number of these measurements that return either $\left\vert
0\right\rangle \left\vert 0\right\rangle $ or $\left\vert 1\right\rangle
\left\vert +\right\rangle $, and calculate the parity\ of this number,
$b_{i}\in\left\{  0,1\right\}  $
\end{itemize}

\item[(3)] Return the sum $\Gamma=b_{1}+\cdots+b_{n/k}$, across all $n/k$
blocks, plus Laplace noise with average magnitude $\sigma$
\end{enumerate}

Our first claim is that $M$ is $1/\sigma$-DP on all states. \ This is a simple
consequence of Proposition \ref{lsigma}.

Our second claim is that there exists a probability distribution $\mathcal{D}%
$\ over $n$-qubit states (in fact, product states), such that given a state
$\rho$\ drawn from $\mathcal{D}$, no mixture-of-products measurement can
return a nontrivial estimate of $\Gamma$. \ This $\mathcal{D}$\ is defined as
follows: first set $\Gamma:=0$\ or $\Gamma:=n/k$, both with equal probability
$1/2$. \ Then let $\rho$\ be a tensor product of pairs of qubits of the form
$\left\vert 0\right\rangle \left\vert 0\right\rangle ,\left\vert
0\right\rangle \left\vert 1\right\rangle ,\left\vert 1\right\rangle \left\vert
+\right\rangle ,\left\vert 1\right\rangle \left\vert -\right\rangle $, which
is chosen uniformly at random among all such tensor products that are
consistent with the chosen value of $\Gamma$.

To prove the claim, it suffices to show that no mixture-of-products
measurement can guess even a \textit{single} parity $b_{i}$, by measuring the
$i^{th}$ block, with bias more than (say) $1/n^{3}$ over chance. \ For this we
appeal to Lemma \ref{noprod}, which says that for each pair of qubits, the
measurement \textit{cannot} perfectly distinguish whether the pair is in the
state $\frac{|00\rangle\langle00|+|1+\rangle\langle1+|}{2}$ (which flips the
parity $b_{i}$),\ or the state $\frac{|01\rangle\langle01|+|1-\rangle
\langle1-|}{2}$\ (which has no effect on $b_{i}$). \ Rather, it can only
distinguish these two mixed states only with some constant rate of noise
$\eta>0$. \ So the situation is equivalent to the following: we are trying to
guess the parity $\left\vert x\right\vert \operatorname{mod}2$\ of an
arbitrary $k/2$-bit string $x$, but each bit of $x$ can be read only noisily,
and has either an $\eta$\ probability of appearing as $1$ despite being $0$ or
vice versa (with the errors independent across the bits). \ In such a
situation, regardless of the value of $\eta>0$, it is well-known that our bias
in guessing the parity of $x$\ falls off like $1/\exp\left(  k\right)  $. \ By
simply setting $k:=C\log n$ for some sufficiently large constant $k$, we can
make this bias less than $1/n^{3}$.

Finally, we just need to choose (say) $\sigma:=\frac{n}{10k}$. \ In that case,
the measurement $M=M_{\sigma}$ is $\varepsilon$-DP\ for $\varepsilon=\frac
{1}{\sigma}=O\left(  \frac{\log n}{n}\right)  $, and it returns a nontrivial
estimate of $\Gamma$. \ By contrast, no mixture-of-products measurement
returns a nontrivial estimate of $\Gamma$, or even distinguishes the case
$\Gamma=0$\ from the case $\Gamma=n/k$\ with bias (say) $\Omega\left(
1/n^{2}\right)  $.
\end{proof}

The proof of Theorem \ref{noprodthm} exploited the fact that, even though
differential privacy is clearly associated with a lack of \textquotedblleft
sensitivity\textquotedblright\ on the measurement's part (i.e., changing a
single register can't change the output by much), this is still compatible
with \textit{local subproblems} solved by the measurement being exquisitely
sensitive to local changes. \ That's what happens with the noisy sum of
parities example: each parity is maximally sensitive to local changes, even
though a noisy sum of them is not.

Now suppose we want to show something stronger: namely, that there's an
$n$-qubit measurement $M$ that's differentially private, but that isn't even
LOCC (that is, cannot be implemented using separate measurements on each
qubit, even with adaptivity). \ We now propose a modification $M^{\prime}$ of
the measurement $M$ from the proof of Theorem \ref{noprodthm}, which we
conjecture has the required property.

Set $k:=C\log n$\ for some constant $C$. \ Then the measurement $M^{\prime
}=M_{\sigma}^{\prime}$ does the following.

\begin{enumerate}
\item[(1)] Group the $n$ qubits into $n/k^{3}$ blocks $B_{1},\ldots
,B_{n/k^{3}}$, each of size $k^{3}$

\item[(2)] Within each block $B_{i}$:

\begin{itemize}
\item Group the qubits into $k$ sub-blocks $S_{1},\ldots,S_{k}$, each of size
$k^{2}$

\item Within each sub-block $S_{j}$:

\begin{itemize}
\item Group the qubits into pairs

\item Perform the swap test on each pair (note: the swap test accepts the
product state $\left\vert v\right\rangle \left\vert w\right\rangle $ with
probability equal to $\frac{1+\left\vert \left\langle v|w\right\rangle
\right\vert ^{2}}{2}$)

\item Call $S_{j}$ \textquotedblleft accepting\textquotedblright\ if
\textit{every} swap test accepts, or \textquotedblleft
rejecting\textquotedblright\ if at least one of them rejects
\end{itemize}

\item Count the number of accepting sub-blocks, and let $b_{i}$\ be the parity
of this number
\end{itemize}

\item[(3)] Return the sum $\Gamma=b_{1}+\cdots+b_{n/k^{3}}$, across all
$n/k^{3}$ blocks, plus Laplace noise with average magnitude $\sigma$
\end{enumerate}

Just like in the proof of Theorem \ref{noprodthm}, it's easy to see that
$M^{\prime}$\ is $1/\sigma$-DP on all states. \ Our conjecture is that
$M^{\prime}$\ cannot be implemented, even approximately, using LOCC
measurements on the qubits. \ The intuition is that, if we're restricted to
LOCC, then at best we can simulate each swap test imperfectly: for example,
using a measurement that accepts the product state $\left\vert v\right\rangle
\left\vert w\right\rangle $ with probability equal to $\frac{1+\left\vert
\left\langle v|w\right\rangle \right\vert ^{2}}{3}$, rather than
$\frac{1+\left\vert \left\langle v|w\right\rangle \right\vert ^{2}}{2}$.
\ This would imply that we can't reliably distinguish the following two cases:

\begin{enumerate}
\item[(1)] within a given sub-block, \textit{every} swap test accepts with
probability $1$, versus

\item[(2)] within that sub-block, $10\log n$ swap tests accept with
probability $1/2$ (i.e., the two qubits are in orthogonal states), while the
remaining swap tests accept with probability $1$.
\end{enumerate}

For the difference between these two cases will get \textquotedblleft lost in
the Gaussian noise,\textquotedblright\ which is of order $\sqrt{k}\gg10\log n$
if the constant $C$ was sufficiently large. \ By contrast, if we take an AND
of \textquotedblleft true\textquotedblright\ swap tests, then we accept with
probability $1$ in case (1), versus with probability $\left(  \frac{1}%
{2}\right)  ^{10\log n}=\frac{1}{n^{10}}$ in case (2).

But if we can't reliably distinguish these cases using LOCC, then certainly we
can't guess the parity, across all $k$ sub-blocks within a given block, with
bias more than $1/\exp\left(  k\right)  $ over random. \ (Whereas by contrast,
using \textquotedblleft true\textquotedblright\ swap tests, we can compute the
parity across the sub-blocks with success probability $1-\frac{1}{n^{O\left(
1\right)  }}$, given the promise that every sub-block satisfies either (1) or
(2) above.)

If so, then the end result is that, using LOCC measurements, we can't compute
the sum of the parities across the $n/k^{3}$ blocks even noisily, whereas
using true swap tests, we can.

\section{Appendix: On Composition of Quantum DP Algorithms\label{CANTCOMPOSE}}

One of the central properties of classical differential privacy is that it
nicely composes: that is, if we run an $\varepsilon_{1}$%
-DP\ algorithm\ followed by an $\varepsilon_{2}$-DP algorithm on the same
database, then the resulting algorithm is $\left(  \varepsilon_{1}%
+\varepsilon_{2}\right)  $-DP. \ Furthermore, \textit{advanced composition}
\cite{drv} shows that, with overwhelming probability, the loss in privacy when
we compose $k$ algorithms is even slower than linear, growing only like
$\sqrt{k}$.

This immediately raises a question: does quantum differential privacy
similarly compose? \ Here we face a new difficulty, not present in classical
case: namely, when we compose quantum DP algorithms, each algorithm will in
general damage our state. \ And this might cause not only a catastrophic loss
in privacy, but even a catastrophic loss in \textit{accuracy}.

Fortunately, we can use our connection between DP and gentleness address the
concern about accuracy, at least in the regime where that connection applies.
\ For certainly \textit{gentleness} composes. \ That is, if we apply an
$\alpha_{1}$-gentle measurement $M_{1}$\ followed by an $\alpha_{2}$-gentle
measurement $M_{2}$, then the result will be $\left(  \alpha_{1}+\alpha
_{2}\right)  $-gentle, by the triangle inequality for trace distance. \ And by
Corollary \ref{damagecor}, this is true even if $M_{2}$\ is guaranteed to be
gentle only on the \textit{original} state (for example, because it's a
product state), and not necessarily on the post-measurement states that result
from applying $M_{1}$. \ We even conjecture that an \textquotedblleft advanced
composition\textquotedblright\ property holds for gentleness (see Section
\ref{OPEN}).

Thus, suppose we want to compose product measurements $M_{1},\ldots,M_{k}%
,$\ that are each $\varepsilon$-DP on product states, for some $\varepsilon
\ll\frac{1}{k\sqrt{n}}$. \ Then by Theorem \ref{main}, these measurements are
each $O\left(  \varepsilon\sqrt{n}\right)  $-gentle. \ So we can compose them
while preserving good accuracy.

Even here, though, there's a potential issue with privacy. \ The issue arises
because the later $M_{i}$'s are applied not to our original state $\rho$, but
to damaged versions of the state. \ And particularly if this damage is
additive rather than multiplicative, we have no guarantee that the later
$M_{i}$'s will preserve DP (with respect to the \textit{original} state $\rho
$) when applied to the damaged versions. \ Indeed, ensuring privacy would
require saying at least \textit{something} about the post-measurement states.
\ If, for example, we implemented some $M_{i}$ in a way that gratuitously
\textquotedblleft amplified\textquotedblright\ the information in (say) the
first register, copying it into the other $n-1$\ registers as a byproduct of
the measurement procedure, then privacy need not be preserved when we apply
$M_{i+1}$. \ On the other hand, it seems plausible to us that $\varepsilon$-DP
measurements, for $\varepsilon\ll\frac{1}{\sqrt{n}}$, can always be
implemented in such a way that privacy is preserved under composition.

By using Lemma \ref{facepalmlemma}, the following proposition confirms that
quantum DP composition works at least in the special case where we're
composing a small number of quantum DP algorithms that are gentle, and all of
whose outputs have appreciably large probabilities on all states.

\begin{proposition}
[Limited Composition for Quantum DP]\label{composedp}Let $\mathcal{M}$\ be the
sequential composition of $k$ measurements $M_{1},\ldots,M_{k}$, where $M_{i}%
$\ is $\varepsilon_{i}$-DP on product states and $\alpha_{i}$-gentle\ on
product states. \ Suppose that for all product states $\rho$\ and all possible
sequences $Y=\left(  y_{1},\ldots,y_{k}\right)  $\ of measurement outcomes, we
have%
\[
\Pr\left[  M_{1}\left(  \rho\right)  \text{ outputs }y_{1}\right]  \cdots
\Pr\left[  M_{k}\left(  \rho\right)  \text{ outputs }y_{k}\right]  \geq p,
\]
where $p\gg\alpha_{1}+\cdots+\alpha_{k}$. \ Then $\mathcal{M}$ achieves a
relative accuracy of $\frac{\alpha_{1}+\cdots+\alpha_{k}}{p}$, in the sense
that%
\[
\left\vert \frac{\Pr\left[  \mathcal{M}\left(  \rho\right)  \text{ outputs
}Y\right]  }{\Pr\left[  M_{1}\left(  \rho\right)  \text{ outputs }%
y_{1}\right]  \cdots\Pr\left[  M_{k}\left(  \rho\right)  \text{ outputs }%
y_{k}\right]  }-1\right\vert \leq\frac{\alpha_{1}+\cdots+\alpha_{k}}{p}%
\]
for all product states $\rho$\ and all $Y$, and in addition is $\varepsilon
$-DP on product states for%
\begin{align*}
\varepsilon &  =\varepsilon_{1}+\cdots+\varepsilon_{k}+\ln\left(
\frac{p+\left(  \alpha_{1}+\cdots+\alpha_{k}\right)  }{p-\left(  \alpha
_{1}+\cdots+\alpha_{k}\right)  }\right) \\
&  =\varepsilon_{1}+\cdots+\varepsilon_{k}+O\left(  \frac{\alpha_{1}%
+\cdots+\alpha_{k}}{p}\right)  .
\end{align*}

\end{proposition}

\begin{proof}
The relative accuracy part follows immediately from the first part of Lemma
\ref{facepalmlemma}, which tells us that%
\[
\left\vert \Pr\left[  M_{1}\left(  \rho\right)  \text{ outputs }y_{1}\right]
\cdots\Pr\left[  M_{k}\left(  \rho\right)  \text{ outputs }y_{k}\right]
-\Pr\left[  \mathcal{M}\left(  \rho\right)  \text{ outputs }Y\right]
\right\vert \leq\alpha_{1}+\cdots+\alpha_{k}.
\]
For the $\varepsilon$-DP part, for all neighboring product states $\rho
,\sigma$\ and all $y_{1},\ldots,y_{k}$\ we have%
\begin{align*}
\Pr\left[  \mathcal{M}\left(  \rho\right)  \text{ outputs }Y\right]   &
\leq\left(  1-\frac{\alpha_{1}+\cdots+\alpha_{k}}{p}\right)  \Pr\left[
M_{1}\left(  \rho\right)  \text{ outputs }y_{1}\right]  \cdots\Pr\left[
M_{k}\left(  \rho\right)  \text{ outputs }y_{k}\right] \\
&  \leq\frac{e^{\varepsilon_{1}+\cdots+\varepsilon_{k}}}{1-\frac{\alpha
_{1}+\cdots+\alpha_{k}}{p}}\Pr\left[  M_{1}\left(  \sigma\right)  \text{
outputs }y_{1}\right]  \cdots\Pr\left[  M_{k}\left(  \sigma\right)  \text{
outputs }y_{k}\right] \\
&  \leq e^{\varepsilon_{1}+\cdots+\varepsilon_{k}}\frac{1+\frac{\alpha
_{1}+\cdots+\alpha_{k}}{p}}{1-\frac{\alpha_{1}+\cdots+\alpha_{k}}{p}}%
\Pr\left[  \mathcal{M}\left(  \sigma\right)  \text{ outputs }Y\right]  .
\end{align*}

\end{proof}

In Proposition \ref{composedp}, product states could have been replaced by any
other set of states that's closed under the neighbor relation. \ For the
special case of product states, though, we can combine Proposition
\ref{composedp}\ with part (2) of Theorem \ref{main} to get the following
corollary, which does not need the gentleness of the underlying measurements
as a separate assumption.

\begin{corollary}
\label{composecor}Let $\mathcal{M}$\ be the sequential composition of $k$
product measurements $M_{1},\ldots,M_{k}$, each on $n$ registers. \ Suppose
that each $M_{i}$\ is $\varepsilon_{i}$-DP on product states,\ where
$\varepsilon:=\varepsilon_{1}+\cdots+\varepsilon_{k}$ is at most $\frac
{1}{10\sqrt{n}}$. \ Suppose also that for all product states $\rho$, all
$i\in\left[  k\right]  $, and all measurement outcomes $y$, we have%
\[
\Pr\left[  M_{i}\left(  \rho\right)  \text{ outputs }y\right]  \geq p_{i},
\]
where $p:=p_{1}\cdots p_{k}$ is at least $10\varepsilon\sqrt{n}$. \ Then
$\mathcal{M}$\ achieves a relative accuracy of $O\left(  \frac{\varepsilon
\sqrt{n}}{p}\right)  $, in the sense that%
\[
\left\vert \frac{\Pr\left[  \mathcal{M}\left(  \rho\right)  \text{ outputs
}y_{1},\ldots,y_{k}\right]  }{\Pr\left[  M_{1}\left(  \rho\right)  \text{
outputs }y_{1}\right]  \cdots\Pr\left[  M_{k}\left(  \rho\right)  \text{
outputs }y_{k}\right]  }-1\right\vert =O\left(  \frac{\varepsilon\sqrt{n}}%
{p}\right)
\]
for all product states $\rho$\ and outcomes $y_{1},\ldots,y_{k}$, and in
addition is $O\left(  \frac{\varepsilon\sqrt{n}}{p}\right)  $-DP on product states.
\end{corollary}

In the remainder of this appendix, we will show that, when $\varepsilon$\ is
large compared to $\frac{1}{\sqrt{n}}$, so that we're outside the range where
DP implies gentleness, the composition of $\varepsilon$-DP\ measurements need
not even preserve \textit{accuracy}.

Recall the \textquotedblleft randomized response\textquotedblright\ algorithm
$R_{\beta}$\ from Section \ref{GDSEP}, which for each $i\in\left[  n\right]  $
independently, simply measures the $i^{th}$ qubit in the $\left\{  \left\vert
0\right\rangle ,\left\vert 1\right\rangle \right\}  $ basis, and returns the
measurement outcome with probability $\frac{1}{2}+\beta$, or its complement
with probability $\frac{1}{2}-\beta$. \ (Thus, the output of $R_{\beta}$\ is
an $n$-bit string.) \ We now give our example:

\begin{theorem}
[Failure of Composition for Quantum DP]\label{cantcompose}There exist
$n$-qubit measurements $M_{1}$\ and $M_{2}$\ that are individually
$\varepsilon$-DP on product states for $\varepsilon=O\left(  \frac{1}{n^{1/4}%
}\right)  $, but such that no implementation of $M_{1}$\ leaves us with a
post-measurement state allowing an accurate result to be returned if we later
run $M_{2}$ (even supposing that we don't condition on the outcome of $M_{1}$).
\end{theorem}

\begin{proof}
Let $M_{1}$\ be the randomized response algorithm $R_{\varepsilon}$, which is
$O\left(  \varepsilon\right)  $-DP by Proposition \ref{rbeta}. \ Also, let
$M_{2}$ be the variant of the $L_{\sigma}$\ mechanism from before, but in the
$\left\{  \left\vert +\right\rangle ,\left\vert -\right\rangle \right\}
$\ basis. \ In other words, $M_{2}$\ returns the number of $\left\vert
+\right\rangle $'s plus Laplace noise of mean $\sigma$. \ We've seen that
$L_{\sigma}$, and hence $M_{2}$, is $\frac{1}{\sigma}$-DP.

Now suppose that in reality, each qubit is either $\left\vert +\right\rangle $
or $\left\vert -\right\rangle $. \ Then by a straightforward calculation,
$M_{1}$ damages each qubit by $\Theta\left(  \varepsilon^{2}\right)  $ in
trace distance, even if we average over both possible measurement outcomes, by
decreasing the magnitudes of the off-diagonal density matrix entries by
$\Theta\left(  \varepsilon^{2}\right)  $. \ In more detail, the effect is
simply: every $\left\vert +\right\rangle $ qubit flips to $\left\vert
-\right\rangle $, and every $\left\vert -\right\rangle $ qubit flips to
$\left\vert +\right\rangle $, with independent probability $\Theta\left(
\varepsilon^{2}\right)  $.

So now consider what happens when we run $M_{2}$. \ If we have an $n$-bit
string $x$, and every bit of $x$ gets flipped with independent probability
$\delta$, then from the corrupted string $x^{\prime}$, we can estimate the
Hamming weight of the original string $x$ to within an additive error of
$\Theta\left(  \sqrt{n\delta\left(  1-\delta\right)  }\right)  $. \ In our
case, $\delta=\Theta\left(  \varepsilon^{2}\right)  $, so this additive error
is $\Theta\left(  \varepsilon\sqrt{n}\right)  $.

But recall that $M_{2}$ needs to estimate the number $\left\vert
+\right\rangle $ qubits to within an additive error of $\Theta\left(
1/\varepsilon\right)  $. \ If $\frac{1}{\varepsilon}\ll\varepsilon\sqrt{n}$,
or equivalently $\varepsilon\gg\frac{1}{n^{1/4}}$, then this is impossible.
\end{proof}

Of course, the above leaves many further questions that one could explore: for
example, what happens for $\varepsilon$ in the range between $\frac{1}%
{\sqrt{n}}$\ and $\frac{1}{n^{1/4}}$? \ Also, what if we restrict our
attention to quantum DP algorithms with only a few possible outcomes (thus
ruling out randomized response applied to each qubit separately)? \ Finally,
what if we allow our \textquotedblleft composed\textquotedblright\ algorithm
to do anything it likes to obtain the desired information, including violating
the specified order (e.g., applying $L_{\sigma}^{+}$ before $R_{\varepsilon}$)
and even more radical changes?

\end{document}